\documentclass[reqno]{amsart}
\usepackage{amsmath}
\usepackage{amssymb}
\usepackage{amsthm}
\usepackage{amsfonts}
\usepackage{mathtools}
\usepackage{bm}
\usepackage{ulem}
\usepackage{fancyhdr}
\usepackage{hyperref}
\usepackage{framed}
\usepackage{float}
\usepackage[pdftex,final]{graphicx}
\usepackage{tikz}
\usepackage{comment}
\usetikzlibrary{calc,arrows}

\usepackage[iso-8859-7]{inputenc}

\theoremstyle{plain}
\newtheorem{thm}{Theorem}[section]
\newtheorem{corollary}[thm]{Corollary}
\newtheorem{lemma}[thm]{Lemma}
\newtheorem{prop}[thm]{Proposition}
\newtheorem*{theorem*}{Theorem}
\theoremstyle{definition}
\newtheorem{rem}[thm]{Remark}
\newtheorem{dfn}[thm]{Definition}

\numberwithin{equation}{section}

\newcommand{\mytilde}{\raise.17ex\hbox{$\scriptstyle\mathtt{\sim}$}}

\newcommand{\I}{\ensuremath{\mathcal{I}}}
\renewcommand{\S}{\ensuremath{\mathcal{S}}}
\renewcommand{\P}{\ensuremath{\mathcal{P}}}
\newcommand{\N}{\ensuremath{\mathcal{N}}}
\newcommand{\U}{\ensuremath{\mathcal{U}}}

\newcommand{\RgeO}{\ensuremath{\mathbb{R}_{ \geq 0}}}
\newcommand{\Rat}[1]{\ensuremath{\mathbb{R}^{#1}}}

\renewcommand{\bf}[1]{\textbf{#1}}

\begin{document}

\author{D. Boskos}
\address{Department of Automatic Control, School of Electrical Engineering, KTH Royal Institute of Technology, Osquldas v\"ag 10, 10044, Stockholm, Sweden}
\email{boskos@kth.se}

\author{D. V. Dimarogonas}
\address{Department of Automatic Control, School of Electrical Engineering, KTH Royal Institute of Technology, Osquldas v\"ag 10, 10044, Stockholm, Sweden}
\email{dimos@kth.se}

\begin{abstract}
The goal of this report is to define abstractions for multi-agent systems with feedback interconnection in their dynamics. In the proposed decentralized framework,  we specify a finite or countable transition system for each agent which only takes into account the discrete positions of its neighbors. The dynamics of the considered systems consist of two components. An appropriate feedback law which guarantees that certain system and network requirements are fulfilled and induces coupled constraints, and additional free inputs which we exploit in order to accomplish high level tasks. In this work, we provide sufficient conditions on the space and time discretization for the abstraction of the system's behaviour which ensure that we can extract a well posed and hence meaningful transition system. Furthermore, these conditions include design parameters whose tuning provides the possibility for multiple transitions, and hence, enable the construction of transition systems with motion planning capabilities.
\end{abstract}

\title{Decentralized Abstractions For Multi-Agent Systems Under Coupled Constraints}
\maketitle

\section{Introduction}

High level task planning for multi-agent systems constitutes an active area of research which lies in the interface between computer science and modern control theory. A challenge in this new interdisciplinary direction constitutes of the problem of defining appropriate abstractions for continuous time multi-agent control systems, which can be used as a tool for the analysis and control of large scale systems and the synthesis of high level plans \cite{LsKk04}, \cite{KmBc08}. Robot motion planning and control constitutes a central field where this line of work is applied \cite{FgGaKGhPg09}, \cite{CyDxSaBc12}. In particular, the use of a suitable discrete system's model allows the automatic synthesis of discrete plans that guarantee satisfaction of the high level specifications. Then, under appropriate relations between the continuous system and its discrete analogue, these plans can be converted to low level primitives such as sequences of feedback controllers, and hence, enable the continuous system to implement the corresponding tasks.
Such tasks in the case of multiple mobile robots in an industrial workspace could include for example the following scenario. Robot 1 should periodically visit regions $A$, $B$, while avoiding $C$ and after collecting an item of type $X$ from robot 2 at location $D$, store it at location $E$.

In order to synthesize high level plans, it is required to specify an abstraction of the original system, namely a system that preserves some properties of interest of the initial system, while ignoring detail. Results in this direction for the nonlinear single plant case have been obtained in the papers \cite{PgGaTp08} and \cite{ZmPgMmTp12}, where the notions of approximate bisimulation and simulation are exploited for certain classes of nonlinear systems, under appropriate stability, and completeness assumptions, respectively. The notion of bisimulation, has its origin in computer science \cite{BcKjp08}, and guarantees that if the initial system and its abstraction are bisimilar, then the task of checking feasibility of high level plans for the original system reduces to the same task for its abstraction and vice versa. Bisimulation relations between transition system models of discrete or continuous time linear control systems with finite affine observation maps were explicitly characterized and constructed in \cite{Pg03}, providing also a generalization of the notion of state space equivalence between continuous time systems \cite{Va04}.

Another abstraction tool for a general class of systems is the hybridization approach \cite{AeDtGa07}, where the behaviour of a nonlinear system is captured by means of a piecewise affine hybrid system on simplices. Motion planing techniques for the latter case have been developed in the recent works \cite{GaMs08}, \cite{GaMs12}, which are also based on the abstraction and controller synthesis framework provided in \cite{HlVj01}, \cite{HlCpVj06}, and further studied in \cite{BmGm14}. Other abstraction techniques for nonlinear systems include \cite{Rg11}, where discrete time systems are studied in a behavioral framework, the sign based abstraction methodology introduced in \cite{Ta08}, which is based on Lie-algebraic type conditions and \cite{AaTaSa09}, where box abstractions are studied for polynomial and other classes of systems (for a literature survey on the subject see also the monograph \cite{Tp09}). It is also noted that certain of the aforementioned approaches have been extended to switched \cite{GaPgTp10}, \cite{GeDxLmBc14} and networked control systems \cite{ZmMmAa14}. Furthermore, abstractions for the case of discrete time interconnected systems that are described by coupled difference equations, can be found in \cite{TyIj08} and \cite{PgPpDm14}, for stabilizable linear systems, and incrementally input-to-state stable nonlinear systems, respectively. Finally, we note that the control design which we adopt for the construction of the symbolic models is in part related to the notion of In-Block Controllability \cite{CpWy95}, \cite{HmCp14b}. 

In particular, we focus on multi-agent systems and assume that the agents' dynamics consist of feedback interconnection terms, which ensure that certain system properties as for instance connectivity or (and) invariance are preserved, and free input terms, which provide the ability for motion planning under coupled constraints. 
In this report, we generalize the corresponding results of our recent work \cite{BdDd15a}, where sufficient conditions for well posed abstractions of the multi-agent system are provided for the case where the agents' workspace is $\Rat{n}$. A well posed abstraction refers to the property that the discrete state transition system which serves as an abstract model of the multi-agent system has at least one outgoing transition from each state. The extension in this work is twofold. First, the results on admissible space and time discretizations in \cite{BdDd15a} which ensure well posed abstractions, are now valid when the agents' workspace is a general domain $D$ of $\Rat{n}$, provided that $D$ is invariant for the dynamics of the system. Also, the corresponding framework is extended for motion planning, and sufficient conditions are provided which guarantee that each agent can perform multiple transitions from each initial discrete state.

The rest of the report  is organized as follows. Basic notation and preliminaries are introduced in Section 2. Section 3 is devoted to the formulation of the problem and motivates the control design that will be utilized for the derivation of the symbolic models. In Section 4, we define well posed abstractions for single integrator multi-agent systems by means of hybrid feedback controllers and prove that the latter provide solutions consistent with the design requirement on the systems' free inputs. Section 5 is devoted to specific properties of the control laws that realize the transitions of the proposed discrete system's model. In Section 6 we quantify space and time discretizations which guarantee well posed transitions with motion planning capabilities. The framework is illustrated through an example in Section 7 including simulation results.  Finally, we conclude and indicate directions of further research in Section 8.

\section{Preliminaries and Notation}

We use the notation $|x|$ for the Euclidean norm of a vector $x\in\Rat{n}$. For a subset $S$ of $\Rat{n}$, we denote by ${\rm cl}(S)$, ${\rm int}(S)$ and $\partial S$ its closure, interior and boundary, respectively, where $\partial S:={\rm cl}(S)\setminus{\rm int}(S)$. Given $R>0$ and $x\in\Rat{n}$, we denote by $B(R)$ the closed ball with center $0\in\Rat{n}$ and radius $R$, namely $B(R):=\{x\in\Rat{n}:|x|\le R \}$ and $B(x;R):=\{y\in\Rat{n}:|x-y|\le R \}$. For two nonempty sets $A,B\subset\Rat{n}$ their Minkowski sum is given as $A+B:=\{x+y\in\Rat{n}:x\in A, y\in B\}$. Also, for a nonempty set $A\subset\Rat{n}$ its diameter is defined as ${\rm diam}(A):=\sup\{|x-y|:x,y\in A\}$. Given a real number $a\in\RgeO$ we denote by $\lfloor a\rfloor$ the integer part of $a$ ($\lfloor a\rfloor:=\max\{n\in\mathbb{N}\cup\{0\}:n\le a\}$). Finally, given a function $f:X\to Y$ and a subset $W$ of $X$, the notation $f|_W$ is used for the restriction of $f$ to $W$.

Consider a multi-agent system with $N$ agents. For each agent $i\in\N:=\{1,\ldots,N\}$ we use the notation $\N_i$ for the set of its neighbors and $N_i:=|\N_i|$ for its cardinality. We also consider an ordering of the agent's neighbors which is denoted by $j_1,\ldots,j_{N_i}$, and define the $N_i$-tuple $ j(i)=(j_1,\ldots,j_{N_i})$. Whenever it is clear from the context, the argument $i$ in the latter notation will be omitted. Given an index set $\I$ and an agent $i\in\N$ with neighbors $j_1,\ldots,j_{N_i}\in\N$, define the mapping ${\rm pr}_i:\I^N\to\I^{N_i+1}$ which assigns to each $N$-tuple $(l_1,\ldots,l_N)\in\I^N$ the $N_i+1$-tuple $(l_i,l_{j_1},\ldots,l_{j_{N_i}})\in\I^{N_i+1}$, i.e., the indices of agent $i$ and its neighbors.

We proceed by providing a formal definition for the notion of a transition system (see for instance \cite{BcKjp08}, \cite{Pg03}, \cite{PgGaTp08}).

\begin{dfn}
A transition system is a tuple $TS:=(Q,Act,\longrightarrow)$, where:

\textbullet\; $Q$ is a set of states.

\textbullet\; $Act$ is a set of actions.

\textbullet\; $\longrightarrow$ is a transition relation with $\longrightarrow\subset Q\times Act\times Q$.

\noindent The transition system is said to be finite, if $Q$ and $Act$ are finite sets. We also use the (standard) notation $q\overset{a}{\longrightarrow} q'$ to denote an element $(q,a,q')\in\longrightarrow$. For every $q\in Q$ and $a\in Act$ we use the notation ${\rm Post}(q;a):=\{q'\in Q:(q,a,q')\in\longrightarrow\}$. The transition system is called deterministic if for each $q\in Q$ and $a\in Act$, $q\overset{a}{\longrightarrow} q'$ and $q\overset{a}{\longrightarrow} q''$ implies that $q'=q''$.
\end{dfn}


\section{Problem Formulation}

We focus on multi-agent systems with single integrator dynamics
\begin{equation}\label{single:integrator}
\dot{x}_{i}=u_{i},x_{i}\in\Rat{n},i\in\N
\end{equation}

\noindent and consider as inputs decentralized control laws of the form
\begin{equation}\label{general:feedback:law}
u_{i}=f_{i}(x_{i},\bf{x}_j)+v_{i}, i\in\N,
\end{equation}

\noindent with $\bf{x}_j(=\bf{x}_{j(i)}):=(x_{j_{1}},\ldots,x_{j_{N_i}})\in\Rat{N_i n}$ (see Section 2 for the Notation $j(i)$), consisting of two terms: a feedback term $f_{i}(\cdot)$ which depends on the states of $i$ and its neighbors, and an extra input term $v_{i}$, which we call free input. We assume that for each $i\in\N$ it holds $x_{i}\in D$ where $D$ is a domain of $\Rat{n}$ and that each $f_{i}(\cdot)$ is locally Lipschitz. We also assume that the feedback terms $f_{i}(\cdot)$ are globally bounded, namely, there exists a constant $M>0$ such that
\begin{equation} \label{dynamics:bound}
|f_{i}(x_{i},\bf{x}_j)|\le M, \forall (x_{i},\bf{x}_j)\in D^{N_i+1}.
\end{equation}

\noindent Furthermore, we consider piecewise continuous free inputs $v_{i}$ that satisfy the bound
\begin{equation}\label{input:bound}
|v_{i}(t)|\le v_{\max},\forall t\ge 0, i\in\N.
\end{equation}

In the subsequent analysis, it is assumed that the maximum magnitude of the feedback terms is higher than that of the free inputs, namely, that
\begin{equation} \label{vmax:vs:M}
v_{\max}<M.
\end{equation}

\noindent This assumption is motivated by the fact that we are primarily interested in maintaining the property that the feedback is designed for, and secondarily, in exploiting the free inputs in order to accomplish high level tasks. A class of multi-agent systems of the form \eqref{single:integrator}-\eqref{general:feedback:law} which justifies this assumption has been studied in our companion work \cite{BdDd15}. In particular, sufficient conditions are provided, which guarantee both connectivity of the network and forward invariance of the system's trajectories inside a given bounded domain, for an appropriate selection of $v_{\max}$ in \eqref{input:bound} which necessitates $v_{\max}$ to satisfy \eqref{vmax:vs:M}. The latter forward invariance property is defined in the Invariance Assumption (IA) below, which we assume that the multi-agent system \eqref{single:integrator}-\eqref{general:feedback:law} satisfies for the rest of the report.

\noindent \textbf{(IA)} For every initial condition $x(0)\in D^{N}$ and any piecewise continuous input $v=(v_1,\ldots,v_n):\RgeO\to\Rat{Nn}$ satisfying \eqref{input:bound}, the (unique) solution of the system \eqref{single:integrator}-\eqref{general:feedback:law} is defined and remains in $D^{N}$ for all $t\ge 0$. $\triangleleft$

\noindent This assumption does not restrict the class of systems under consideration since it is satisfied for any forward complete system when $D=\Rat{n}$. Recall that the system \eqref{single:integrator}-\eqref{general:feedback:law} is forward complete (see e.g, \cite{AdSe99}) if for each initial condition in $\Rat{n}$ and each measurable locally essentially bounded input $v=(v_1,\ldots,v_n):\RgeO\to\Rat{Nn}$, its solution exists for all positive times. Also, notice that due to the above bounds on the dynamics and the free input terms, the system \eqref{single:integrator}-\eqref{general:feedback:law} is forward complete. Finally, when $D$ is bounded, as is the case in \cite{BdDd15}, a finite partition of the workspace by bounded sets can lead to a finite transition system which captures the properties of interest of the multi-agent system and hence enables the investigation for computable solutions with respect to high level specifications.

In what follows, we consider a cell decomposition of the state space $D$ (which can be regarded as a partition of $D$) and a time step $\delta t>0$. We will refer to this selection as a space and time discretization. For the definition of a cell decomposition we adopt a modification of the corresponding definition from \cite[p 129-called cell covering]{Gl02}.

\begin{dfn} \label{cell:decomposition}
Let $D$ be a domain of $\Rat{n}$. A \bf{cell decomposition} $\mathcal{S}=\{S_{l}\}_{l\in\mathcal{I}}$ of $D$, where $\mathcal{I}$ is a finite or countable index set, is a family of nonempty connected sets $S_{l}$, $l\in\mathcal{I}$, such that $\sup\{{\rm diam}(S_l),l\in\I\}<\infty$, ${\rm cl}({\rm int}(S_l))=S_l$ for all $l\in\I$, ${\rm int}(S_{l})\cap {\rm int}(S_{\hat{l}})=\emptyset$ for all $l\ne\hat{l}$ and $\cup_{l\in\mathcal{I}} S_{l}=D$. $\triangleleft$
\end{dfn}

\noindent Given a cell decomposition $\S :=\{S_l\}_{l\in\I}$ of $D$, we use the notation $\bf{l}_i=(l_i,l_{j_1},\ldots,$ $l_{j_{N_i}})\in\I^{N_i+1}$ to denote the indices of the cells where agent $i$ and its neighbors belong at a certain time instant and call it the cell configuration of agent $i$. Similarly, we use the notation $\bf{l}=(l_{1},\ldots,l_{N})\in\mathcal{I}^{N}$ to specify the indices of the cells where all the $N$ agents belong at a given time instant and call it the cell configuration (of all agents). Thus, given a cell configuration $\bf{l}$, it is possible to determine the cell configuration of agent $i$ as $\bf{l}_i={\rm pr}_i(\bf{l})$ (see Section 2 for the definition of ${\rm pr}_i(\cdot)$).

Through the space and time discretization we aim at capturing reachability properties of the original continuous time system, by means of a discrete state transition system. Informally, we would like to consider for each agent $i$, its individual transition system with state set the cells of the state partition, actions defined to be all possible cells of its neighbors, and transition relation specified as follows. Given the initial cells of agent $i$ and its neighbors, it is possible for $i$ to perform a transition to a final cell, \textit{if for all states in its initial cell there exists a free input, such that its trajectory will reach the final cell at time $\delta t$, for all possible initial states of its neighbors in their cells, and their corresponding free inputs}. Feasibility of high level plans requires the corresponding system to be well posed (meaningful), which implies that for each initial cell it is possible to transit to (at least) one final cell.

We next illustrate the concept of a well posed space-time discretization, namely, a discretization which generates for each agent a meaningful transition system in accordance with the discussion above. Consider a cell decomposition as depicted in Fig. 1 and a time step $\delta t$. The tips of the arrows in the figure are the endpoints of agent's $i$ trajectories at time $\delta t$. In both cases in the figure we focus on agent $i$ and consider the same cell configuration for $i$ and its neighbors. However, we consider different dynamics for Cases (i) and (ii). In Case (i), we observe that for the three distinct initial positions in cell $S_{l_i}$, it is possible to drive agent $i$ to cell $S_{l_i'}$ at time $\delta t$. We assume that this is possible for all initial conditions in this cell and irrespectively of the initial conditions of $i$'s neighbors in their cells and the inputs they choose. We also assume that this property holds for all possible cell configurations of $i$ and for all the agents of the system. Thus we have a well posed discretization for system (i). On the other hand, for the same cell configuration and system (ii), we observe the following. For three distinct initial conditions of $i$ the corresponding reachable sets at $\delta t$, which are enclosed in the dashed circles, lie in different cells. Thus, it is not possible given this cell configuration of $i$ to find a cell in the decomposition which is reachable from every point in the initial cell and we conclude that discretization is not well posed for system (ii).

\begin{figure}[H]
\begin{center}
\begin{tikzpicture} [scale=.85]

\draw[color=gray,thick] (0,0) -- (4,0);
\draw[color=gray,thick] (0,1) -- (4,1);
\draw[color=gray,thick] (0,2) -- (4,2);
\draw[color=gray,thick] (0,3) -- (4,3);

\draw[color=gray,thick] (0,0) -- (0,3);
\draw[color=gray,thick] (1,0) -- (1,3);
\draw[color=gray,thick] (2,0) -- (2,3);
\draw[color=gray,thick] (3,0) -- (3,3);
\draw[color=gray,thick] (4,0) -- (4,3);

\draw[color=blue,very thick] (0,0) -- (1,0) -- (1,1) -- (0,1) -- (0,0);
\draw[color=green,very thick] (2,1) -- (3,1) -- (3,2) -- (2,2) -- (2,1);
\draw[color=green,very thick] (3,2) -- (4,2) -- (4,3) -- (3,3) -- (3,2);
\draw[color=cyan,very thick] (0,2) -- (1,2) -- (1,3) -- (0,3) -- (0,2);


\fill[yellow] (0.3,2.5) circle (0.3cm);
\fill[yellow] (0.4,2.7) circle (0.3cm);
\fill[yellow] (0.5,2.6) circle (0.3cm);

\draw[black, dashed] (0.3,2.5) circle (0.3cm);
\draw[black, dashed] (0.4,2.7) circle (0.3cm);
\draw[black, dashed] (0.5,2.6) circle (0.3cm);

\draw [color=red,thick,->,>=stealth'](0.2,0.3) .. controls (0.3,1.5) .. (0.3,2.5);
\fill[black] (0.2,0.3) circle (1.5pt);

\draw [color=red,thick,->,>=stealth'](0.5,0.8) .. controls (0.5,1.5) .. (0.4,2.7);
\fill[black] (0.5,0.8) circle (1.5pt);

\draw [color=red,thick,->,>=stealth'](0.7,0.6) .. controls (0.7,1.5) .. (0.5,2.6);
\fill[black] (0.7,0.6) circle (1.5pt);

\coordinate [label=left:$S_{l_i}$] (A) at (0,0.5);
\coordinate [label=above left:$x_{i}$] (A) at (0.85,0.15);
\coordinate [label=right:$S_{l_{j_1}}$] (A) at (3,1.5);
\fill[black] (2.3,1.4) circle (1.5pt) node[right]{$x_{j_{1}}$};
\coordinate [label=right:$S_{l_{j_2}}$] (A) at (4,2.5);
\fill[black] (3.4,2.8) circle (1.5pt) node[below]{$x_{j_{2}}$};
\coordinate [label=left:$S_{l_i'}$] (A) at (0,2.5);

\draw[color=gray,thick] (5,0) -- (9,0);
\draw[color=gray,thick] (5,1) -- (9,1);
\draw[color=gray,thick] (5,2) -- (9,2);
\draw[color=gray,thick] (5,3) -- (9,3);

\draw[color=gray,thick] (5,0) -- (5,3);
\draw[color=gray,thick] (6,0) -- (6,3);
\draw[color=gray,thick] (7,0) -- (7,3);
\draw[color=gray,thick] (8,0) -- (8,3);
\draw[color=gray,thick] (9,0) -- (9,3);

\draw[color=blue,very thick] (5,0) -- (6,0) -- (6,1) -- (5,1) -- (5,0);
\draw[color=green,very thick] (7,1) -- (8,1) -- (8,2) -- (7,2) -- (7,1);
\draw[color=green,very thick] (8,2) -- (9,2) -- (9,3) -- (8,3) -- (8,2);
\draw[color=cyan,dashed,very thick] (5,2) -- (6,2) -- (6,3) -- (5,3) -- (5,2);
\draw[color=cyan,dashed,very thick] (6,2) -- (7,2) -- (7,3) -- (6,3) -- (6,2);
\draw[color=cyan,dashed,very thick] (6,1) -- (7,1) -- (7,2) -- (6,2) -- (6,1);

\fill[yellow] (5.3,2.5) circle (0.3cm);
\fill[yellow] (6.4,2.7) circle (0.3cm);
\fill[yellow] (6.5,1.7) circle (0.3cm);

\draw[black, dashed] (5.3,2.5) circle (0.3cm);
\draw[black, dashed] (6.4,2.7) circle (0.3cm);
\draw[black, dashed] (6.5,1.7) circle (0.3cm);

\draw [color=red,thick,->,>=stealth'](5.2,0.3) .. controls (5.3,1.5) .. (5.3,2.5);
\fill[black] (5.2,0.3) circle (1.5pt);

\draw [color=red,thick,->,>=stealth'](5.5,0.8) .. controls (5.5,1.5) .. (6.4,2.7);
\fill[black] (5.5,0.8) circle (1.5pt);

\draw [color=red,thick,->,>=stealth'](5.7,0.6) .. controls (5.7,1) .. (6.5,1.7);
\fill[black] (5.7,0.6) circle (1.5pt);

\coordinate [label=below:System (i):] (A) at (2,-0.5);
\coordinate [label=below: $\dot{x}_{i}\text{$=$}f_{i,({\rm i})}(x_{i}\text{$,$}x_{j_{1}}\text{$,$}x_{j_{2}})+v_{i,({\rm i})}$] (A) at (2,-1);
\coordinate [label=below:System (ii):] (A) at (7,-0.5);
\coordinate [label=below: $\dot{x}_{i}\text{$=$}f_{i,({\rm ii})}(x_{i}\text{$,$}x_{j_{1}}\text{$,$}x_{j_{2}})+v_{i,({\rm ii})}$] (A) at (7,-1);

\coordinate [label=left:$S_{l_i}$] (A) at (5,0.5);
\coordinate [label=above left:$x_{i}$] (A) at (5.85,0.15);
\coordinate [label=right:$S_{l_{j_1}}$] (A) at (8,1.5);
\fill[black] (7.3,1.4) circle (1.5pt) node[right]{$x_{j_{1}}$};
\coordinate [label=right:$S_{l_{j_2}}$] (A) at (9,2.5);
\fill[black] (8.4,2.8) circle (1.5pt) node[below]{$x_{j_{2}}$};

\draw[dashed,->,>=stealth] (0.4,2.7) -- (1.5,3.5);
\coordinate [label=right:$x_{i}(\delta t)$] (A) at (1.5,3.5);

\draw[dashed,->,>=stealth] (5.3,2.5) -- (6.5,3.5);
\coordinate [label=right:$x_{i}(\delta t)$] (A) at (6.5,3.5);

\end{tikzpicture}
\caption{Illustration of a space-time discretization which is well posed for system (i) but non-well posed for system (ii).}
\end{center}
\end{figure}
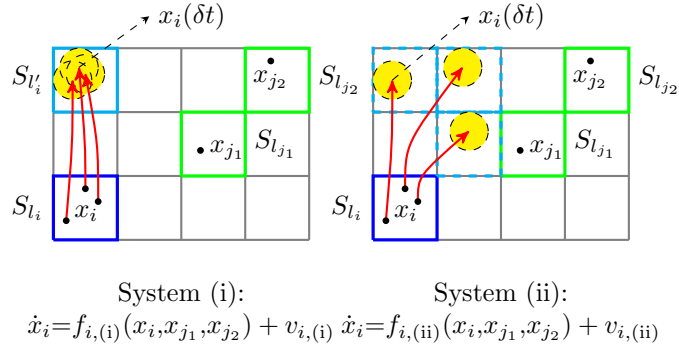

One main challenge in the attempt to provide meaningful decentralized abstractions for system \eqref{single:integrator}-\eqref{general:feedback:law} is the interconnection between the agents through the $f_{i}(\cdot)$ terms. The latter in conjunction with the considerations above, motivates the design of appropriate hybrid feedback laws in place of the $v_{i}$'s which will guarantee the desired well posed transitions.

We next define the particular feedback laws that are utilized for the construction of the symbolic models in this report, which also motivate the notion of well posed discretizations that will be formulated in the next section. Consider  a cell decomposition $\mathcal{S}=\{S_{l}\}_{l\in\mathcal{I}}$ of $D$ and a time step $\delta t$. For each agent $i\in\N$ and cell configuration $\bf{l}_i=(l_i,l_{j_1},\ldots,l_{j_{N_i}})\in\I^{N_i+1}$ of $i$ let 
\begin{equation} \label{reference:point}
(x_{i,G},\bf{x}_{j,G})\in S_{l_{i}}\times (S_{l_{j_1}}\times\cdots\times S_{l_{j_{N_i}}})
\end{equation}

\noindent be an arbitrary $N_i+1$-tuple of reference points and define the family of feedback laws  $k_{i,\bf{l}_i}:[0,T(x_{i0},w_i))\times D^{N_i+1}\to\Rat{n}$ parameterized by $x_{i0}\in S_{l_i}$ and $w_i\in W$ as
\begin{align}
k_{i,\bf{l}_i}(t,x_i,\bf{x}_j;x_{i0},w_i): & =k_{i,\bf{l}_i,1}(x_i,\bf{x}_j) \nonumber \\
&+k_{i,\bf{l}_i,2}(x_{i0})+k_{i,\bf{l}_i,3}(t;x_{i0},w_i), \label{feedback:ki}
\end{align}

\noindent where
\begin{equation} \label{set:W}
W:=B(\lambda v_{\max}),\lambda\in (0,1),
\end{equation}

\noindent and
\begin{align}
&k_{i,\bf{l}_i,1}(x_{i},\bf{x}_j):=f_i(x_i,\bf{x}_{j,G})-f_i(x_i,\bf{x}_j), \label{feedback:ki1} \\
&k_{i,\bf{l}_i,2}(x_{i0}):=\frac{1}{\delta t}(x_{i,G}-x_{i0}), \label{feedback:ki2} \\
& k_{i,\bf{l}_i,3}(t;x_{i0},w_i):=F_{i,\bf{l}_i}(\chi_{i}(t))+w_i-F_{i,\bf{l}_i}(\chi_i(t)+tw_i  \nonumber \\
&\qquad\qquad\qquad\qquad\qquad+(1-\tfrac{t}{\delta t})(x_{i0}-x_{i,G})),\label{feedback:ki3:plan} \\
& T(x_{i0},w_i):=\sup\{\bar{t}\in [0,T_{\max}):\chi_i(t)+tw_i \nonumber \\
& +\left(1-\tfrac{t}{\delta t}\right)(x_{i0}-x_{i,G})\in D,\forall t\in [0,\bar{t}]\}, \label{time:T:plan} \\
& t\in [0,T(x_{i0},w_i)),(x_i,\bf{x}_j)\in D^{N_i+1},x_{i0}\in S_{l_i},w_i\in W. \nonumber
\end{align}

\noindent The function $F_{i,\bf{l}_i}(\cdot)$ in \eqref{feedback:ki3:plan} is defined as
\begin{equation}\label{averaged:dynamics}
F_{i,\bf{l}_i}(x_i):=f_i(x_i,\bf{x}_{j,G}), x_i\in D
\end{equation}

\noindent and $\chi_i(\cdot)$ in \eqref{feedback:ki3:plan}, \eqref{time:T:plan} is the solution of the initial value problem
\begin{equation}\label{reference:solution}
\dot{\chi}_i=F_{i,\bf{l}_i}(\chi_i),\chi_i(0)=x_{i,G},
\end{equation}

\noindent which is defined and remains in $D$ on the maximal right interval $[0,T_{\max})$ ($T_{\max}$ is the same as in \eqref{time:T:plan}). In particular, $\chi_i(\cdot)$ constitutes a reference trajectory, whose endpoint agent $i$ should reach at time $\delta t$, when the agent's initial condition lies in $S_{l_i}$ and the feedback $k_{i,\bf{l}_i}(\cdot)$ above is applied (plus some extra hypotheses for the rest of the agents). The time $T(x_{i0},w_i)$ in \eqref{time:T:plan} stands for the right endpoint of the maximal right interval for which a modification of the reference trajectory that depends on $x_{i0}$ and $w_i$ is guaranteed to remain inside the domain $D$. Also, the parameter $\lambda$ in \eqref{set:W} stands for the portion of the free input that can be exploited for motion planning. In particular, each vector $w_i$ from the set $W$ in \eqref{set:W} provides a possible ``constant velocity" of a motion that we superpose to the reference trajectory $\chi_i(\cdot)$ of agent $i$, allowing thus the agent to reach all points inside a ball with center the position of the reference trajectory at time $\delta t$. Note that the control laws $k_{i,\bf{l}_i}(\cdot)$ are decentralized, since they only use information of agent $i$'s neighbors states. In addition, they depend on the cell configuration $\bf{l}_i$ through the reference points $(x_{i,G},\bf{x}_{j,G})$ which are involved in \eqref{feedback:ki1}-\eqref{feedback:ki3:plan}.

We proceed by providing some extra intuition for the selection of the control laws in \eqref{feedback:ki1}, \eqref{feedback:ki2} and \eqref{feedback:ki3:plan}, based on Fig. 2 below. Consider a cell decomposition of $D$, a time step $\delta t$ and select an agent $i$, a cell configuration of $i$ and a tuple of reference points as in \eqref{reference:point}. 
The reference trajectory of $i$ is obtained from \eqref{reference:solution}, by ``freezing" agent $i$'s neighbors at their corresponding reference points through the feedback term $k_{i,\bf{l}_i,1}(\cdot)$. Also, by selecting a vector $w_i$ in $W$ and informally assuming that we can superpose to the reference trajectory the motion of $i$ with constant speed $w_i$, namely, move along the curve $\bar{x}_i(\cdot)$ defined as
\begin{equation} \label{bar:xi:plan} 
\bar{x}_i(t):=\chi_i(t)+tw_i, t\in [0,T_{\max}),
\end{equation}

\noindent  we can reach the point $x$ inside the dashed ball at time $\delta t$ from the reference point $x_{i,G}$, as depicted in Fig. 2. In a similar way, it is possible to reach any point inside this ball by a different selection of $w_i$.
This ball has radius
\begin{equation} \label{distance:r}
r:=\lambda v_{\max}\delta t,
\end{equation}

\noindent namely, the distance that the agent can cross in  time $\delta t$ by exploiting the part  of the free input that is available for planning. \textit{Our abstraction requirement is that the transition to each cell which has nonempty intersection with $B(\chi_i(\delta t);r)$ is well posed.} This will be verified for system  \eqref{single:integrator}-\eqref{general:feedback:law} through the establishment of condition \eqref{planning:condition} in Theorem \ref{discretizations:for:planning} (Section 6) for appropriate space-time discretizations. For instance, in order for the transition to the cell where the point $x$ belongs to be well posed, we will require the following. That the feedback law  $k_{i,\bf{l}_i}(\cdot)$, which is selected in place of agent's $i$ free input, guarantees that for each initial condition of $i$ in cell $S_{l_i}$, the endpoint of $i$'s trajectory will coincide with the endpoint of the curve $\bar{x}_i(\cdot)$. In order to compensate for the deviation of the initial state with respect to the reference point and reach the point $x$, we use the extra terms $k_{i,\bf{l}_i,2}(\cdot)$ and $k_{i,\bf{l}_i,3}(\cdot)$. These enforce the agent to move with the velocity of the reference trajectory plus two constant velocity terms, one analogous to the displacement between the agents initial state and the reference point, and the other analogous to the distance between $x$ and the endpoint of $\chi_i(\cdot)$.

\begin{figure}[H]\label{fig2}
\begin{center}
\begin{tikzpicture}[scale=0.8]

\draw[dashed, color=red,thick] (3,0.5) circle (1cm);

\draw[color=gray] (0.1,0.8) -- (0.2,1.4) -- (1,1.5) -- (1.8,0.8);
\draw[color=gray] (1,1.5) -- (2.1,1.7);
\draw[color=gray] (0.3,-0.7) -- (0.8,0.1) -- (1.8,0.8)-- (1.6,-0.5) -- (0.3,-0.7);
\draw[color=gray] (0.1,0.8) -- (1.8,0.8);
\draw[color=green!50!black,thick] (-0.3,-0.5) -- (0.3,-0.7) -- (0.8,0.1) -- (0.1,0.8)-- (-0.3,-0.5);
\draw[color=cyan!70!black,thick] (1.8,0.8) -- (2.8,-0.7)-- (1.6,-0.5) -- (1.8,0.8);
\draw[color=cyan!70!black,thick] (1.8,0.8)-- (2.8,-0.7)-- (3.3,0.3) -- (3,1) -- (1.8,0.8);
\draw[color=cyan!70!black,thick] (2.8,-0.7)-- (4,-0.6) -- (4.1,0.4) -- (3.3,0.3) ;
\draw[color=cyan!70!black,thick] (1.8,0.8) -- (2.1,1.7) -- (3.2,1.7);
\draw[color=cyan!70!black,thick] (3,1) -- (3.2,1.7) -- (4.2,1.2) -- (4.1,0.4);

\fill[black] (3,0.5) circle (2pt);
\fill[black] (3.3,-0.2) circle (2pt);


\coordinate [label=left:$\textcolor{green!50!black}{S_{l_i}}$] (A) at (-0.3,-0.5);
\coordinate [label=right:$\textcolor{cyan!70!black}{S_{l_i'}}$] (A) at (3,-1);

\draw[dashed,->,>=stealth'] (3.2,0.5) -- (4.5,0.5) node[right] {$\chi_{i}(\delta t)$};
\draw[dashed,->,>=stealth'] (3.5,-0.2) -- (4.5,-0.2) node[right] {$x_{i}(\delta t)=x$};

\draw[color=red,dashed,thick] (6,0.5) -- (6.5,0.5);
\coordinate [label=right:$B(\chi_{i}(\delta t);r)$] (A) at (6.5,0.5);

\draw[color=blue,dashed,thick] (6,-0.7) -- (6.5,-0.7);
\coordinate [label=right:$\chi_{i}(\delta t)+tw_i$] (A) at (6.5,-0.7);


\draw[color=red,thick] (6,1.2) -- (6.5,1.2);
\coordinate [label=right:$x_{i}(t)$] (A) at (6.5,1.2);

\draw[color=olive!50!black,very thick,densely dotted] (6,1.9) -- (6.5,1.9);
\coordinate [label=right:$\bar{x}_{i}(t)$] (A) at (6.5,1.9);

\draw[color=blue,dashed,->,thick,>=stealth'] (3,0.5) -- (3.3,-0.2);

\draw[color=blue,->,thick,>=stealth'] (0,0)  .. controls (1.4,0.7) .. (3,0.5);
\draw[color=olive!50!black,->,very thick,>=stealth',densely dotted] (0,0)  .. controls (1.6,0.5) .. (3.3,-0.2);
\draw[color=red,->,thick,>=stealth'] (0.2,-0.6) .. controls (1.6,0.3) .. (3.3,-0.2);

\fill[black] (0.2,-0.6) circle (1.5pt) node[below] {$x_{i0}$};;
\fill[black] (0,0) circle (2pt) node[above left]{$x_{i,G}$};

\end{tikzpicture}
\vspace{-0.4cm}

\end{center}
\caption{Consider any point $x$ inside the ball with center $\chi_i(\delta t)$. Then, for each initial condition $x_{i0}$ in the cell $S_{l_i}$, the endpoint of agent's $i$ trajectory $x_i(\cdot)$ coincides with the endpoint of the curve $\bar{x}_i(\cdot)$, which is precisely $x$, and lies in $S_{l_i'}$, namely, $x_i(\delta t)=\bar{x}_i(\delta t)=x\in S_{l_i'}$.}
\end{figure}
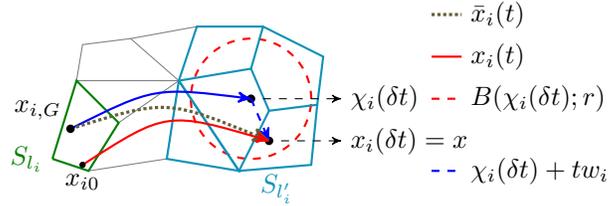

\section{Abstractions for Multi-Agent Systems}

In this section we formalize the discussion in Section 3, by exploiting a class of hybrid feedback laws containing the control laws introduced in \eqref{feedback:ki}. One reason for employing the subsequent analysis in an abstract framework is that the selection in  \eqref{feedback:ki1}-\eqref{feedback:ki3:plan} is not the only possible. For instance, it can be shown that selecting the control laws $k_{i,\bf{l}_i,1}(x_{i},\bf{x}_j):=f_i(x_{i,G},\bf{x}_{j,G})-f_i(x_i,\bf{x}_j)$, $k_{i,\bf{l}_i,2}(x_{i0})$ as before and $k_{i,\bf{l}_i,3}(w_i):=w_i$, can also provide well posed discretizations for system \eqref{single:integrator}-\eqref{general:feedback:law}. However, the latter will necessitate finer admissible discretizations and hence, increase the complexity of the symbolic model. In the sequel, given a time step $\delta t$ and the bounds $M$ and $v_{\max}$ on the feedback and input terms provided by \eqref{dynamics:bound} and \eqref{input:bound}, respectively, it is convenient to introduce the following lengthscale
\begin{equation} \label{Rmax}
R_{\max}:=\delta t(M+v_{\max}).
\end{equation}

\noindent It follows from \eqref{single:integrator}, \eqref{general:feedback:law}, \eqref{dynamics:bound}, \eqref{input:bound} and \eqref{Rmax} that $R_{\max}$ is the maximum distance an agent can travel within time $\delta t$.

Before defining the notion of a well posed space-time discretization we define the class of hybrid feedback laws which are assigned to the free inputs $v_i$ in order to obtain meaningful discrete transitions. For each agent, these control laws are parameterized by the agent's initial conditions and a set of auxiliary parameters belonging to a nonempty subset $W$ of $\Rat{n}$. These parameters, as discussed in the previous section, are exploited for motion planning. In particular, for every agent $i$, each vector $w_i\in W$ is in a one-to-one correspondence with a point inside a reachable ball for $i$, and the agent can reach this point by selecting the control law corresponding to the specific parameter $w_i$. The latter provides the possibility for the agent to perform transitions to different cells, namely, all cells which have nonempty intersection with that ball. Furthermore, we note that in accordance to the control laws introduced in \eqref{feedback:ki} for each agent $i$, the feedback laws in the following definition depend on the selection of the cells where $i$ and its neighbors belong. One basic requirement for this class of controllers consists of conditions that guarantee well posed solutions for the system (condition (P1) in Definition \ref{control:class}, below). We also impose a consistency requirement (condition (P2) in Definition \ref{control:class}) that their magnitude does not exceed the maximum bound on the free inputs \eqref{input:bound}, when the states of the agent and its neighbors lie in an appropriate inflation of their corresponding cells (in particular, an overapproximation of their reachable states over the time interval $[0,\delta t]$).

\begin{dfn}\label{control:class}
\noindent Given a cell decomposition $\S=\{S_l\}_{l\in\I}$ of $D$, a time step $\delta t$ and a nonempty subset $W$ of $\Rat{n}$, consider an agent $i\in\N$ and an initial cell configuration $\bf{l}_i=(l_i,l_{j_1},\ldots,l_{j_{N_i}})$ of $i$. For each $x_{i0}\in S_{l_i}$ and $w_i\in W$, let $T(x_{i0},w_i)>0$ and consider a mapping $k_{i,\bf{l}_i}(\cdot;x_{i0},w_i):[0,T(x_{i0},w_i))\times D^{N_i+1}\to\Rat{n}$, parameterized by $x_{i0}\in S_{l_i}$ and $w_i\in W$. We say that $k_{i,\bf{l}_i}(\cdot)$ satisfies \bf{Property \textbf(P)}, if the following conditions are satisfied.

\noindent\textbf{(P1)} For each $x_{i0}\in S_{l_i}$ and $w_i\in W$, the mapping $k_{i,\bf{l}_i}(\cdot;x_{i0},w_i)$ is locally Lipschitz continuous.

\noindent\textbf{(P2)} It holds
\begin{align}
& |k_{i,\bf{l}_i}(t,x_{i},\bf{x}_j;x_{i0},w_i)|\le v_{\max}, \forall t \in[0,\delta t]\cap [0,T(x_{i0},w_i)),\nonumber \\
& x_{i} \in (S_{l_i}+B(R_{\max}))\cap D, x_{j_{m}} \in (S_{l_{j_{m}}} +B(R_{\max}))\cap D,  \nonumber \\
& m=1,\ldots,N_i, x_{i0} \in S_{l_i},w_i\in W, \label{feedback:k:bound}
\end{align}

\noindent with $v_{\max}$ as given in \eqref{input:bound} and $R_{\max}$ as in \eqref{Rmax}.

\noindent\textbf{(P3)} It holds $T(x_{i0},w_i)>\delta t$, for all $x_{i0}\in S_{l_i}$, $w_i\in W$. $\triangleleft$
\end{dfn}

The motivation for considering the time interval $[0,T(x_{i0},w_i))$ in Definition \ref{control:class} comes from the maximal right interval on which the {modification of agent's $i$ reference trajectory in \eqref{time:T:plan}} remains inside the domain $D$. We next provide an extra Condition (C) for the feedback laws provided in the above definition, which is needed in order to define well posed discretizations.

\begin{dfn} \label{conditionC}
Consider a cell decomposition $\S=\{S_l\}_{l\in\I}$ of $D$, a time step $\delta t$ and a nonempty subset $W$ of $\Rat{n}$. Given an agent $i\in\N$, a cell configuration $\bf{l}_i=(l_i,l_{j_1},\ldots,l_{j_{N_i}})$ of $i$, a control law
\begin{equation} \label{feedback:for:i}
v_i=k_{i,\bf{l}_i}(t,x_{i},\bf{x}_j;x_{i0},w_i)
\end{equation}

\noindent as in Definition \ref{control:class} that satisfies Property (P), a vector $w_i\in W$, and a cell index $l_i'\in\I$, we say that \bf{Condition (C)} is satisfied, or specifically, that $\bf{l}_i$, $k_{i,\bf{l}_i}(\cdot)$, $w_i$, $l_i'$ satisfy Condition (C), if the following hold. For each initial cell configuration $\bf{l}$ with ${\rm pr}_i(\bf{l})=\bf{l}_i$, $\bf{l}=(l_{1},\ldots,l_{N})$, and for all $\ell\in\N\setminus \{{i}\}$ and feedback laws
\begin{equation} \label{feedback:for:others}
v_{\ell}=k_{\ell,\bf{l}_{\ell}}(t,x_{\ell},\bf{x}_{j(\ell)};x_{\ell 0},w_{\ell}),
\end{equation}

\noindent that satisfy Property (P) (with $\bf{l}_{\ell}={\rm pr}_{\ell}(\bf{l})$), the solution of the closed-loop system \eqref{single:integrator}-\eqref{general:feedback:law}, \eqref{feedback:for:i}-\eqref{feedback:for:others} is well defined on $[0,\delta t]$ and satisfies $x_{i}(\delta t,x(0))\in S_{l_i'}$, for all initial conditions $x(0)\in D^N$ with $x_i(0)=x_{i0}\in S_{l_i}$ $x_{\ell}(0)=x_{\ell 0}\in S_{l_{\ell}}$, $\ell\in\N\setminus\{i\}$ and $w_{\ell}\in W$,  $\ell\in\N\setminus\{i\}$. $\triangleleft$
\end{dfn}

Notice that when Condition (C) is satisfied, agent $i$ is driven to cell $S_{l_i'}$ precisely in time $\delta t$ under the feedback law $k_{i,\bf{l}_i}(\cdot)$ corresponding to the given parameter $w_i$ in the definition. In particular, Condition (C) ensures that the latter holds for any choice of feedback laws in place of the other agents' free inputs, as long as these control laws satisfy Property (P). We next provide the definition of a well posed space-time discretization. This definition formalizes our discussion on the possibility to assign a feedback law to each agent, in order to enable a meaningful transition from an initial to a final cell.

\begin{dfn}\label{well:posed:discretization}
Consider a cell decomposition $\S=\{S_l\}_{l\in\I}$ of $D$, a time step $\delta t$ and a nonempty subset $W$ of $\Rat{n}$.

\noindent \textbf{(i)} Given an agent $i\in\N$, an initial cell configuration $\bf{l}_i=(l_i,l_{j_1},\ldots,l_{j_{N_i}})$ of $i$ and a cell index $l_i'\in\I$ we say that \bf{the transition} $l_i\overset{\bf{l}_i}{\longrightarrow}l_i'$ \bf{is well posed with respect to the space-time discretization $\S-\delta t$}, if there exist a feedback law  $v_i=k_{i,\bf{l}_i}(\cdot;x_{i0},w_i)$ as in Definition \ref{control:class} that satisfies Property (P), and a vector $w_i\in W$, such that Condition (C) in Definition \ref{conditionC} is fulfilled.

\noindent \textbf{(ii)} We say that the \bf{space-time discretization} $\S-\delta t$ \bf{is well posed}, if for each agent $i\in\N$ and cell configuration $\bf{l}_i=(l_i,l_{j_1},\ldots,l_{j_{N_i}})$ of $i$, there exists a cell index $l_i'\in\I$ such that the transition $l_i\overset{\bf{l}_i}{\longrightarrow}l_i'$ is well posed with respect to $\S-\delta t$.
\end{dfn}

Given a space-time discretization $\S-\delta t$ and based on Definition \ref{well:posed:discretization}(i), it is now possible to provide an exact definition of the discrete transition system which serves as an abstract model for the behaviour of each agent.

\begin{dfn} \label{individual:ts}
For each agent $i$, its \bf{individual transition system} $TS_i:=(Q_i,Act_i,\longrightarrow_i)$ is defined as follows:

\noindent \textbullet\; $Q_i:=\I$ (the indices of the cell decomposition)

\noindent \textbullet\; $Act_i:=\I^{N_i+1}$ (the set of all cell configurations of $i$)

\noindent \textbullet\; $l_i\overset{\bf{l}_i}{\longrightarrow_i}l_i'$ iff $l_i\overset{\bf{l}_i}{\longrightarrow}l_i'$ is well posed, for each
$l_i,l_i'\in Q_i$ and $\bf{l}_i=(l_i,l_{j_1},\ldots,l_{j_{N_i}})\in Act_i$. $\triangleleft$
\end{dfn}

\noindent We have preferred to use the term actions instead of labels (as for instance in \cite{PgGaTp08}) for the elements of the set $Act_i$, because the cell configuration of $i$ indicates how the feedback term $f_i(\cdot)$ acts on and affects the possible transitions of agent $i$.

\begin{rem} \label{remark:post:nonempty}
\bf{(i)} Given a well posed space-time discretization $\S-\delta t$ and an initial cell configuration $\bf{l}=(l_{1},\ldots,l_{N})$, it follows from Definitions \ref{well:posed:discretization} and  \ref{individual:ts} that for each agent $i\in\N$ it holds ${\rm Post}_i(l_{i};{\rm pr}_{i}(\bf{l}))\ne\emptyset$ (${\rm Post}_i(\cdot)$ refers to the transition system $TS_i$ of each agent-see also Section 2).

\noindent \bf{(ii)} According to Definition \ref{well:posed:discretization}, given a control law $k_{i,\bf{l}_i}(\cdot)$ it is possible to perform transitions to different cells by an alternative selection of $w_i$.

\noindent \bf{(iii)} In addition, it is also possible to obtain different transitions by choosing an alternative control law. In particular, 
it is possible for the control laws considered in \eqref{feedback:ki} to obtain a different reference trajectory in \eqref{reference:solution} by selecting another set of points $(x_{i,G},\bf{x}_{j,G})$ and hence, reach a ball which intersects different cells (see Fig. 2). $\triangleleft$
\end{rem}

Assume a well posed space-time discretization $\S-\delta t$ is given. Based on Definition \ref{well:posed:discretization} and Remark \ref{remark:post:nonempty} we proceed by providing a 
modification of each transition system $TS_i$ that captures additional information on the control actions that realize the individual transitions. In particular, for each $i\in\N$ and cell configuration $\bf{l}_i=(l_i,l_{j_1},\ldots,l_{j_{N_i}})$ of $i$ we pick a control law $k_{i,\bf{l}_i}(\cdot)$ which generates at least one well posed transition, i.e., such that $\bf{l}_i$, $k_{i,\bf{l}_i}(\cdot)$, $w_i$, $l_i'$ satisfy Condition (C) for certain $w_i\in W$ and $l_i'\in\I$ (this is always possible since the discretization is well posed) and define for all $l\in\I$
\begin{equation}\label{wi:class}
[w_i]_{(\bf{l}_i,l)}:=\{w\in W:\bf{l}_i,k_{i,\bf{l}_i}(\cdot),w,l\ \textup{satisfy Condition(C)}\}.
\end{equation}

\noindent Based on \eqref{wi:class}, we next provide for each agent the controlled version of its individual transitions system. 

\begin{dfn} \label{controlled:TS}
\noindent Consider a well posed space-time discretization and select for each agent $i$ and cell configuration $\bf{l}_i$  a control law $k_{i,\bf{l}_i}(\cdot)$ which generates at least one well posed transition. Then, the \bf{controlled individual transition system} $TS_i^{c}:=(Q_i,Act_i^c,\longrightarrow_i^c\nobreak)$ of each agent $i$ is defined as follows

\noindent \textbullet\; $Q_i:=\I$

\noindent \textbullet\; $Act_i:=\I^{N_i+1}\times 2^W$

\noindent \textbullet\; $l_i\overset{(\bf{l}_i,[w_i])}{\longrightarrow_i^c}l_i'$ iff $l_i\overset{\bf{l}_i}{\longrightarrow}_il_i'$ and $[w_i]=[w_i]_{(\bf{l}_i,l_i')}\ne\emptyset$, for each
$l_i,l_i'\in Q_i$, $\bf{l}_i=(l_i,l_{j_1},\ldots,l_{j_{N_i}})\in\I^{N_i+1}$ and $[w_i]\in 2^W$, with $[w_i]_{(\bf{l}_i,l_i')}$ as defined in \eqref{wi:class}.  $\triangleleft$
\end{dfn}

\begin{rem} 
\noindent \bf{(i)} Although the set $2^W$ has uncountable cardinality, we have preferred this representation for the control actions (instead of e.g., to select $W$), because for each agent $i$ and cell configuration  $\bf{l}_i$ the possible actions which generate a transition are bounded by the successor states of the agent. 

\noindent \bf{(ii)} For each agent $i$ and cell configuration $\bf{l}_i=(l_i,l_{j_1},\ldots,l_{j_{N_i}})$ it holds $\cup_{[w_i]\in 2^W} {\rm Post}_i^c(l_i;(\bf{l}_i,[w_i]))\subset {\rm Post}_i(l_i;\bf{l}_i)$, since the transitions in $TS_i^c$ are associated to the specific controller selection for the cell configuration. $\triangleleft$ 
\end{rem}

Next, notice that according to Definition \ref{well:posed:discretization}, a well posed space-time discretization requires the existence of a well posed transition for each agent $i$. The latter reduces to the selection of an appropriate feedback controller for $i$, which also satisfies Property (P), and the requirement that the selected feedback controllers of the other agents also satisfy (P). Yet, it is not completely evident, that given an initial cell configuration and a well posed transition for each agent, it is possible to choose a feedback law for each agent, so that the resulting closed-loop system will guarantee all these well posed transitions (for all possible initial conditions in the cell configuration). The following proposition clarifies this point.

\begin{prop}\label{discrete:transitions:result}
Consider system \eqref{single:integrator}-\eqref{general:feedback:law}, let $\bf{l}=(l_{1},\ldots,l_{N})$ be an initial cell configuration and assume that the space-time discretization $S-\delta t$ is well posed, which according to Remark \ref{remark:post:nonempty} implies that for all $i\in\N$ it holds that ${\rm Post}_i(l_{i};{\rm pr}_{i}(\bf{l}))\ne\emptyset$. Then, for every final cell configuration $\bf{l}'=(l_1',\ldots,l_N')\in{\rm Post}_1(l_1;{\rm pr}_1(\bf{l}))\times\cdots\times{\rm Post}_N(l_N;{\rm pr}_N(\bf{l}))$, there exist feedback laws
\begin{equation} \label{feedback:for:all}
v_i=k_{i,{\rm pr}_i(\bf{l})}(t,x_{i},\bf{x}_j;x_{i0},w_i),i\in\N,
\end{equation}

\noindent satisfying Property (P), and $w_1,\ldots,w_N\in W$, such that for each $i\in\N$, the solution of the closed-loop system \eqref{single:integrator}-\eqref{general:feedback:law}, \eqref{feedback:for:all} (with $v_{m}=k_{m,{\rm pr}_{m}(\bf{l})}$, $m\in\N$) is well defined on $[0,\delta t]$, and its $i$-th component satisfies
\begin{align}
x_{i}(\delta t,x(0)) & \in S_{l_i'}, \forall x(0)\in D^N: \nonumber \\
x_{m}(0) & =x_{m 0}\in S_{l_{m}},m\in\N. \label{contoler:compatibility}
\end{align}
\end{prop}

\begin{proof}
The proof is given in the Appendix.
\end{proof}

The result of the following proposition guarantees that the selection of the controllers introduced in Definition \ref{control:class} provides well posed solutions for the closed-loop system on the time interval $[0,\delta t]$. We exploit this result in Theorem \ref{discretizations:for:planning}, where we derive sufficient conditions for well posed space-time discretizations, which are also suitable for motion planning. Furthermore, Proposition \ref{completeness:result} guarantees that the magnitude of the hybrid feedback laws does not exceed the maximum allowed magnitude $v_{\max}$ of the free inputs on $[0,\delta t]$, and hence, establishes consistency with the initial design requirement. In particular, it follows that every solution of the closed-loop
 system on $[0,\delta t]$ is identical to a solution of the original system \eqref{single:integrator}-\eqref{general:feedback:law} with the same initial condition and certain free input $v(\cdot)$ satisfying $|v_i(t)|\le v_{\max}$, for all $t\ge 0$ and $i\in\N$. For certain technical reasons concerning the proofs in the next sections, it is convenient to obtain the first results of the proposition for feedback laws that only satisfy Properties (P1) and (P2) of Definition \ref{control:class}.

\begin{prop} \label{completeness:result}
Consider the space-time discretization $\S-\delta t$ corresponding to the cell decomposition $\S$ of $D$ and the time step $\delta t$. Let $\bf{l}=(l_{1},\ldots,l_{N})$ be an initial cell configuration and consider any feedback laws of the form
\begin{equation} \label{feedback:for:all2}
v_i=k_{i,{\rm pr}_i(\bf{l})}(t,x_{i},\bf{x}_j;x_{i0},w_i),i\in\N
\end{equation}

\noindent assigned to the agents that satisfy Properties (P1) and (P2). Then:

\noindent \textbf{(i)} For each $w_i\in W$, $i\in\N$ and initial condition $x(0)\in D^N$ with $x_{i}(0)=x_{i0}\in S_{l_i}$, $i\in\N$, the solution of the closed-loop system \eqref{single:integrator}-\eqref{general:feedback:law}, \eqref{feedback:for:all2} (with $v_i=k_{i,{\rm pr}_i(\bf{l})}$, $i\in\N$) is defined and remains in $D^N$ for all $t\in [0,\tilde{T})$, where
\begin{equation}\label{time:tildeT}
\tilde{T}:=\min\{\delta t,\min\{T(x_{i0},w_i):i\in\N\}\}
\end{equation}

\noindent and
\begin{equation}\label{limit:solution}
\lim_{t\to\tilde{T}^-}x(t)\in D^N.
\end{equation}

Assume additionally that (P3) also holds, namely, that (P) is satisfied. Then:

\noindent \textbf{(iia)} The solution $x(t)$ of \eqref{single:integrator}-\eqref{general:feedback:law}, \eqref{feedback:for:all2} above remains in $D^{N}$ for all $t\in [0,\delta t]$ and satisfies
\begin{equation} \label{feedback:consistency}
|k_{i,{\rm pr}_i(\bf{l})}(t,x_{i}(t),\bf{x}_j(t);x_{i0},w_i)|\le v_{\max},\forall t\in[0,\delta t], i\in\N,
\end{equation}

\noindent which provides the desired consistency with the design requirement \eqref{input:bound} on the $v_i$'s.

\noindent \textbf{(iib)} There exists a piecewise continuous function $v=(v_1,\ldots,v_N):[0,\infty)\to\Rat{Nn}$ satisfying $|v_i(t)|\le v_{\max}$, $\forall t\ge 0$, $i\in\N$, such that the solution $x(\cdot)$ above and the solution $\xi(\cdot)$ of \eqref{single:integrator}-\eqref{general:feedback:law}, with the same initial condition as $x(\cdot)$ and input $v(\cdot)$, coincide on $[0,\delta t]$.
\end{prop}

\begin{proof}
The proof is given in the Appendix.
\end{proof}

\begin{rem}
Note, that the result of part (i) of Proposition \ref{completeness:result} holds for any selection of feedback laws $v_i=k_{i,{\rm pr}_i(\bf{l})}(\cdot)$ that satisfy Properties (P1) and (P2). Respectively, the results of parts (iia) and (iib) hold for all selections of feedback laws $v_i=k_{i,{\rm pr}_i(\bf{l})}(\cdot)$ that satisfy Property (P).
\end{rem}

In the final result of this section, we merge the results of Propositions \ref{discrete:transitions:result} and \ref{completeness:result}, and show that each infinite discrete behaviour of the decentralized abstraction can be implemented by a continuous controller, which is compatible with the restrictions on the free inputs, and produces a continuous trajectory that satisfies the Invariance Assumption (IA). In particular, we prove that for each possible discrete transition sequence of the overall system which is compatible with the individual transition system of each agent, there exists a trajectory of the continuous time system \eqref{single:integrator}-\eqref{general:feedback:law} which satisfies (IA) and generates the discrete trajectory when sampled at time intervals of length $\delta t$. Before proceeding, we introduce the following notion of the product transition system.

\begin{dfn}\label{product:TS}
\textbf{(i)} Consider the space-time discretization $\S-\delta t$, and for each agent $i\in\N$, its individual transition system $TS_i$ as provided by Definition \ref{individual:ts}. The \bf{product transition system} $TS_{\P}:=(Q_{\P},Act_{\P},\longrightarrow_{\P})$ is defined as follows:

\noindent \textbullet\; $Q_{\P}:=\I^N$ (all possible cell configurations)

\noindent \textbullet\; $Act_{\P}:=\{*\}$\footnote{Following notation \cite[page 11]{Tp09}} 

\noindent \textbullet\; $\bf{l}\overset{*}{\longrightarrow_{\P}}\bf{l}'$, iff $l_i'\in{\rm Post}_i(l_i;{\rm pr}_i(\bf{l})),\forall i\in\N$, for all $\bf{l}=(l_1,\ldots,l_N)$, $\bf{l}'=(l_1',\ldots,l_N')$.

\noindent \textbf{(ii)} Given an initial cell configuration $\bf{l}^0\in \I^N$, a path originating from $\bf{l}^0 $ in $TS_{\P}$, is an infinite sequence of states $\bf{l}^0 \bf{l}^1\bf{l}^2\ldots$ such that $\bf{l}^i\overset{*}{\longrightarrow}\bf{l}^{i+1}$ for all $i\in\mathbb{N}\cup\{0\}$. $\triangleleft$
\end{dfn}

\begin{rem}
Given a well posed space-time discretization $\S-\delta t$ and an initial cell configuration $\bf{l}^0\in \I^N$, it follows from Remark \ref{remark:post:nonempty} and Definition \ref{product:TS} that there exists at least one path $\bf{l}^0 \bf{l}^1\bf{l}^2\ldots$ in $TS_{\P}$ originating from $\bf{l}^0$. $\triangleleft$
\end{rem}

We are now in position to state the result providing the consistency between the discrete plan level and the continuous controller implementation.

\begin{prop} \label{proposition:runs}
Assume that the space-time discretization $\S-\delta t$ is well posed for the multi-agent system  \eqref{single:integrator}-\eqref{general:feedback:law}. Then, for each initial cell configuration $\bf{l}^0=(l_1^0,\ldots,l_N^0)$ and path $\bf{l}^0 \bf{l}^1\bf{l}^2\ldots$ originating from $\bf{l}^0 $ in $TS_{\P}$, and for each initial condition $x(0)\in D^N$ of \eqref{single:integrator}-\eqref{general:feedback:law} satisfying $x_i(0)\in S_{l_i^0}$, $i\in\N$, there exists a free input $\bar{v}(\cdot)$ satisfying $|\bar{v}_i(t)|\le v_{\max}$, $\forall t\ge 0$, $i\in\N$, such that the solution $x(t)$ of the system remains in $D^N$ for all $t\ge 0$ and satisfies $x_i(m\delta t)\in S_{l_i^{m}}$ for each $m\in\mathbb{N}$ and $i\in\N$.
\end{prop}

\begin{proof}
The proof is carried out by induction and is based on the results of Propositions \ref{discrete:transitions:result} and \ref{completeness:result} (iib). Before stating the induction hypothesis, we provide certain basic properties of the deterministic control system \eqref{single:integrator}-\eqref{general:feedback:law} which can be found in \cite[Chapter 1]{KiJj11}, or \cite[Chapter 2]{Se98}.

Recall that according to our hypotheses, the input set $\mathcal{U}$ of the multi-agent system consists of all piecewise continuous inputs $v:\RgeO\to\Rat{Nn}$ satisfying $|v_i(t)|\le v_{\max}$, $\forall t\ge 0$, $i\in\N$. Also, for each $r>0$ we define the shift operator ${\rm Sh}_r:\mathcal{U}\to\mathcal{U}$ as
$$
{\rm Sh}_r(v)(t):=v(t+r),\forall t\ge 0,
$$

\noindent which implies that
\begin{equation}\label{shift:operator}
v'={\rm Sh}_r(v)\iff v'(t)=v(t+r),\forall t\ge 0\iff v'(t-r)=v(t),\forall t\ge r.
\end{equation}

\noindent To the control system \eqref{single:integrator}-\eqref{general:feedback:law} we associate the transition map $\varphi:A_{\varphi}\to D^N$ with $$
A_{\varphi}:=\{(t,t_0,x_0;v):t\ge t_0\ge 0,x_0\in D^N,v\in\mathcal{U}\},
$$

\noindent where $\varphi(t,t_0,x_0;v)$ denotes the value at time $t$ of the unique solution of \eqref{single:integrator}-\eqref{general:feedback:law} with initial condition $x_0$ at time $t_0$ and input $v(\cdot)$. Notice, that by virtue of the Invariance Assumption (IA), $\varphi(\cdot)$ is well defined. The map $\varphi(\cdot)$ satisfies the following properties:

\noindent \textbullet\textbf{Causality.}\; For each $t> t_0\ge 0$, $x_0\in D^N$ and $v^1,v^2\in\mathcal{U}$ with $v^1|_{[t_0,t)}=v^2|_{[t_0,t)}$ it holds
$$
\varphi(t,t_0,x_0;v^1)=\varphi(t,t_0,x_0;v^2),
$$

\noindent where $v^1|_{[t_0,t)}$ denotes the restriction of $v^1(\cdot)$ to $[t_0,t)$ (see Section 2).

\noindent \textbullet\textbf{Semigroup Property.}\; For each $t_2\ge t_1\ge t_0\ge 0$, $x_0\in D^N$ and $v\in\mathcal{U}$ it holds
$$
\varphi(t_2,t_1,\varphi(t_1,t_0,x_0;v);v)=\varphi(t_2,t_0,x_0;v).
$$

\noindent \textbullet\textbf{Time Invariance.}\; For each $r>0$, $t\ge t_0\ge r$, $x_0\in D^N$ and $v\in\mathcal{U}$ it holds
$$
\varphi(t,t_0,x_0;v)=\varphi(t-r,t_0-r,x_0;{\rm Sh}_r(v)).
$$

Now consider an initial cell configuration $\bf{l}^0=(l_1^0,\ldots,l_N^0)$, a path $\bf{l}^0 \bf{l}^1\bf{l}^2\ldots$ originating from $\bf{l}^0 $ in $TS^{\P}$, and an initial condition $x(0)\in D^N$ satisfying $x_i(0)\in S_{l_i^0}$, $i\in\N$. We will determine a free input $\bar{v}(\cdot)$ satisfying $|\bar{v}_i(t)|\le v_{\max}$, $\forall t\ge 0$, $i\in\N$, such that the corresponding solution $x(t):=\varphi(t,0,x(0);\bar{v})$ remains in $D^N$ for all $t\ge 0$ and satisfies $x_i(m\delta t)\in l_i^{m}$ for each $m\in\mathbb{N}$ and $i\in\N$. The construction of $\bar{v}(\cdot)$ is based on the following Induction Hypothesis, which constitutes the core of the proposition.

\noindent \textbf{Induction Hypothesis (IH).} For each $m\in\mathbb{N}$ there exists a piecewise continuous input $v^{m}:\RgeO\to\Rat{Nn}$ satisfying $|v^{m}_i(t)|\le v_{\max}$, $\forall t\ge 0$, $i\in\N$ and such that
\begin{align}
v^{m}|_{[0,\kappa\delta t)} & =v^{\kappa}|_{[0,\kappa\delta t)},\forall \kappa=1,\ldots,m-1, \label{IH:property1} \\
\varphi(t,0,x(0);v^{m}) & =\varphi(t,0,x(0);v^{\kappa}),\forall \kappa=1,\ldots,m-1,t\in[0,\kappa\delta t], \label{IH:property2} \\
\varphi_i(m\delta t,0,x(0);v^{m}) & \in S_{l^{m}_i},\forall i\in\N. \label{IH:property3}
\end{align}

\noindent $\triangleright$\textbf{Proof of (IH).} In order to prove (IH) for $m=1$, we need  to show that \eqref{IH:property3} is fulfilled. From the fact that $x_i(0)\in S_{l_i^0}$, $i\in\N$ and that $\bf{l}^0\overset{*}{\longrightarrow}\bf{l}^1$, we deduce from Definition \ref{product:TS} and Proposition \ref{discrete:transitions:result} that there exist feedback laws $k_{i,{\rm pr}_i(\bf{l})}(\cdot)$ as in \eqref{feedback:for:all} (with $\bf{l}=\bf{l}^0$), which satisfy Property (P), and $w_1,\ldots,w_N\in W$ such that  \eqref{contoler:compatibility} holds with $\bf{l}'=\bf{l}^1$. Hence, it follows from Proposition \ref{completeness:result}(iib) that there exists a piecewise continuous input $v^1:\RgeO\to\Rat{Nn}$ satisfying $|v^1_i(t)|\le v_{\max}$, $\forall t\ge 0$, $i\in\N$ and such that
$$
\varphi_i(\delta t,0,x(0);v^1) \in S_{l^1_i},\forall i\in\N
$$

\noindent which establishes \eqref{IH:property3} for $m=1$.

Now assume that (IH) holds for certain $m\in\mathbb{N}$. We will show that it is also valid for $m+1$. By exploiting Property \eqref{IH:property3} of (IH) for $m$ and that $\bf{l}^{m}\overset{*}{\longrightarrow}\bf{l}^{m+1}$, we deduce from Definition \ref{product:TS} and Proposition \ref{discrete:transitions:result} that there exist feedback laws $k_{i,{\rm pr}_i(\bf{l})}(\cdot)$ as in \eqref{feedback:for:all} (with $\bf{l}=\bf{l}^{m}$), which satisfy Property (P), and $w_1,\ldots,w_N\in W$ such that  \eqref{contoler:compatibility} holds with $\bf{l}'=\bf{l}^{m+1}$. Hence, it follows from Proposition \ref{completeness:result}(iib) that there exists a piecewise continuous input $v:\RgeO\to\Rat{Nn}$ satisfying $|v_i(t)|\le v_{\max}$, $\forall t\ge 0$, $i\in\N$ and such that
\begin{equation} \label{phi:property1}
\varphi_i(\delta t,0,\varphi(m\delta t,0,x(0);v^{m});v) \in S_{l^{m+1}_i},\forall i\in\N.
\end{equation}

\noindent We following define $v^{m+1}:\RgeO\to\Rat{Nn}$ as
\begin{equation}  \label{tilde:v}
v^{m+1}(t):=\left\lbrace \begin{array}{ll}
v^{m}(t), & t\in [0,m\delta t), \\
v(t-m\delta t), & t\in [m\delta t,\infty).
\end{array}\right.
\end{equation}

\noindent Then, it follows from \eqref{tilde:v} that $v^{m+1}(\cdot)$ satisfies \eqref{IH:property1} (with $m:=m+1$) and the latter implies \eqref{IH:property2} (with $m:=m+1$) by causality. Hence, we get from \eqref{IH:property2} that
\begin{equation} \label{phi:property2}
\varphi(m\delta t,0,x(0);v^{m+1})=\varphi(m\delta t,0,x(0);v^{m}).
\end{equation}

\noindent From  the semigroup property and \eqref{phi:property2} we deduce that
\begin{align}
\varphi((m+1)\delta t,0,x(0);v^{m+1}) & =\varphi((m+1)\delta t,m\delta t,\varphi(m\delta t,0,x(0);v^{m+1});v^{m+1}) \nonumber \\
& = \varphi((m+1)\delta t,m\delta t,\varphi(m\delta t,0,x(0);v^m);v^{m+1}). \label{phi:property3}
\end{align}

\noindent Also, we get from \eqref{tilde:v} and \eqref{shift:operator} that
\begin{equation}\label{tilde:v:property}
v={\rm Sh}_{m\delta t}(v^{m+1}).
\end{equation}

\noindent Thus, it follows from time invariance  and \eqref{tilde:v:property} that
\begin{align}
\varphi((m+1)\delta t,m\delta t,\varphi(m\delta t,0,x(0);v^{m});v^{m+1}) & =\varphi(\delta t,0,\varphi(m\delta t,0,x(0);v^{m});{\rm Sh}_{m\delta t}(v^{m+1})) \nonumber \\
& = \varphi(\delta t,0,\varphi(m\delta t,0,x(0);v^m);v). \label{phi:property4}
\end{align}

\noindent Hence, we conclude from \eqref{phi:property1}, \eqref{phi:property3} and \eqref{phi:property4} that \eqref{IH:property3} also holds (with $m:=m+1$) and the proof of (IH) is complete. $\triangleleft$

In order to finish the proof of the proposition, define $\bar{v}:\RgeO\to\Rat{Nn}$ by
\begin{equation} \label{v:bar}
\bar{v}(t):=v^m(t),m\in\mathbb{N},t\in[(m-1)\delta t,m\delta t),
\end{equation}

\noindent with $v^m(\cdot)$ as given by (IH) for each $m\in\mathbb{N}$. Then, it follows from \eqref{v:bar} and (IH) that  $|\bar{v}_i(t)|\le v_{\max}$, $\forall t\ge 0$, $i\in\N$ and thus, by the Invariance Assumption (IA), the solution $x(t)$ of the system remains in $D^N$ for all $t\ge 0$. Furthermore, it holds that
\begin{equation} \label{v:bar:property}
\bar{v}|_{[0,m\delta t)}=v^m|_{[0,m\delta t)},\forall m\in\mathbb{N}.
\end{equation}

\noindent Indeed, for each $t\in [0,m\delta t)$ there exists $\kappa\in\{1,\ldots,m\}$ such that $t\in[(\kappa-1)\delta t,\kappa\delta t)$. If $\kappa=m$, then it follows from \eqref{v:bar} that $\bar{v}(t)=v^m(t)$. If $\kappa\in\{1,\ldots,m-1\}$, then we get from \eqref{v:bar} that $\bar{v}(t)=v^{\kappa}(t)$ and thus from \eqref{IH:property1} that $\bar{v}(t)=v^m(t)$. Hence, \eqref{v:bar:property} is valid. Finally, from \eqref{v:bar:property}, \eqref{IH:property3} and causality we conclude that for each $m\in\mathbb{N}$ it holds
$$
x_i(m\delta t)=\varphi_i(m\delta t,0,x(0);\bar{v})=\varphi_i(m\delta t,0,x(0);v^{m}) \in S_{l^m_i}
$$

\noindent and the proof is complete.
\end{proof}

We note that the result of Proposition \ref{proposition:runs} remains valid if we consider the product of the agents' controlled transition systems as given by Definition \ref{controlled:TS}, which will be determined explicitly in Section 6 for the control laws $k_{i,\bf{l}_i}$ in \eqref{feedback:ki}. This observation is summarized in the following remark.

\begin{rem} \label{remark:controlled:product}
Instead of the product $TS_{\P}$ formed by the agents' individual transition systems $TS_i$, consider the \bf{controlled product transition system} $TS_{\P}^c:=(Q_{\P}^c,Act_{\P}^c,\longrightarrow_{\P}^c\nobreak)$ formed by the controlled transition systems $TS_i^c$, $i\in\N$ in Definition \ref{controlled:TS} as follows: 

\noindent \textbullet\; $Q_{\P}^c=\I^N$; 

\noindent \textbullet\; $Act_{\P}^c=\{*\}$; 

\noindent \textbullet\; $\bf{l}\overset{*}{\longrightarrow_{\P}^c}\bf{l}'$, iff there exist $[w_1],\ldots,[w_N]\in 2^W$ such that $l_i'\in{\rm Post}_i^c(l_i;({\rm pr}_i(\bf{l}),[w_i])),\forall i\in\N$, for all $\bf{l}=(l_1,\ldots,l_N)$, $\bf{l}'=(l_1',\ldots,l_N')$. 

\noindent Then, the result of Proposition \ref{proposition:runs} remains valid for any initial cell  configuration $\bf{l}^0$ and path $\bf{l}^0 \bf{l}^1\bf{l}^2\ldots$ originating from $\bf{l}^0 $ in $TS_{\P}^c$. $\triangleleft$ 
\end{rem}

\section{Time Domain Properties of the Control Laws}

In this section we use the results of Section 4 in order to prove certain useful properties of the reference trajectory $\chi_i(\cdot)$ and the time domain $[0,T_i(x_{i0},w_i))$ of the control laws \eqref{feedback:ki} as specified by \eqref{time:T:plan}. We proceed by providing some extra details for the dynamics as determined by the control law in \eqref{general:feedback:law}. In particular, we assume that the $f_i$'s are globally Lipschitz functions. Furthermore, if we want to achieve more accurate bounds for the dynamics of the feedback controllers assigned to the free inputs $v_i$ (those will be clarified in the proof of Theorem \ref{discretizations:for:planning} in the next section), we can choose (possibly) different Lipschitz constants $L_{1},L_{2}>0$ such that
\begin{align}
|f_{i}(x_{i},\bf{x}_j)-f_{i}(x_{i},\bf{y}_j)|\le & L_{1}|(x_{i},\bf{x}_j)-(x_{i},\bf{y}_j)|, \label{dynamics:bound1} \\
|f_{i}(x_{i},\bf{x}_j)-f_{i}(y_{i},\bf{x}_j)|\le & L_{2}|(x_{i},\bf{x}_j)-(y_{i},\bf{x}_j)|,  \label{dynamics:bound2} \\
\forall x_{i},y_{i} \in & D, \bf{x}_j,\bf{y}_j \in D^{N_i}, i\in\N.\nonumber
\end{align}

\noindent In order to provide some extra informal motivation on considering both constants $L_1$ and $L_2$, we recall that in order to derive sufficient conditions for a well posed discretization, we design for each agent $i$ inside a cell $S_{l_i}$ a feedback, in order to ``track" a given reference trajectory (of $i$) starting in the same cell. In particular, the constant $L_{1}$ provides bounds on the feedback term \eqref{feedback:ki1} which compensates for the deviation of agent's $i$ dynamics from its corresponding dynamics along the reference trajectory, due to the time evolution of its neighbors' states. On the other hand, the constant $L_{2}$ provides bounds on the feedback term \eqref{feedback:ki3:plan} which compensates for the deviation of the initial state with respect to the initial state of the reference trajectory.

Based on the global Lipschitz assumption, we establish uniqueness of the reference trajectory $\chi_i(\cdot)$ and provide a lower bound for the right endpoint $T_{\max}$ of its maximal interval of existence, which is independent of the selection of $(x_{i,G},\bf{x}_{j,G})$ in \eqref{reference:point}. 

\begin{lemma} \label{lemma:Tmax}
For each tuple of reference points $(x_{i,G},\bf{x}_{j,G})$ as in \eqref{reference:point}, the initial value problem \eqref{reference:solution} has a unique solution which is defined and remains in $D$ on the right maximal interval $[0,T_{\max})$. Furthermore, it holds
\begin{equation}\label{Tmax:lower:bound}
T_{\max}>\frac{v_{\max}}{2ML_1\max\{\sqrt{N_i}:i\in\N\}}.
\end{equation}
\end{lemma}

\begin{proof}
For the proof of the lemma we exploit the result of Proposition \ref{completeness:result}. In particular, we show that the solution $\chi_i(\cdot)$ of \eqref{reference:solution} coincides on a suitable time interval with the $i$-th component of the solution of the multi-agent system  \eqref{single:integrator}-\eqref{general:feedback:law} under an appropriate selection of the initial conditions and feedback controllers for the $v_i$'s. Hence, by implicitly exploiting the Invariance Assumption (IA) that leads to the result of Proposition \ref{completeness:result}(iia), which is valid for any choice of feedback laws that satisfy Property (P), we will verify that \eqref{Tmax:lower:bound} is fulfilled.

In order to proceed with the proof, let $(x_{i,G},\bf{x}_{j,G})$ be a tuple of reference points as in \eqref{reference:point}, corresponding to a cell decomposition $\{S_{l}\}_{l\in\I}$ of $D$  and a cell configuration $\bf{l}_i$ of agent $i$,  and consider another cell decomposition $\{S^o_{l_o}\}_{l_o\in\I_o}$ of $D$ and an initial cell configuration $\bf{l}_o=(l_{o1},\ldots,l_{oN})\in\I_o^{N}$ with ${\rm pr}_i(\bf{l}_o)=(l_{oi},l_{oj_1},\ldots,l_{oj_{N_i}})$, such that
\begin{equation} \label{initial:cond:in:cells:lemma}
x_{i,G}\in S^o_{l_{oi}}\;{\rm and}\;S^o_{l_{oj_{\kappa}}}=x_{j_{\kappa},G},\kappa=1,\ldots,N_i.
\end{equation}

\noindent We have selected the auxiliary cell decomposition $\{S^o_{l_o}\}_{l_o\in\I_o}$ with the sets $S^o_{l_{oj_{\kappa}}}$ consisting of a single element, because this slightly simplifies the subsequent analysis and also allows obtaining a greater (uniform) lower bound for the time $T_{\max}$. Next, define the time step
\begin{equation} \label{time:deltato}
\delta t_o:=\frac{v_{\max}}{2ML_1\max\{\sqrt{N_i}:i\in\N\}}
\end{equation}

\noindent and consider the feedback laws $k_{i,{\rm pr}_i(\bf{l}_o)}:D^{N_i+1}\to\Rat{n}$ given by
\begin{equation}  \label{feedback:ki:aux}
k_{i,{\rm pr}_i(\bf{l}_o)}(x_{i},\bf{x}_j):=f_i(x_i,\bf{x}_{j,G})-f_i(x_i,\bf{x}_j)= F_{i,\bf{l}_i}(x_i)-f_i(x_i,\bf{x}_j),
\end{equation}

\noindent with $F_{i,\bf{l}_i}(\cdot)$ as in \eqref{averaged:dynamics} and $k_{\ell,{\rm pr}_{\ell}(\bf{l}_o)}:D^{N_{\ell}+1}\to\Rat{n}$ for $\ell\in\N\setminus\{i\}$ given by
\begin{equation} \label{feedback:ki:aux:others}
k_{\ell,{\rm pr}_{\ell}(\bf{l}_o)}(x_{\ell},\bf{x}_{j(\ell)}):=0.
\end{equation}

\noindent Note that the feedback laws $k_{\ell,{\rm pr}_{\ell}(\bf{l}_o)}(\cdot)$ for $\ell\in\N\setminus\{i\}$ satisfy Property (P) by default. Hence, in order to invoke Proposition \ref{completeness:result}(iia), we show that $k_{i,{\rm pr}_i(\bf{l}_o)}(\cdot)$ also satisfies (P). Property (P3) is obvious, since $k_{i,{\rm pr}_i(\bf{l}_o)}(\cdot)$ is independent of $t$. Property (P1) follows from the corresponding Lipschitz property for $f_{i}(\cdot)$ and $F_{i,\bf{l}_i}(\cdot)$, since the latter satisfies the Lipschitz condition
\begin{equation}  \label{tildefi:Lipschitz:const}
|F_{i,\bf{l}_i}(x)-F_{i,\bf{l}_i}(y)| \le L_{2}|x-y|,\forall x,y\in D.
\end{equation}

\noindent Indeed, due to \eqref{dynamics:bound2} and \eqref{averaged:dynamics}, we have that for each $x,y\in D$ it holds
\begin{equation*}
|F_{i,\bf{l}_i}(x)-F_{i,\bf{l}_i}(y)|=|f_{i}(x,\bf{x}_{j,G})-f_{i}(y,\bf{x}_{j,G})| \le L_{2}|(x,\bf{x}_{j,G})-(y,\bf{x}_{j,G})|=L_{2}|x-y|.
\end{equation*}

\noindent In order to show (P2), notice that due to \eqref{time:deltato} we get
\begin{equation} \label{time:deltato:consequence}
v_{\max}\ge 2M\delta t_o L_1\sqrt{N_i},\;\textup{for all}\;i\in\N.
\end{equation}

\noindent Hence, we get from \eqref{Rmax}, \eqref{dynamics:bound1}, \eqref{vmax:vs:M}, \eqref{initial:cond:in:cells:lemma} and \eqref{time:deltato:consequence} that for every $x_i\in (S^{o}_{l_{oi}}+B(R_{\max}))\cap D$ and $x_{j_{\kappa}}\in B(x_{j_{\kappa},G},R_{\max})\cap D$, $\kappa=1,\ldots,N_i$, it holds
\begin{align*}
|k_{i,{\rm pr}_i(l_o)}(x_{i},\bf{x}_j)|= & |f_i(x_i,\bf{x}_j)-f_i(x_i,\bf{x}_j)|\le L_1|\bf{x}_j-\bf{x}_{j,G}| = L_1\left(\sum_{\kappa=1}^{N_i}(x_{j_{\kappa}}-x_{j_{\kappa},G})^2\right)^{\frac{1}{2}} \\
\le & L_1\sqrt{N_i} R_{\max} = L_1\sqrt{N_i}\delta t_o(M+v_{\max}) < 2M\delta t_oL_1\sqrt{N_i}\le v_{\max},
\end{align*}

\noindent and thus (P2) holds as well, since $k_{i,{\rm pr}_i(\bf{l}_o)}(\cdot)$ is independent of $t$, $x_{i0}$ and $w_i$. Then, it follows from Proposition \ref{completeness:result}(iia) that the solution $x(t)$ of the closed-loop system \eqref{single:integrator}-\eqref{general:feedback:law}, \eqref{feedback:ki:aux}-\eqref{feedback:ki:aux:others} with initial condition $x(0)\in D^N$ satisfying $x_i(0)=x_{i,G}$, $x_{j_1}(0)=x_{j_1,G},\ldots,$ $x_{j_{N_i}}(0)=x_{j_{N_i},G}$ (and the initial state of each other agent $\ell$ belonging to $S^{o}_{l_{o\ell}}$) is defined and remains in $D^{N}$ for all $t\in [0,\delta t_o]$. Hence, the $i$-th component of the solution $x(\cdot)$ satisfies
\begin{equation} \label{xi:in:D:lemma}
x_i(t)\in D,\forall t\in[0,\delta t_0],
\end{equation}

\noindent and by virtue of \eqref{single:integrator}-\eqref{general:feedback:law} and \eqref{feedback:ki:aux}, it holds
\begin{equation} \label{correspondance:ref:solution}
\dot{x}_i=F_{i,\bf{l}_i}(x_i), x_i(0)=x_{i,G}, t\in[0,\delta t_0].
\end{equation}

\noindent Hence, it follows from \eqref{correspondance:ref:solution} that $x_i(\cdot)$ coincides with the unique solution $\chi_i(\cdot)$ of \eqref{reference:solution} on $[0,\delta t_o]\cap[0,T_{\max})$, which in conjunction with \eqref{xi:in:D:lemma} implies that $\chi_i(t)$ remains in a compact subset of $D$ for $t\in[0,\delta t_o]\cap[0,T_{\max})$. From the latter, we deduce that $T_{\max}>\delta t_o$. Indeed, otherwise $\chi_i(t)$ would remain in a compact subset of $D$ for $t\in [0,T_{\max})$, contradicting maximality of $[0,T_{\max})$. Thus, we conclude that \eqref{Tmax:lower:bound} is satisfied.
\end{proof}

By exploiting Lemma \ref{lemma:Tmax}, it will be shown in the next section that $T_{\max}$ is always greater than the maximum possible selection of the time step $\delta t$ for a well posed discretization. The latter in conjunction with the result of Lemma \ref{lemma:Txi0wi} below enables us to prove that in this case the control law $k_{i,\bf{l}_i,3}(\cdot)$ and hence also $k_{i,\bf{l}_i}(\cdot)$ are well defined on $[0,\delta t]$.

\begin{lemma} \label{lemma:Txi0wi}
Consider a cell decomposition $\S$ of $D$, a time step $\delta t$ and select an agent $i\in\N$ and a cell configuration $\bf{l}_i=(l_i,l_{j_1},\ldots,l_{j_{N_i}})$ of $i$. Also, consider a tuple of reference points $(x_{i,G},\bf{x}_{j,G})$ as in \eqref{reference:point} and the control law  $k_{i,\bf{l}_i}(\cdot)$ in \eqref{feedback:ki}. We assume that  $k_{i,\bf{l}_i}(\cdot)$ satisfies Properties (P1) and (P2) of Definition \ref{control:class}, and that the right endpoint $T_{\max}$ of the interval where the reference trajectory \eqref{reference:solution} is defined,  satisfies $T_{\max}>\delta t$. Then, for all $x_{i0}\in S_{l_i}$ and $w_i\in W$, the time $T_i(x_{i0},w_i)$ satisfies $T_i(x_{i0},w_i)>\delta t$, which implies that $k_{i,\bf{l}_i}(\cdot)$ also satisfies Property (P3) of  Definition \ref{control:class}.
\end{lemma}

\begin{proof}
Indeed, let  $x_{i0}\in S_{l_i}$ and $w_i\in W$. By defining
\begin{equation} \label{tilde:xi:plan}
\hat{x}_i(t) :=\bar{x}_i(t)+\left(1-\frac{t}{\delta t}\right)(x_{i0}-x_{i,G}), t\in [0,T_{\max}),
\end{equation}

\noindent with $\bar{x}_i(t)=\chi_i(t)+tw_i$ as given in \eqref{bar:xi:plan}, and taking into account the definition of $T(x_{i0},w_i)$ in \eqref{time:T:plan}, we want to show that $\hat{x}_i(\cdot)$ remains in $D$ for more than time $\delta t$. By virtue of our assumption that $T_{\max}>\delta t$, the latter is meaningful to verify and implies that  $T(x_{i0},w_i)>\delta t$. We next show that $\hat{x}_i(\cdot)$ coincides on a suitable time interval with the $i$-th component of the solution of \eqref{single:integrator}-\eqref{general:feedback:law} by choosing appropriate initial conditions and feedback laws that satisfy (P1) and (P2).

Let $x_{i0}\in S_{l_i}$, $w_i\in W$, consider an arbitrary initial cell configuration $\bf{l}$ with ${\rm pr}_{i}(\bf{l})=\bf{l}_i$, $\bf{l}=(l_1,\ldots,l_N)$, and assign the feedback law $k_{i,{\rm pr}_{i}(\bf{l})}=k_{i,\bf{l}_i}$ (as the latter is given by \eqref{feedback:ki}) to $i$ and the feedback laws $k_{\ell,{\rm pr}_{\ell}(\bf{l})}:=0$ to the rest of the agents $\ell\in\N\setminus\{i\}$. It also follows from the assumptions of the lemma for $i$, and trivially for the other agents, that the feedback laws satisfy Properties (P1) and (P2). Thus, we can use the result of Proposition~\ref{completeness:result}(i). By selecting an initial condition $x(0)\in D^{N}$ with $x_i(0)=x_{i0}$ and $x_{j_{m}}(0)\in S_{l_{m}},m=1,\ldots,N_i$, and recalling that $w_i\in W$, we get from Proposition~\ref{completeness:result}(i) that the $i$-th component of the solution satisfies
\begin{align}
x_{i}(t) & \in D,\forall t\in [0,\tilde{T}), \tilde{T}:=\min\{\delta t,T(x_{i0},w_i)\} \label{invariance:solution:xi:plan} \\
\lim_{t\to\tilde{T}^-}x_{i}(t) & \in D. \label{limit:solution:xi:plan}
\end{align}

\noindent We proceed by showing that $x_i(t)=\hat{x}_i(t)$, for all $t\in[0,\tilde{T})$, with $\tilde{T}$ as given in \eqref{invariance:solution:xi:plan}, or equivalently, that
\begin{equation} \label{xi:eq:tildexi:solution1:plan}
x_{i}(t)=\chi_i(t)+tw_i+\left(1-\frac{t}{\delta t}\right)(x_{i0}-x_{i,G}),\forall t\in [0,\tilde{T})
\end{equation}

\noindent Indeed, from \eqref{bar:xi:plan}, \eqref{reference:solution}, \eqref{single:integrator}-\eqref{general:feedback:law}, \eqref{feedback:ki}, \eqref{feedback:ki1} and \eqref{averaged:dynamics} we have that
\begin{align*}
\dot{\bar{x}}_{i}(t) & =F_{i,\bf{l}_i}(\chi_{i}(t))+w_i, \\
\dot{x}_{i}(t) & =F_{i,\bf{l}_i}(x_{i}(t))+k_{i,\bf{l}_i,2}(x_{i0})+k_{i,\bf{l}_i,3}(t;x_{i0},w_i).
\end{align*}

\noindent  By recalling that $\bar{x}_{i}(0)=x_{i,G}$, $x_{i}(0)=x_{i0}$ and that due to \eqref{time:T:plan} and \eqref{invariance:solution:xi:plan} it holds $\tilde{T}\le T(x_{i0},w_i)\le T_{\max}$, and thus $\chi_i(\cdot)$, $x_i(\cdot)$ and $k_{i,\bf{l}_i,3}(\cdot)$ are well defined on $[0,\tilde{T})$, it follows from \eqref{feedback:ki2}, \eqref{feedback:ki3:plan} and \eqref{bar:xi:plan} that 
\begin{align*}
x_{i}(t) -\bar{x}_{i}(t) & =x_{i0}-x_{i,G}+\int_{0}^{t}[F_{i,\bf{l}_i}(x_{i}(s))-F_{i,\bf{l}_i}(\chi_{i}(s)) \\
& +k_{i,\bf{l}_i,2}(x_{i0})+k_{i,\bf{l}_i,3}(s;x_{i0},w_i)-w_i]ds  \\
& =\left(1-\frac{t}{\delta t}\right)(x_{i0}-x_{i,G})+\int_{0}^{t}[F_{i,\bf{l}_i}(x_i(s)) \\
& -\left.F_{i,\bf{l}_i}\left(\bar{x}_{i}(s)+\left(1-\frac{s}{\delta t}\right)(x_{i0}-x_{i,G})\right)\right]ds,\forall t\in [0,\tilde{T}). 
\end{align*}


\noindent Hence, we get from \eqref{tildefi:Lipschitz:const} that for all $t\in [0,\tilde{T})$ it holds $|x_{i}(t)-\bar{x}_i(t)-\left(1-\frac{t}{\delta t}\right)(x_{i0}-x_{i,G})|\le \int_{0}^{t}L_2\left|x_{i}(s)-\bar{x}_{i}(s)-\left(1-\frac{s}{\delta t}\right)(x_{i0}-x_{i,G})\right|ds$. Application of the Gronwall Lemma, \eqref{bar:xi:plan}, and the fact that $\tilde{T}\le T_{\max}$, imply that \eqref{xi:eq:tildexi:solution1:plan} holds.

We are now in position to prove that $T(x_{i0},w_i)>\delta t$. Indeed, suppose on the contrary that $T(x_{i0},w_i)\le\delta t$, which by virtue of the assumption that $T_{\max}>\delta t$, and \eqref{invariance:solution:xi:plan}, implies that $T(x_{i0},w_i)<T_{\max}$ and $\tilde{T}= T(x_{i0},w_i)$. From the latter, together with \eqref{tilde:xi:plan}, \eqref{invariance:solution:xi:plan}, \eqref{limit:solution:xi:plan} and continuity of $\hat{x}_i(\cdot)$, we get that $\hat{x}_i(T(x_{i0},w_i))=\lim_{t\to T(x_{i0},w_i)^-}\hat{x}_i(t)=\lim_{t\to T(x_{i0},w_i)^-}x_i(t)\in D$. Hence, from \eqref{tilde:xi:plan}, the deduction that $T(x_{i0},w_i)<T_{\max}$ and continuity of $\hat{x}_i(\cdot)$, it follows that there exists $\varepsilon>0$ such that $\hat{x}_i(t)\in D$ for $t\in [T(x_{i0},w_i),T(x_{i0},w_i)+\varepsilon)$, which contradicts \eqref{time:T:plan}. Thus we conclude that $T(x_{i0},w_i)>\delta t$, which establishes validity of (P3).
\end{proof}

\section{Well Posed Space-Time Discretizations with Motion Planning Capabilities}

In this section, we exploit the controllers introduced in~\eqref{feedback:ki} to provide sufficient conditions for well posed space-time discretizations. By exploiting the result of Proposition \ref{discrete:transitions:result} this  framework can be applied for motion planning, by specifying different possibilities for transitions for each agent through modifying its controller. Consider again the system \eqref{single:integrator}-\eqref{general:feedback:law}, a cell decomposition $\mathcal{S}=\{S_{l}\}_{l\in\mathcal{I}}$ of $D$ and a time step $\delta t$. In addition, consider the least upper bound on the diameter of the cells in $\mathcal{S}$, namely,

\begin{equation}\label{dmax:dfn}
d_{\max}:=\sup\{{\rm diam}(S_{l}),l\in\mathcal{I}\},
\end{equation}

\noindent which due to Definition \ref{cell:decomposition} is well defined. We will call $d_{\max}$ the \textbf{diameter} of the cell decomposition. \textit{Our goal is to determine sufficient conditions relating the Lipschitz constants $L_{1}$, $L_{2}$, the bounds $M$, $v_{\max}$ for the system's dynamics, as well as the space and time scales $d_{\max}$ and $\delta t$ of the space-time discretization $\mathcal{S}-\delta t$, which guarantee that $\mathcal{S}-\delta t$ is well posed}. According to Definition \ref{well:posed:discretization}, establishment of a well posed discretization is based on the selection of appropriate feedback laws which guarantee well posed transitions for all agents and their possible cell configurations. For each agent $i\in\N$ and cell configuration $\bf{l}_i=(l_i,l_{j_1},\ldots,l_{j_{N_i}})$ of $i$ let $(x_{i,G},\bf{x}_{j,G})$ be a reference point as in \eqref{reference:point}. We consider the family of feedback laws given in~\eqref{feedback:ki1}, \eqref{feedback:ki2}, \eqref{feedback:ki3:plan}, and parameterized by $x_{i0}\in S_{l_i}$ and $w_i\in W$. The function $F_{i,\bf{l}_i}(\cdot)$ is given in~\eqref{averaged:dynamics}, and $\chi_i(\cdot)$ is the reference solution of the initial value problem~\eqref{reference:solution}, defined on $[0,T_{\max})$. Recall that the parameter $\lambda$ in \eqref{set:W} provides the portion of the free input that is exploited for planning. Thus, it can be regarded as a measure for the degree of control freedom that is chosen for the abstraction. In the following results, we also introduce an additional parameter $\mu$ which provides a lower bound on the minimum number ($\ge 1$) of discrete transitions that are possible from each initial cell, as will be clarified in the corollary at the end of the section. Before proceeding to the desired sufficient conditions for well posed discretizations and their reachability properties, we prove the auxiliary Propositions \ref{proposition:bounds} and \ref{proposition:property:P}. Proposition \ref{proposition:bounds} below provides bounds on the hybrid control laws  $k_{i,\bf{l}_i}(\cdot)$ in \eqref{feedback:ki}.

\begin{prop} \label{proposition:bounds}
Consider a cell decomposition $\S$ of $D$ with diameter $d_{\max}$ and a time step $\delta t$. Also, for  each agent $i\in\N$ and cell configuration $\bf{l}_i=(l_i,l_{j_1},\ldots,l_{j_{N_i}})$ of $i$ let  $(x_{i,G},\bf{x}_{j,G})$ be a reference point as in \eqref{reference:point} and consider the feedback law $k_{i,\bf{l}_i}(\cdot)$ in \eqref{feedback:ki}. Then its components $k_{i,\bf{l}_i,1}(\cdot)$, $k_{i,\bf{l}_i,2}(\cdot)$ and $k_{i,\bf{l}_i,3}(\cdot)$ as given in \eqref{feedback:ki1},  \eqref{feedback:ki2} and \eqref{feedback:ki3:plan}, respectively, satisfy the bounds
\begin{align}
& |k_{i,\bf{l}_i,1}(x_{i},\bf{x}_j)| \le  L_{1}\sqrt{N_i}(R_{\max}+d_{\max}),\forall x_i \in D, \nonumber \\
& x_{j_{m}}\in (S_{l_{m}}+B(R_{\max}))\cap D, m=1,\ldots,N_i, \label{ki1:bound} \\
& |k_{i,\bf{l}_i,2}(x_{i0})|\le \frac{1}{\delta t}d_{\max},\forall x_{i0}\in S_{l_i}, \label{ki2:bound} \\
&  |k_{i,\bf{l}_i,3}(t;x_{i0},w_i)| \le L_{2}(\delta t\lambda v_{\max}+d_{\max})+\lambda v_{\max}, \nonumber \\
& \forall t\in[0,\delta t]\cap[0,T(x_{i0},w_i)),x_{i0}\in S_{l_i},w_i\in W. \label{ki3:bound:plan}
\end{align}

\noindent with $R_{\max}$ as given in \eqref{Rmax}.
\end{prop}

\begin{proof}
\noindent Indeed, in order to show \eqref{ki1:bound}  let $\bf{x}_j\in D^{N_i}$ satisfying $x_{j_{m}}\in (S_{l_{m}}+B(R_{\max})) \cap D, m=1,\ldots,N_i$. Then, for each $m=1,\ldots,N_i$ there exists $\hat{x}_{j_{m}}$ with $\hat{x}_{j_{m}}\in S_{l_{m}}$ and $|\hat{x}_{j_{m}}-x_{j_{m}}|\le R_{\max}$. Hence, from the latter together with \eqref{feedback:ki1} and  \eqref{dynamics:bound1}, we get 
\begin{align*}
|k_{i,\bf{l}_i,1}(x_{i},\bf{x}_j)| & \le L_{1}|(x_{j_{1}}-x_{j_{1},G},\ldots,x_{j_{N_i}}-x_{j_{N_i},G})| \\
& \le L_{1}\left(\sum_{m=1}^{N_i}(|x_{j_{m}}-\hat{x}_{j_{m}}|+|\hat{x}_{j_{m}}-x_{j_{m},G}|)^{2}\right)^{\frac{1}{2}} \\
& \le L_{1}\left(\sum_{m=1}^{N_i}(R_{\max}+d_{\max})^{2}\right)^{\frac{1}{2}}= L_{1}\sqrt{N_i}(R_{\max}+d_{\max}),
\end{align*}

\noindent which establishes \eqref{ki1:bound}. Furthermore, by recalling that $x_{i,G}\in S_{l_i}$, it follows directly from \eqref{feedback:ki2} that $|k_{i,\bf{l}_i,2}(x_{i0})|=\frac{1}{\delta t}|x_{i0}-x_{i,G}|$ and hence, that \eqref{ki2:bound} is satisfied. Finally, for $k_{i,\bf{l}_i,3}(\cdot)$ we get from \eqref{feedback:ki3:plan}  and \eqref{tildefi:Lipschitz:const} that 
$$
|k_{i,\bf{l}_i,3}(t;x_{i0},w_i)|\le L_2\left|\left(\chi_i(t)+tw_i+\left(1-\tfrac{t}{\delta t}\right)(x_{i0}-x_{i,G})\right)-\chi_i(t)\right|+|w_i|,
$$

\noindent which due to \eqref{set:W} implies validity of \eqref{ki3:bound:plan}.
\end{proof}

Based on the result of Proposition \ref{proposition:bounds} we next provide  conditions on $d_{\max}$ and $\delta t$ which guarantee that the feedback laws $k_{i,\bf{l}_i}(\cdot)$ satisfy Property (P). Additionally it is shown that the radius $r$ introduced in \eqref{distance:r} satisfies a design requirement which is related later in Corollary~\ref{corollary:num:transitions} to a lower bound on the number of possible transitions through the parameter $\mu$.

\begin{prop}\label{proposition:property:P}
Consider a cell decomposition $\S$ of $D$ with diameter $d_{\max}$, a time step $\delta t$, the parameters $\lambda\in(0,1)$, $\mu>0$ and define
\begin{equation} \label{constant:tildeL:prime}
L:=\max\{3L_{2}+4L_{1}\sqrt{N_i},i\in\N\},
\end{equation}

\noindent with $L_1$ and $L_2$ as given in \eqref{dynamics:bound1} and \eqref{dynamics:bound2}. We assume that $\lambda$, $\mu$, $d_{\max}$ and $\delta t$ satisfy the following restrictions, as provided by the three cases below:

\noindent \textbf{Case I.} $0\le\mu\le\frac{2\lambda}{1-\lambda}$.
\begin{align}
d_{\max}\in & \left(0,\frac{(1-\lambda)^2 v_{\max}^{2}}{4ML}\right], \label{dmax:interval:caseA} \\
\delta t \in & \left[\frac{(1-\lambda)v_{\max}-\sqrt{(1-\lambda)^2 v_{\max}^{2}-4MLd_{\max}}}{2ML},\frac{(1-\lambda)v_{\max}+\sqrt{(1-\lambda)^2 v_{\max}^{2}-4MLd_{\max}}}{2ML}\right]. \label{deltat:interval:caseA}
\end{align}

\noindent \textbf{Case II.} $\frac{2\lambda}{1-\lambda}<\mu<\frac{4\lambda}{1-\lambda}$.
\begin{align}
& d_{\max} \in \left(0,\frac{2(\lambda(1-\lambda)\mu-2\lambda^2)v_{\max}^{2}}{\mu^2ML}\right], \label{dmax:interval:caseBi} \\
& \delta t \in \left[\frac{\mu}{2\lambda v_{\max}}d_{\max},\frac{(1-\lambda)v_{\max}+\sqrt{(1-\lambda)^2 v_{\max}^{2}-4MLd_{\max}}}{2ML}\right], \label{deltat:interval:caseBi}
\end{align}

\noindent or
\begin{equation}   \label{dmax:interval:caseBii}
d_{\max}\in\left (\frac{2(\lambda(1-\lambda)\mu-2\lambda^2)v_{\max}^{2}}{\mu^2ML},\frac{(1-\lambda)^2 v_{\max}^{2}}{4ML}\right]
\end{equation}

\noindent and $\delta t$ satisfies \eqref{deltat:interval:caseA}.

\noindent \textbf{Case III.} $\mu\ge\frac{4\lambda}{1-\lambda}$. $d_{\max}$ and $\delta t$ satisfy \eqref{dmax:interval:caseBi} and \eqref{deltat:interval:caseBi}, respectively.

\noindent Then, the intervals in Cases I, II, III are well defined, and for each agent $i\in\N$, cell configuration $\bf{l}_i=(l_i,l_{j_1},\ldots,l_{j_{N_i}})$ of $i$ and reference point $(x_{i,G},\bf{x}_{j,G})$ as in \eqref{reference:point} the solution $\chi_i(t)$ of \eqref{reference:solution} is defined and remains in $D$ for all $t\in[0,\delta t]$. In addition the  feedback law $k_{i,\bf{l}_i}(\cdot)$ in \eqref{feedback:ki} satisfies property (P) and the distance $r$ as defined in \eqref{distance:r} satisfies the design requirement
\begin{equation} \label{r:design:req}
r\ge\frac{\mu}{2}d_{\max}
\end{equation}
\end{prop}

\begin{proof}
The proof of the fact that the intervals in Cases I, II, III are well defined is provided in the Appendix. Also, it follows from Lemma \ref{lemma:Tmax} that for any reference point $(x_{i,G},\bf{x}_{j,G})$ as in \eqref{reference:point} the solution $\chi_i(t)$ of \eqref{reference:solution} is defined and remains in $D$ for all $t\in[0,T_{\max})$, and by virtue of \eqref{Tmax:lower:bound} and the assumed bounds on $\delta t$ in Cases I, II, III, that
\begin{equation} \label{Tmax:gt:deltat:plan}
T_{\max}>\delta t,
\end{equation}

\noindent establishing thus that $\chi_i(t)\in D$ for all $t\in[0,\delta t]$. We break the subsequent proof in the following steps.

\noindent \textbf{STEP 1: Verification of Properties (P1) and (P2) for the feedback law~\eqref{feedback:ki} for $d_{\max}-\delta t$ as given by Cases I, II, III, in conjunction with the design requirement \eqref{r:design:req}.} In this step we prove that the proposed feedback law~\eqref{feedback:ki} satisfies Properties (P1) and (P2). Verification of (P1) is straightforward. Thus, we proceed to show that \eqref{feedback:k:bound} holds, which implies (P2), and simultaneously, that \eqref{r:design:req} is fulfilled. By taking into account \eqref{feedback:ki} and the result of Proposition \ref{proposition:bounds}, namely, \eqref{ki1:bound}, \eqref{ki2:bound} and \eqref{ki3:bound:plan}, we need to prove that
\begin{align}
L_{1}\sqrt{N_i}(R_{\max}+d_{\max}) & +\frac{1}{\delta t}d_{\max} \nonumber \\
+L_{2}(\delta t\lambda v_{\max} & +d_{\max})+\lambda v_{\max}\le v_{\max}. \label{condition:dmax:vmax:plan}
\end{align}

\noindent By recalling \eqref{Rmax}, \eqref{vmax:vs:M} and the fact that $\lambda\in (0,1)$ we get that $\delta t\lambda v_{\max}\le \frac{R_{\max}}{2}$. Also, from the fact that $d_{\max}$ and $\delta t$ are selected according to the Cases I, II, III, it follows from elementary calculations which are provided in the Appendix that 
\begin{equation} \label{Rmax:vs:dmax}
d_{\max}\le R_{\max}.
\end{equation}

\noindent Hence, it suffices instead of \eqref{condition:dmax:vmax:plan} to show that $(2L_{1}\sqrt{N_i}+\frac{3}{2}L_{2})R_{\max}+\frac{1}{\delta t}d_{\max}\le (1-\lambda) v_{\max}$, which by virtue of \eqref{Rmax} is equivalent to
\begin{equation}
(M+v_{\max})(2L_{1}\sqrt{N_i}+\frac{3}{2}L_{2})\delta t^{2}-(1-\lambda)v_{\max}\delta t+d_{\max}\le 0. \label{quadratic:condition:plan}
\end{equation}

\noindent By taking into account \eqref{vmax:vs:M}, it suffices instead of \eqref{quadratic:condition:plan} to show that $M(3L_{2}+4L_{1}\sqrt{N_i})\delta t^{2}-(1-\lambda)v_{\max}\delta t+d_{\max}\le 0$ which by virtue of \eqref{constant:tildeL:prime} follows from
\begin{equation} \label{quadratic:condition2:plan}
ML\delta t^{2}-(1-\lambda)v_{\max}\delta t+d_{\max}\le 0.
\end{equation}

\noindent In order for the above equation to have real roots, it is required that
\begin{equation} \label{dmax:bound:plan}
(1-\lambda)^2 v_{\max}^{2}-4MLd_{\max}\ge 0\iff d_{\max}\le\frac{(1-\lambda)^2 v_{\max}^{2}}{4ML}.
\end{equation}

\noindent Hence, by collecting our requirements \eqref{dmax:bound:plan}, \eqref{quadratic:condition2:plan}, \eqref{Rmax:vs:dmax} and \eqref{r:design:req}  together with the fact that $d_{\max}>0$ we have
\begin{align}
& 0<d_{\max}\le\frac{(1-\lambda)^{2}v_{\max}^{2}}{4ML}, \label{dmax:bound2:plan} \\
& \frac{(1-\lambda)v_{\max}-\sqrt{(1-\lambda)^2v_{\max}^{2}-4MLd_{\max}}}{2ML}\le\delta t\le\frac{(1-\lambda)v_{\max}+\sqrt{(1-\lambda)^2v_{\max}^{2}-4MLd_{\max}}}{2ML}, \label{deltat:bound2:plan} \\
&\frac{1}{M+v_{\max}}d_{\max}\le\delta t, \label{deltat:vs:dmax:plan} \\
&\frac{\mu}{2\lambda v_{\max}} d_{\max}\le \delta t. \label{mu:slope:requirement:plan}
\end{align}

\noindent We can then show that for all Cases I, II and III as in the statement of the proposition the above requirements are satisfied and hence, that (P2) holds. The proof of this fact can be found in the Appendix.

\noindent \textbf{STEP 2: Verification of Property (P3).} In order to show (P3), it suffices to prove that for the given selection of $\lambda\in(0,1)$, $\mu>0$, $d_{\max}$ and $\delta t$ as provided by Cases I, II, III, the agent $i$ and the cell configuration $\bf{l}_i$ it holds $T(x_{i0},w_i)>\delta t$, for all $x_{i0}\in S_{l_i}$ and $w_i\in W$. The latter is a direct consequence of \eqref{Tmax:gt:deltat:plan} and Lemma \ref{lemma:Txi0wi}.
\end{proof}

We are now in position to state our main result on sufficient conditions for well posed abstractions.

\begin{thm}\label{discretizations:for:planning}
Consider a cell decomposition $\S$ of $D$ with diameter $d_{\max}$, a time step $\delta t$, the parameters $\lambda\in(0,1)$, $\mu>0$ and assume that $\lambda$, $\mu$, $d_{\max}$ and $\delta t$ satisfy the restrictions of Proposition~\ref{proposition:property:P}. Then, for each agent $i\in\N$ and cell configuration $\bf{l}_i=(l_i,l_{j_1},\ldots,l_{j_{N_i}})$ of $i$ we have ${\rm Post}_i(l_i;\bf{l}_i)\ne\emptyset$, namely, the space-time discretization is well posed for the multi-agent system \eqref{single:integrator}-\eqref{general:feedback:law}. In particular, for any tuple of reference points $(x_{i,G},\bf{x}_{j,G})$ as in \eqref{reference:point} and corresponding reference trajectory $\chi_i(\cdot)$ of $i$ as given by \eqref{reference:solution} it holds
\begin{align}
B(\chi_i(\delta t);r) & \subset D, \label{reachable:states:inD} \\
{\rm Post}_i(l_i;\bf{l}_i) & \supset\{l\in\I:S_l\cap B(\chi_i(\delta t);r)\ne\emptyset\}, \label{planning:condition}
\end{align}

\noindent where $r$ is defined in \eqref{distance:r}.
\end{thm}

\begin{proof}
For the proof, pick $i\in\N$, $\bf{l}_i=(l_i,l_{j_1},\ldots,l_{j_{N_i}})$, $(x_{i,G},\bf{x}_{j,G})$ as in \eqref{reference:point} and notice that by virtue of Proposition~\ref{proposition:property:P}, the reference trajectory $\chi_i(\cdot)$ is well defined on $[0,\delta t]$. In addition, consider the control law $k_{i,\bf{l}_i}(\cdot)$ in \eqref{feedback:ki}. Then, it follows again from Proposition~\ref{proposition:property:P} that the latter satisfies Property~(P). Next, for each $x\in B(\chi_i(\delta t);r)$ define
\begin{equation} \label{vector:wi}
w_i(=w_i(x)):=\frac{x-\chi_i(\delta t)}{\delta t}.
\end{equation}

\noindent Then, we get from \eqref{distance:r} that $|w_i|\le \frac{r}{\delta t}=\lambda v_{\max}$ and hence, by virtue of \eqref{set:W} that 
\begin{equation} \label{wi:is:inW}
w_i(=w_i(x))\in W,\forall x\in B(\chi_i(\delta t);r).
\end{equation}

\noindent In order to prove the theorem, we need to verify that \eqref{reachable:states:inD} and~\eqref{planning:condition} are fulfilled.

\noindent \bf{Proof of \eqref{reachable:states:inD}.} In order to show \eqref{reachable:states:inD}, pick $x\in B(\chi_i(\delta t);r)$, $w_i$ as in \eqref{vector:wi} and recall that the control law $k_{i,\bf{l}_i}(\cdot)$  satisfies Property (P). Then we have from \eqref{wi:is:inW} that $w_i\in W$ and thus, we get from Property (P3) applied with $x_{i0}=x_{i,G}$ and the selected parameter $w_i$ that $T(x_{i,G},w_i)>\delta t$. From the latter and \eqref{time:T:plan} we obtain that $\chi_i(\delta t)+w_i\delta t\in D$, which by  virtue of \eqref{vector:wi} implies that $x\in D$ and establishes validity of \eqref{reachable:states:inD}.

\noindent \bf{Proof of \eqref{planning:condition}.} For the verification  of \eqref{planning:condition} it suffices to prove the following claim.

\noindent \bf{Claim II.} Consider the control law $k_{i,\bf{l}_i}(\cdot)$ above and pick any  $w_i\in W$. Then, for any initial cell configuration $\bf{l}$ with ${\rm pr}_i(\bf{l})=\bf{l}_i$, $\bf{l}=(l_{1},\ldots,l_{N})$, $\ell\in\N\setminus \{{i}\}$ and selection of feedback laws in \eqref{feedback:for:others} which satisfy (P) the following hold. The solution of the closed-loop system \eqref{single:integrator}-\eqref{general:feedback:law}, \eqref{feedback:ki}, \eqref{feedback:for:i}, \eqref{feedback:for:others}, with the selected parameter $w_i$ for $k_{i,\bf{l}_i}(\cdot)$, is well defined on $[0,\delta t]$ and satisfies
\begin{equation} \label{control:to:same:point}
x_i(\delta t)(:=x_i(\delta t,x(0)))=\chi_i(\delta t)+\delta tw_i=\bar{x}_i(\delta t),
\end{equation}

\noindent for all  $x(0)\in D^N$ with $x_{m0}\in S_{l_m}$, $m\in\N$ and  $w_{m}\in W$, $m\in\N\setminus \{{i}\}$, with the last equality in \eqref{control:to:same:point} being a consequence of \eqref{bar:xi:plan} (see also Fig. 2 in Section 3).$\quad\triangleleft$

Indeed, let any $l\in\I$ such that
\begin{equation} \label{Sl:cap:ball}
S_l\cap  B(\chi_i(\delta t);r)\ne\emptyset.
\end{equation}

\noindent In order to show that $l\in {\rm Post}_i(l_i;\bf{l}_i)$, i.e., that the transition $l_i\overset{\bf{l}_i}{\longrightarrow}l$ is well posed, it suffices according to Definition  \ref{well:posed:discretization}(i) to verify that there exists $w_i\in W$ such that Condition (C) holds with the control law $k_{i,\bf{l}_i}(\cdot)$ above. By exploiting \eqref{Sl:cap:ball}, we pick $x\in S_l\cap B(\chi_i(\delta t);r)$ and $w_i$ in \eqref{vector:wi} as the parameter for $k_{i,\bf{l}_i}(\cdot)$, which by virtue of \eqref{wi:is:inW} satisfies $w_i\in W$. Thus, it follows that the conclusion of Claim II is fulfilled with \eqref{vector:wi} and \eqref{control:to:same:point} implying that $x_i(\delta t)=x\in S_l$. Hence, Condition (C) is satisfied and we conclude that  $l_i\overset{\bf{l}_i}{\longrightarrow}l$ is well posed. It thus remains to verify Claim II.

\noindent \bf{Proof of Claim II.}  Let $x$, $k_{i,\bf{l}_i}(\cdot)$ and $w_i(=w_i(x))$ as in the statement of Claim II. We first note that due to Proposition~\ref{completeness:result}(iia), the solution of the closed-loop system is defined and remains in $D^{N}$ on the whole interval $[0,\delta t]$. In order to show that $x_i(\delta t)=\bar{x}_i(\delta t)$, we show that $x_i(\cdot)$ is an appropriate modification of the trajectory $\bar{x}_i(\cdot)$. In particular, it holds $x_{i}(t)=\bar{x}_i(t)+\left(1-\frac{t}{\delta t}\right)(x_{i0}-x_{i,G}),\forall t\in [0,\delta t]$, which implies the desired result. The proof of the latter is based precisely on the same arguments used for the proof of \eqref{xi:eq:tildexi:solution1:plan} in Lemma \ref{lemma:Txi0wi} and is therefore omitted. Hence, we conclude that $x_{i}(\delta t)=\bar{x}_{i}(\delta t)$ and the proof is complete.
\end{proof}

\begin{rem}\label{remark:refernece:point:depend}
Notice that the reachable cells provided by the left hand side of \eqref{planning:condition} are a subset of the reachable cells from the specific cell configuration of $i$. In addition, these cells depend on the reference points $(x_{i,G},\bf{x}_{j,G})$, since a different selection will correspond to a different control law $k_{i,\bf{l}_i}$ and will lead in  principle to another reference trajectory, thus, modifying the center of the ball $B(\chi_i(\delta t);r)$. $\triangleleft$  
\end{rem}

Based on the above remark, we next show that by selecting for each agent $i$ and cell configuration $\bf{l}_i$ a tuple of reference points $(x_{i,G},\bf{x}_{j,G})$ and the control law \eqref{feedback:ki}, the right hand side of \eqref{planning:condition} provides the set $\cup_{[w_i]\in 2^W}{\rm Post}_i^c(l_i;(\bf{l}_i,[w_i]))$ of the controlled transition systems $TS_i^c$ in Definition~\ref{controlled:TS}.

\begin{corollary} \label{corollary:controlled:transitions}
Consider a space-time discretization $\S-\delta t$ satisfying the hypotheses of Theorem \ref{discretizations:for:planning}. Also, select for each agent $i$ and cell configuration $\bf{l}_i$, a tuple of reference points  $(x_{i,G},\bf{x}_{j,G})$ as in \eqref{reference:point}, the control law $k_{i,\bf{l}_i}(\cdot)$ in \eqref{feedback:ki} and consider the corresponding reference trajectory $\chi_i(\cdot)$ given by \eqref{reference:solution}. Then, for each $i$,  $\bf{l}_i=(l_i,l_{j_1},\ldots,l_{j_{N_i}})$ and $l\in\I$ the actions $[w_i]_{(\bf{l}_i,l)}$ in \eqref{wi:class} for the specification of the controlled transition system $TS_i^c$ are given as 
\begin{equation} \label{wi:class:Slcapball}
[w_i]_{(\bf{l}_i,l)}=\left\lbrace\frac{x-\chi_i(\delta t)}{\delta t}:x\in S_l\cap B(\chi_i(\delta t);r)\right\rbrace,
\end{equation}

\noindent with $r$ as in \eqref{distance:r}, and it holds
\begin{align}
\cup_{[w_i]\in 2^W} & {\rm Post}_i^c(l_i;(\bf{l}_i,[w_i])) \nonumber\\
& =\{l\in\I:S_l\cap B(\chi_i(\delta t);r)\ne\emptyset\}. \label{Post:controlled}
\end{align}
\end{corollary}

\begin{proof}
\bf{Proof of \eqref{wi:class:Slcapball}.} We first consider the case where $S_l\cap B(\chi_i(\delta t);r)\ne\emptyset$ and show that $[w_i]_{(\bf{l}_i,l)}\supset\{\frac{x-\chi_i(\delta t)}{\delta t}:x\in S_l\cap B(\chi_i(\delta t);r)\}$. Indeed, let $x\in S_l\cap B(\chi_i(\delta t);r)$ and $w_i$ as given by \eqref{vector:wi}, which by virtue of \eqref{wi:is:inW} satisfies $w_i\in W$. Thus, from Claim II in the proof of Theorem \ref{discretizations:for:planning} we obtain that $\bf{l}_i$, $k_{i,\bf{l}_i}(\cdot)$, $w_i$ and $l$ satisfy Condition (C), since \eqref{vector:wi} and \eqref{control:to:same:point} imply that $x_i(\delta t)=x\in S_l$. Thus, we obtain  from \eqref{wi:class} that $\frac{x-\chi_i(\delta t)}{\delta t}\in [w_i]_{(\bf{l}_i,l)}$. In order to prove the reverse inclusion, let $w_i\in [w_i]_{(\bf{l}_i,l)}\subset W$ and assume on the contrary that
\begin{equation} \label{wi:contradiction}
w_i\notin \left\lbrace\frac{x-\chi_i(\delta t)}{\delta t}:x\in S_l\cap B(\chi_i(\delta t);r)\right\rbrace.
\end{equation} 

\noindent Let $x=\chi_i(\delta t)+\delta tw_i$ and notice that due to \eqref{set:W} and \eqref{distance:r} it holds $x\in B(\chi_i(\delta t);r)$. In addition, by exploiting Claim II applied with this selection of $w_i$, we obtain from \eqref{control:to:same:point} that $x_i(\delta t)=\chi_i(\delta t)+\delta tw_i=x$. Also, since $w_i\in [w_i]_{(\bf{l}_i,l)}$ we get from \eqref{wi:class}, namely, the fact that $\bf{l}_i$, $k_{i,\bf{l}_i}(\cdot)$, $w_i$, $l$ satisfy Condition (C), that $x_i(\delta t)\in S_l$. Thus, $x\in S_l\cap B(\chi_i(\delta t);r)$ which contradicts \eqref{wi:contradiction}.

Finally, in order to verify \eqref{wi:class:Slcapball} for the general case we need to show that $[w_i]_{(\bf{l}_i,l)}=\emptyset$ when $S_l\cap B(\chi_i(\delta t);r)=\emptyset$. Notice that the latter is equivalently written as 
\begin{equation} \label{x:outside:ball}
|x-\chi_i(\delta t)|>r,\forall x\in S_l.
\end{equation}  

\noindent Hence, suppose on the contrary that \eqref{x:outside:ball} holds and that $[w_i]_{(\bf{l}_i,l)}\ne\emptyset$. Then, by picking any $w_i\in [w_i]_{(\bf{l}_i,l)}\subset W$, we obtain from \eqref{control:to:same:point} in Claim II that $x_i(\delta t)=\chi_i(\delta t)+\delta tw_i$. On the other hand, since $w_i\in  [w_i]_{(\bf{l}_i,l)}$, we get from \eqref{wi:class} that $x_i(\delta t)\in S_l$. Hence, it follows that $\chi_i(\delta t)+\delta tw_i\in S_l$ which by virtue of \eqref{x:outside:ball} implies that $|w_i|>\frac{r}{\delta t}$. Thus, we obtain from \eqref{set:W} and \eqref{distance:r} that $w_i\notin W$, which is a contradiction.

\noindent \bf{Proof of \eqref{Post:controlled}.} This follows directly from \eqref{wi:class:Slcapball} and the definition of transitions in $TS_i^c$ as provided by Definition~\ref{controlled:TS}. 
\end{proof}

\begin{rem} \label{remark:deterministic}
It is noted that if we select $W$ in \eqref{set:W} as the open ball ${\rm int}(B(\lambda v_{\max}))$ we will obtain in \eqref{planning:condition} the cells which have nonempty intersection with the open ball ${\rm int}(B(\chi_i(\delta t);r))$. In this case, it follows from the properties of a cell decomposition that each transition system $TS_i^c$ is deterministic. The latter relies on the requirement that the interiors of the cells in the decomposition are disjoint and the fact that when the intersection of a cell with this ball is nonempty, it will also have nonempty interior. $\triangleleft$ 
\end{rem}

The results of Theorem~\ref{discretizations:for:planning} and Corollary~\ref{corollary:controlled:transitions} can be utilized for motion planning tasks through the following procedure:

\noindent \bf{Step 1.} Given the Lipschitz constants $L_1$, $L_2$ and the bounds $M$, $v_{\max}$ on the agents' dynamics, pick design parameters $\lambda$, $\mu$ and select a well posed space-time discretization $\S-\delta t$ for the multi-agent system based on Theorem \ref{discretizations:for:planning}. 

\noindent \bf{Step 2.} Fix a reference point for each cell $S_l$ of the decomposition $\{S_l\}_{l\in\I}$. Then, derive the controlled transition system $TS_i^c$ of each agent $i$ as follows. For each cell configuration $\bf{l}_i=(l_i,l_{j_1},\ldots,l_{j_{N_i}})$ compute the endpoint $\chi_i(\delta t)$ of the reference trajectory \eqref{reference:solution} at time $\delta t$, corresponding to the reference points $x_{i,G},x_{j_1,G},\ldots,x_{j_{N_i},G}$ of the cells $S_{l_i},S_{l_{j_1}},\ldots,S_{l_{j_{N_i}}}$, as selected at the beginning of Step 2. Then, specify the cells which have nonempty intersection with $B(\chi_i(\delta t);r)$, in order to obtain all the transitions $l_i\overset{(\bf{l}_i,[w_i])}{\longrightarrow_i^c}l_i'$ to the cells in \eqref{Post:controlled}, with  $[w_i]=[w_i]_{(\bf{l}_i,l_i')}$ as given by \eqref{wi:class:Slcapball}. Also, note that the actions $[w_i]$ do not need to be specified until Step 4.   

\noindent \bf{Step 3.} Find a path $\bf{l}^0\bf{l}^1\bf{l}^2\cdots$ in the controlled product transition system $TS_{\P}^c$ defined in Remark~\ref{remark:controlled:product} which satisfies the plan and project it for each agent $i$ to a sequence of transitions $l_i^0\overset{(\bf{l}_i^0,[w_i]^0)}{\longrightarrow_i^c}l_i^1\overset{(\bf{l}_i^1,[w_i]^1)}{\longrightarrow_i^c}l_i^2\cdots$.

\noindent \bf{Step 4.} Select the control laws to implement the individual transitions by the continuous time system as follows. For each transition $l_i^m\overset{(\bf{l}_i^m,[w_i]^m)}{\longrightarrow_i^c}l_i^{m+1}$ pick any parameter $w_i\in [w_i]^m=[w_i]_{(\bf{l}_i^m,l_i^{m+1})}$, with the latter as given in \eqref{wi:class:Slcapball}, and apply the control law \eqref{feedback:ki} with the selected $w_i$.

The following corollary provides a lower bound for the minimum number of cells each agent can reach in time $\delta t$, depending on the selection of the design parameter $\mu$ for the space-time discretization.
\begin{corollary} \label{corollary:num:transitions}
Consider a cell decomposition $\S$ of $D$ with diameter $d_{\max}$, a time step $\delta t$, and parameters $\lambda\in(0,1)$, $\mu>0$ such that the hypotheses of Theorem \ref{discretizations:for:planning} are fulfilled. Then for each agent $i\in\N$ and each cell configuration of $i$, there exist at least $\lfloor \mu^n \rfloor+1$, if $\mu^n\notin\mathbb{N}$, or $\lfloor \mu^n \rfloor$, if $\mu^n\in\mathbb{N}$ possible discrete transitions.
\end{corollary}

\begin{proof}
In order to prove the result, we need by virtue of \eqref{planning:condition} to show that
\begin{equation} \label{cardinality:bound}
\#\{l\in\I:S_l\cap B(\chi_i(\delta t);r)\ne\emptyset\}\ge\left\lbrace
\begin{array}{ll}
\lfloor \mu^n \rfloor+1, & \;{\rm if}\;\mu^n\notin\mathbb{N}, \\
\lfloor \mu^n \rfloor, &\;{\rm if}\;\mu^n\in\mathbb{N},
\end{array}
\right.
\end{equation}

\noindent where $\#$ denotes the cardinality of a set. By using the notation ${\rm Vol}(S)$ for the volume (Lebesgue measure) of a measurable set $S\subset\Rat{n}$, it follows from \eqref{dmax:dfn} and the iso-diametric inequality that for each $l\in \I$ it holds
\begin{equation} \label{volume:Sl:bound}
{\rm Vol}(S_l)\le{\rm Vol}\left(B\left(\frac{d_{\max}}{2}\right)\right)=\left(\frac{d_{\max}}{2}\right)^n\beta(n):=S_{\max},
\end{equation}

\noindent where
$$
\beta(n):={\rm Vol}(B(1)),
$$

\noindent namely, the volume of the ball with center 0 and radius 1 in $\Rat{n}$. It then follows from \eqref{reachable:states:inD}, namely, that $B(\chi_i(\delta t);r)\subset D$, \eqref{volume:Sl:bound} and the fact that due to Definition \ref{cell:decomposition} it holds $\cup_{l\in\I}S_l=D$, that
\begin{equation} \label{cardinality:bound2}
\#\{l\in\I:S_l\cap B(\chi_i(\delta t);r)\ne\emptyset\}\ge\left\lbrace
\begin{array}{ll}
\lfloor \frac{ {\rm Vol}\left(B(\chi_i(\delta t);r)\right)}{S_{\max}} \rfloor+1, & \;{\rm if}\;\frac{ \left(B(\chi_i(\delta t);r)\right)}{S_{\max}}\notin\mathbb{N}, \\
\lfloor \frac{ \left(B(\chi_i(\delta t);r)\right)}{S_{\max}} \rfloor, &\;{\rm if}\;\frac{ \left(B(\chi_i(\delta t);r)\right)}{S_{\max}}\in\mathbb{N}.
\end{array}
\right.
\end{equation}

\noindent By taking into account \eqref{r:design:req} and \eqref{volume:Sl:bound}, we get that
\begin{equation} \label{cardinality:bound3}
\frac{{\rm Vol} \left(B(\chi_i(\delta t);r)\right)}{S_{\max}}\le \frac{\left(\frac{\mu}{2}d_{\max}\right)^n\beta(n)}{\left(\frac{d_{\max}}{2}\right)^n\beta(n)}=\mu^n
\end{equation}

\noindent and thus, \eqref{cardinality:bound} is a direct consequence of \eqref{cardinality:bound2} and \eqref{cardinality:bound3}. The proof is now complete.
\end{proof}

Finally, we provide certain upper bounds on the complexity of the controlled transition system of each agent as a function of $\lambda$, namely, the part of the input that is exploited for reachability purposes. Therefore, given a finite cell decomposition $\S=\{S_l\}_{l\in\I}$ of a bounded domain $D$, it is convenient to introduce the length
\begin{equation} \label{din:dfn}
d_{\rm in}=2\sup\{R>0:\forall l\in\I,\exists\,x\in S_l,B(x,R)\subset S_l\}
\end{equation}

\noindent corresponding to the maximum diameter of a ball that can be inscribed in all cells. The following corollary provides the corresponding complexity result. 
\begin{corollary} 
Consider a bounded domain $D$ admitting a finite cell decomposition $\S$ of diameter $d_{\max}$, a time step $\delta t$ and a parameter $\lambda\in(0,1)$, such that the hypotheses of Theorem~\ref{discretizations:for:planning} are fulfilled with $\mu=0$, implying that $d_{\max}$ and $\delta t$ satisfy \eqref{dmax:interval:caseA} and \eqref{deltat:interval:caseA}, respectively. If in addition $d_{\max}$ is the maximum possible diameter that satisfies \eqref{dmax:interval:caseA}, i.e., $d_{\max}=\frac{(1-\lambda)^2 v_{\max}^{2}}{4ML}$, then the cardinalities of the state set $Q_i^c$ and transition relation $\longrightarrow_i^c$ of agent's $i$ individual controlled transition system are upper bounded by $C_1\frac{1}{(1-\lambda)^{2n}}$ and $C_2\frac{1}{(1-\lambda)^{(2(N_i+1)+1)n}}$, respectively, with $C_1=\frac{{\rm Vol}(D)}{{\rm Vol}(B(\frac{1}{2}))}\left(\frac{4ML}{cv_{\max}^2}\right)^n$, $C_2=C_1^{N_i+1}\left(\frac{4}{c}\right)^n$, $c=\frac{d_{\rm in}}{d_{\max}}$,  $d_{\rm in}$ as given in \eqref{din:dfn} and ${\rm Vol}(D)$, ${\rm Vol}(B(\frac{1}{2}))$ being the volume of the domain $D$ and the ball with radius $\frac{1}{2}$ in $\Rat{n}$, respectively. Finally, the cardinality of the state set $Q_{\P}^c$ and transition relation $\longrightarrow_{\P}^c$ of the controlled product transition system are upper bounded by $C_1'\frac{1}{(1-\lambda)^{2Nn}}$ and $C_2'\frac{1}{(1-\lambda)^{3Nn}}$, respectively, with $C_1'=C_1^N$ and $C_2'=C_1^N\left(\frac{4}{c}\right)^{Nn}$. 
\end{corollary} 

\begin{proof}
Notice first, that when $d_{\max}=\frac{(1-\lambda)^2 v_{\max}^{2}}{4ML}$, it follows from \eqref{deltat:interval:caseA} that necessarily 
\begin{equation} \label{deltat:seclection}
\delta t=\frac{(1-\lambda)v_{\max}}{2ML}.
\end{equation}

\noindent Next, based on the finiteness hypothesis of $\S$ which implies that the length $d_{\rm in}$ as defined in \eqref{din:dfn} is positive, we provide an upper bound on the number of cells  in $\S$ and the number of cells that may intersect any reachable ball in $D$ with radius $r$, as the latter is given by \eqref{distance:r}. First, notice that due to \eqref{din:dfn} it holds
\begin{equation} \label{vol:Sl:lt:balldin}
{\rm Vol}(S_l)\le {\rm Vol}(B(\tfrac{d_{\rm in}}{2})),\forall l\in\I.
\end{equation}

\noindent In addition, in order to obtain a bound on the cardinality of $\I$, we exploit the facts that the cells of the decomposition cover $D$ and have disjoint interiors, which implies that for any subset of cells, the volume of their union equals the sum of their individual volumes. The latter properties in conjunction with \eqref{vol:Sl:lt:balldin} imply that 
\begin{align}
{\rm Vol}(\cup_{i\in\I}S_l) & ={\rm Vol}(D)\Longrightarrow \sum_{l\in\I}{\rm Vol}(S_l)={\rm Vol}(D) \Longrightarrow
\sum_{l\in\I}{\rm Vol}(B(\tfrac{d_{\rm in}}{2})) \le {\rm Vol}(D)\Longrightarrow \nonumber  \\ 
\#\I & \le\frac{{\rm Vol}(D)}{{\rm Vol}(B(\tfrac{d_{\rm in}}{2}))} = \frac{{\rm Vol}(D)}{{\rm Vol}(B(\tfrac{1}{2}))}\frac{1}{d_{\rm in}^n}=\frac{{\rm Vol}(D)}{{\rm Vol}(B(\tfrac{1}{2}))}\frac{1}{(cd_{\max})^n} \nonumber \\
& = \frac{{\rm Vol}(D)}{{\rm Vol}(B(\tfrac{1}{2}))}\left(\frac{4ML}{c(1-\lambda)^2v_{\max}^2}\right)^n=\frac{{\rm Vol}(D)}{{\rm Vol}(B(\tfrac{1}{2}))}\left(\frac{4ML}{cv_{\max}^2}\right)^n\frac{1}{(1-\lambda)^{2n}} \nonumber \\
& = C_1\frac{1}{(1-\lambda)^{2n}}, \label{cells:upper:bound}
\end{align}

\noindent with $c$ and $C_1$ as given in the statement of the corollary. We next proceed to determine an upper bound on the number of cells which intersect any reachable ball $B(\chi_i(\delta t);r)$. Note first, that by the definition of $d_{\max}$ it follows that 
\begin{align}
\{l\in\I: S_l\cap B(\chi_i(\delta t);r) \ne\emptyset \} & \subset \{l\in\I: S_l\subset B(\chi_i(\delta t);r+d_{\max}) \ne\emptyset \} \Longrightarrow \nonumber \\
\#\{l\in\I: S_l\cap B(\chi_i(\delta t);r) \ne\emptyset \} & \le \#\{l\in\I: S_l\subset B(\chi_i(\delta t);r+d_{\max})\}. \label{cardinality:intersect:vs:larger:ball}
\end{align} 
   
\noindent In addition, by taking into account as above that the cells of the decomposition have disjoint interiors, we get that
\begin{align}
{\rm Vol}\left(\bigcup\{S_l\in\S: S_l\subset B(\chi_i(\delta t);r+d_{\max})\right) & \subset {\rm Vol}(B(\chi_i(\delta t);r+d_{\max})) \Longrightarrow \nonumber \\
\sum_{\{l\in\I: S_l\subset B(\chi_i(\delta t);r+d_{\max})\}}{\rm Vol}(S_l) & \le {\rm Vol}(B(\chi_i(\delta t);r+d_{\max})) \Longrightarrow \nonumber \\
\#\{l\in\I: S_l\subset B(\chi_i(\delta t);r+d_{\max})\} & \le
\frac{{\rm Vol}(B(\chi_i(\delta t);r+d_{\max}))}{{\rm Vol}(B(\tfrac{d_{\rm in}}{2}))} \nonumber \\
\phantom{\#\{l\in\I: S_l\subset B(\chi_i(\delta t);r+d_{\max})\}} & =\left(\frac{r+d_{\max}}{\frac{d_{\rm in}}{2}}\right)^n. \label{rplusdmax:over:din} 
\end{align} 

\noindent From the selection of $d_{\max}$ in the statement of the corollary, $\delta t$ in \eqref{deltat:seclection} and the definition of $r$ we get that   
\begin{align}
\left(\frac{r+d_{\max}}{\frac{d_{\rm in}}{2}}\right)^n & = \left(\frac{\lambda v_{\max}\frac{(1-\lambda)v_{\max}}{2ML}+\frac{(1-\lambda)^2 v_{\max}^{2}}{4ML}}{\frac{cd_{\max}}{2}}\right)^n\le\left(\frac{\frac{(1-\lambda)v_{\max}^{2}}{2ML}}{\frac{cd_{\max}}{2}}\right)^n \nonumber \\
& = \left(\frac{\frac{(1-\lambda)v_{\max}^{2}}{ML}}{\frac{c(1-\lambda)v_{\max}^{2}}{4ML}}\right)^n=\left(\frac{4}{c}\right)^n\frac{1}{(1-\lambda)^n}. \label{rplusdmax:over:din2}
\end{align}

\noindent Due to \eqref{cells:upper:bound} it follows directly that the cardinality of the state set of each agent's individual (controlled) transition system is upper bounded by $ C_1\frac{1}{(1-\lambda)^{2n}}$. In order to obtain the corresponding bound for the transition relation, note that due to  \eqref{wi:class:Slcapball} and \eqref{Post:controlled}, for each $\bf{l}_i=(l_i,l_{j_1},\ldots,l_{j_{N_i}})$ it holds that
\begin{equation} \label{cells:intersect:up:bound}
l_i\overset{(\bf{l}_i,[w_i])}{\longrightarrow_i^c}l_i'\quad{\rm iff}\quad l_i'\in \{l\in\I: S_l\cap B(\chi_i(\delta t);r) \ne\emptyset \}\quad{\rm and}\quad [w_i]=[w_i]_{(\bf{l}_i,l_i')}
\end{equation}

\noindent and that each action $[w_i]_{(\bf{l}_i,l_i')}$ is uniquely determined by the successor cell $l_i'$. Thus, the cardinality of each agent's $i$ transition relation is evaluated by summing the numbers of the possible successor cells over all the possible cell configurations of $i$. This observation implies by virtue of \eqref{cells:intersect:up:bound} that  
\begin{equation} \label{individual:trans:exact:num}
\#\longrightarrow_i^c=\sum_{\bf{l}_i\in\I^{N_i+1}} \#\{l\in\I: S_l\cap B(\chi_i(\delta t);r) \ne\emptyset \}.
\end{equation}

\noindent Hence, by exploiting \eqref{cells:upper:bound}, \eqref{cardinality:intersect:vs:larger:ball}, \eqref{rplusdmax:over:din} and \eqref{rplusdmax:over:din2}, we get that 
\begin{align} \label{individual:trans:exact:num}
\#\longrightarrow_i^c & \le\#\I^{N_i+1} \left(\frac{4}{c}\right)^n\frac{1}{(1-\lambda)^n}=(\#\I)^{N_i+1} \left(\frac{4}{c}\right)^n\frac{1}{(1-\lambda)^n} \nonumber \\
& \le \left(C_1\frac{1}{(1-\lambda)^{2n}}\right)^{N_i+1} \left(\frac{4}{c}\right)^n\frac{1}{(1-\lambda)^n} \nonumber \\
& =C_1^{N_i+1}\left(\frac{4}{c}\right)^n\frac{1}{(1-\lambda)^{(2(N_i+1)+1)n}}=C_2\frac{1}{(1-\lambda)^{(2(N_i+1)+1)n}},
\end{align}

\noindent with $C_2$ as given in the statement of the corollary. Finally, from the definition of the product controlled transition system, we obtain from  \eqref{cells:upper:bound} that the cardinality of its state set satisfies
$$
\#Q_{\P}^c=\#\I^{N}=(\#\I)^N\le \left(C_1\frac{1}{(1-\lambda)^{2n}}\right)^N= C_1^N\frac{1}{(1-\lambda)^{2Nn}}=C_1'\frac{1}{(1-\lambda)^{2Nn}},
$$  

\noindent with $C_1'$ as given in the statement of the corollary. In addition, we obtain that the cardinality of the transition relation is given by
$$
\#\longrightarrow_{\P}^c=\sum_{\bf{l}\in\I^{N}}\prod_{i=1}^N \#\{l\in\I: S_l\cap B(\chi_{i,{\rm pr}_i(\bf{l})}(\delta t);r) \ne\emptyset \},
$$

\noindent where for each $\bf{l}\in\I^N$ and $i\in\N$, $\chi_{i,{\rm pr}_i(\bf{l})}(\delta t)$ denotes the reference trajectory corresponding to the cell configuration ${\rm pr}_i(\bf{l})$ of agent $i$. Thus, we get from \eqref{cells:upper:bound}, \eqref{cardinality:intersect:vs:larger:ball}, \eqref{rplusdmax:over:din} and \eqref{rplusdmax:over:din2} that
\begin{align*}
\#\longrightarrow_{\P}^c & \le \#\I^N \prod_{i=1}^N \left(\frac{4}{c}\right)^n\frac{1}{(1-\lambda)^n}\le (\#\I)^N \left(\frac{4}{c}\right)^{Nn}\frac{1}{(1-\lambda)^{Nn}} \\
& \le \left(C_1\frac{1}{(1-\lambda)^{2n}}\right)^N\left(\frac{4}{c}\right)^{Nn}\frac{1}{(1-\lambda)^{Nn}}=C_1^N \left(\frac{4}{c}\right)^{Nn}\frac{1}{(1-\lambda)^{3Nn}}=C_2'\frac{1}{(1-\lambda)^{3Nn}},
\end{align*}

\noindent with $C_2'$ as given in the statement of the corollary. The proof is now complete.
\end{proof}

\section{Example and Simulation Results}

As an illustrative example we consider a system of four agents with states $x_1,x_2,x_3,x_4\in\Rat{2}$, whose initial conditions lie inside the circular domain ${\rm int}(B(R))(=\{x\in\Rat{2}:|x|<R\})$ with center zero and radius $R>0$. Their dynamics are given as:
\begin{align}
\dot{x}_1 & = {\rm sat}_{\rho}(x_2-x_1)+g(x_1)+v_1, \nonumber \\
\dot{x}_2 & = g(x_2)+v_2, \nonumber \\
\dot{x}_3 & = {\rm sat}_{\rho}(x_2-x_3)+g(x_3)+v_3, \nonumber \\
\dot{x}_4 & = {\rm sat}_{\rho}(x_3-x_4)+g(x_4)+v_4, \label{example:dynamics}
\end{align}

\noindent where the function ${\rm sat}_{\rho}:\Rat{2}\to\Rat{2}$ is defined as ${\rm sat}_{\rho}(x):=x$ if $|x|\le\rho$; ${\rm sat}_{\rho}(x):=\frac{\rho}{|x|}x$, if $|x|>\rho$. The agents' neighbors' sets in this example are $\N_1=\{2\}$, $\N_2=\emptyset$, $\N_3=\{2\}$, $\N_4=\{3\}$ and specify the corresponding network topology. The constant $\rho>0$ in \eqref{example:dynamics} satisfies $\rho\le R$ and represents a bound on the distance between agents 1, 2, and agents 2, 3, that we will require the system to satisfy during its evolution. The function $g(\cdot)$ is defined as
\begin{equation} \label{function:g}
g(x):=\begin{cases}
    0, & {\rm if}\; |x|<R-\frac{\rho}{2}, \\
    ((R-\frac{\rho}{2})-|x|)\frac{x}{|x|}, & {\rm if}\; R-\frac{\rho}{2} \le |x| <R, \\
    -\frac{\rho}{2}\frac{x}{|x|}, & {\rm if}\; R\le |x| \\
    \end{cases}
\end{equation}

\noindent and determines for each agent a repulsive vector field from the boundary of ${\rm int}(B(R))$ when the agent is located in ${\rm int}(B(R))$, 
in order to ensure invariance of the agents' trajectories inside this circular domain.

\noindent We next show, that if the initial distances between agents 1 and 2 (and similarly for agents 2 and 3) is less than $\rho$, it will also remain less than $\rho$ for all positive times, for an appropriate bound on the magnitude of the free input terms $v_i$. By selecting the energy function $V(x_1,x_2):=\frac{1}{2}|x_1-x_2|^2$ and evaluating its derivative along the right hand side of \eqref{example:dynamics}, we obtain that
\begin{align}
\dot{V} & =\langle x_1-x_2, x_2-x_1+g(x_1)+v_1-g(x_2)-v_2 \rangle \nonumber \\
& \le-|x_1-x_2|^2+\langle x_1-x_2, g(x_1)-g(x_2)\rangle+2|x_1-x_2|v_{\max}, {\rm if}\; |x_1-x_2|<\rho \label{V12:dot:case1} \\
\dot{V} & =\langle x_1-x_2, \frac{\rho}{|x_2-x_1|}(x_2-x_1)+g(x_1)+v_1-g(x_2)-v_2\rangle \nonumber \\
& \le-(\rho-2v_{\max})|x_1-x_2|+\langle x_1-x_2, g(x_1)-g(x_2)\rangle, {\rm if}\; |x_1-x_2|\ge \rho \label{V12:dot:case2}
\end{align}

\noindent where $\langle \cdot,\cdot \rangle$ denotes the inner product in $\Rat{2}$. Next, notice that
\begin{equation} \label{g:property}
\langle x-y,g(x)-g(y) \rangle\le 0, \forall x,y\in\Rat{2}
\end{equation}

\noindent Indeed, assume that without any loss of generality it holds $|x|\ge|y|$. Then, it follows from \eqref{function:g} that there exist $\alpha\ge\beta\ge 0$, such that $g(x)=-\alpha x$ and $g(y)=-\beta y$. Hence, we get that $\langle x-y,g(x)-g(y) \rangle=\langle x-y,-\alpha x-\beta y \rangle \le-\alpha |x|^2+(\alpha+\beta)|x||y|-\beta|y|^2$. By evaluating the discriminant of the latter second order expression, we obtain that it is always nonpositive, which implies \eqref{g:property}. By additionally assuming that
\begin{equation} \label{example:vmax}
v_{\max}=\frac{\rho}{2}
\end{equation}

\noindent we obtain from \eqref{V12:dot:case1}, \eqref{V12:dot:case2} and \eqref{g:property} that $\dot{V}\le 0$ when $|x_1-x_2|\ge \rho$. Thus, it follows that $|x_1(0)-x_2(0)|\le \rho$ implies that $|x_1(t)-x_2(t)|\le \rho$ for all positive times. Analogously, by considering the  function $V(x_2,x_3):=\frac{1}{2}|x_2-x_3|^2$ it follows that the same holds for $|x_2(t)-x_3(t)|$. Furthermore, under the selection of $v_{\max}=\frac{\rho}{2}$, it can be deduced (along the lines of the corresponding result in \cite{BdDd15}) that the circular domain remains invariant for the dynamics of the system. Finally, we obtain from \eqref{example:dynamics} and \eqref{function:g} the following dynamics bounds and Lipschitz constants in \eqref{dynamics:bound}, \eqref{dynamics:bound1} and \eqref{dynamics:bound2}, respectively:
\begin{equation}\label{example:dynamics:bounds}
M=\frac{3}{2}\rho,\quad L_1=1, \quad L_2=2.
\end{equation}

\noindent Thus, it follows that system \eqref{example:dynamics} satisfies all requirements for the derivation of well posed discretizations.

\noindent In this example, it is also assumed that the reference point of each cell of the square partition is the center of the square. This enables us to obtain the following improved bounds on the feedback laws in \eqref{feedback:ki}, for their corresponding values of $t$, $x_i$, $\bf{x}_j$ and $w_i$:
\begin{align*}
|k_{i,\bf{l}_i,1}(x_i,\bf{x}_j)| & \le L_1\left((M+v_{\max})\delta t+\frac{d_{\max}}{2}\right) \\
|k_{i,\bf{l}_i,2}(x_{i0})| & \le \frac{d_{\max}}{2\delta t} \\
|k_{i,\bf{l}_i,3}(t;x_{i0},w)| & \le L_2\left(\lambda v_{\max}\delta t+\frac{d_{\max}}{2}\right)+\lambda v_{\max}
\end{align*}

\noindent Thus, in order to verify Property (P2), we need to select $d_{\max}$ and $\delta t$ satisfying
$$
L_1((M+v_{\max})\delta t+\frac{d_{\max}}{2})+\frac{d_{\max}}{2\delta t}+L_2\left(\lambda v_{\max}\delta t+\frac{d_{\max}}{2}\right)+\lambda v_{\max} \le v_{\max}.
$$

\noindent Equivalently, by virtue of  \eqref{example:vmax} and \eqref{example:dynamics:bounds}, it is required that
\begin{align*}
\left(2\rho\delta t+\frac{d_{\max}}{2}\right)+\frac{d_{\max}}{2\delta t}+2\left(\lambda \frac{\rho}{2}\delta t+\frac{d_{\max}}{2}\right) & \le (1-\lambda)\frac{\rho}{2} \iff \\
d_{\max} & \le \rho\frac{(1-\lambda)\delta t-4\delta t^2}{3\delta t+1}
\end{align*}

\noindent By evaluating the derivative of $d_{\max}(\cdot)$ with respect to $\delta t$ we obtain that
\begin{align*}
\dot{d}_{\max}=0\iff & ((1-\lambda)-8\delta t)(3\delta t+1)-3((1-\lambda)\delta t-4\delta t^2)=0 \\
\iff & -12\delta t^2-8\delta t+(1-\lambda)=0
\end{align*}

\noindent Hence, we obtain the time
$$
\bar{\delta t}:=\frac{8-4\sqrt{4+3(1-\lambda)}}{-24}=\frac{-2+\sqrt{4+3(1-\lambda)}}{6}
$$

\noindent corresponding to the maximum possible diameter
$$
\bar{d}_{\max}=\rho\frac{(1-\lambda)\bar{\delta t}-4\bar{\delta t}^2}{3\bar{\delta t}+1}
$$

For the simulation results, we select the distance $\rho=10$ and the radius of the circular domain $R=10$. We also assume that the agents 1, 2, 3 and 4, are initially located at $x_{10}=(5,-3)$, $x_{20}=(5,3)$, $x_{30}=(0,6)$ and $x_{40}=(-4,6)$, respectively. Thus, it follows that agents 1, 2, and 2, 3, satisfy the requirement on their initial relative distance. In the sequel we will focus on the behaviour of the system for times $t\in[0,2]$. Given this time interval and a selection of the planning parameter $\lambda\in(0,1)$, we choose the time step $\delta t$ as the largest possible time step not exceeding $\bar{\delta t}$ above, in such a way that the number of time steps $NT:=\frac{2}{\delta t}$ is a positive integer. We also choose the largest possible cell diameter $d_{\max}$ corresponding to $\delta t$ and consider a square grid in $\Rat{2}$. Each square has side length $d$, where $d$ is the largest number not exceeding $\frac{\sqrt{2}}{2}d_{\max}$, such that the quotient $\frac{2R}{d}$ is an integer. Thus, we can form a cell decomposition of the circular domain $D$ by defining as a cell each square in the grid which has nonempty intersection with $D$. In Figs. 3, 4 we have plotted (half of) the grid lines, in order to illustrate how the grid is affected by the choice of $\lambda$. We next consider two cases for the motion of agent 2, which is unaffected by the coupled constraints.

\textbf{Case I:} It holds $v_2(t)=v_{2c}$, $\forall t\in[0,2]$, with $v_{2c}=(-3,-3)$.

\textbf{Case II:} It holds $v_2(t)=v_{2c}+v_{2d}(t)$, $\forall t\in[0,2]$, with $v_{2c}$ as above and $v_{2d}\in\U_{d}$, where $\U_d$ is the set of all piecewise continuous functions $\tilde{v}:[0,2]\to\Rat{2}$ that satisfy $\tilde{v}(t)=\gamma(t)(\frac{\sqrt{2}}{2},-\frac{\sqrt{2}}{2})$, with $-1\le\gamma(t)\le 1$ for all $t\in [0,0.9]$ and $\gamma(t)=0$, for all $t\in (0.9,2]$.

\noindent Notice that in Case I we consider a pre-specified path for agent 2, by selecting a constant control, whereas in Case II we allow for the possibility to modify this path and superpose a motion perpendicular to it (up to certain bound) over the time interval $[0,0.9]$. Furthermore, in both cases the magnitude of $v_{2}(\cdot)$ is bounded by $v_{\max}(=5)$.

For Case I, we assign reachability goals to agents 1, 3 and 4 which should be fulfilled at the end of the time interval $[0,2]$, given the selected path for agent 2. Specifically, we want agents 1, 3 and 4 to reach the corresponding boxes in the workspace that are depicted in Fig. 3. First, we sample the trajectories of 2 at the time instants $\kappa\delta t$, $\kappa=0,1\ldots,NT$ and specify the sequence $l_2^0l_2^1\cdots l_2^{NT}$ corresponding to the cells $S_{l_{\kappa}}$ with $x_2(\kappa\delta t,x_{20})\in S_{l_{\kappa}}$. Then, 
we exploit the individual controlled transition systems of agents 1 and 3, in order to determine (an underapproximation of) their reachable cells for the given sampled trajectory of agent 2. In particular, by denoting as $l_1^0$ the index of the cell where the initial state $x_{10}$ of agent 1 belongs, we can evaluate the indices of its reachable cells at time $\kappa\delta t$ as $Q_1^{\kappa}={\rm Post}_1^c(Q_1^{\kappa-1};l_2^{\kappa-1})$, $\kappa=0,1,\ldots,NT$, where $Q_1^0:=\{l_1^0\}$ and we have used the notational convention ${\rm Post}_1^c(l_1;l_2):={\rm Post}_1^c(l_1;(l_1,l_2))$ (recall that $(l_1,l_2)$ stands for a cell configuration of agent 1) and the definition ${\rm Post}_1^c(Q_1;l_2):=\cup_{l_1\in Q_1}{\rm Post}_1^c(l_1;l_2)$. The approach followed in this case is possible because agent 2 is decoupled from the other agents and the individual transition system of agent 1 depends only on the cell indices of agent 2. Similarly, we can evaluate the reachable cells of agent 1 and check whether it fulfils its reachability task. Next, by computing the reachable cells of agent 3 which lie in its target box at the final time step $NT$, we calculate the backward reachable cells of the agent in order to encode the discrete trajectories which fulfil its  reachability goal. Then, we exploit the individual transition system of agent 4 in order to determine its reachable cells for all the possible trajectories of agent 3 that satisfy its reachability task. The corresponding simulation results are depicted in Fig. 3 for $\lambda=0.2$ (left) and $\lambda=0.3$ (right). The figure also illustrates the effect of the parameter $\lambda$ in the accomplishment of the reachability goals, since for $\lambda=0.2$ only agent 3 reaches its target box, whereas for $\lambda=0.3$ all agents achieve their corresponding task. 

\begin{figure}[H]
\includegraphics[width =0.49\columnwidth]{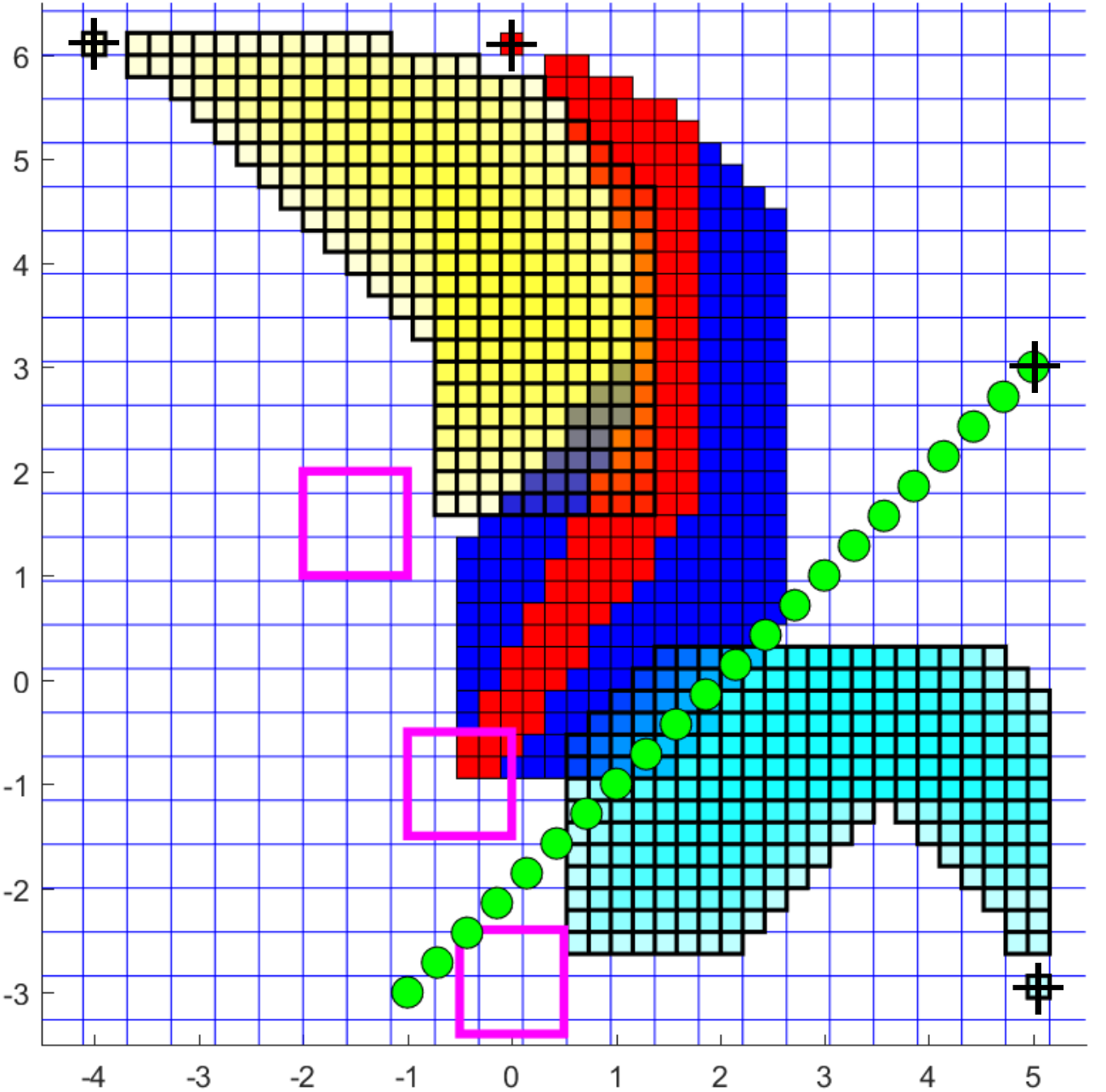} \includegraphics[width =0.49\columnwidth]{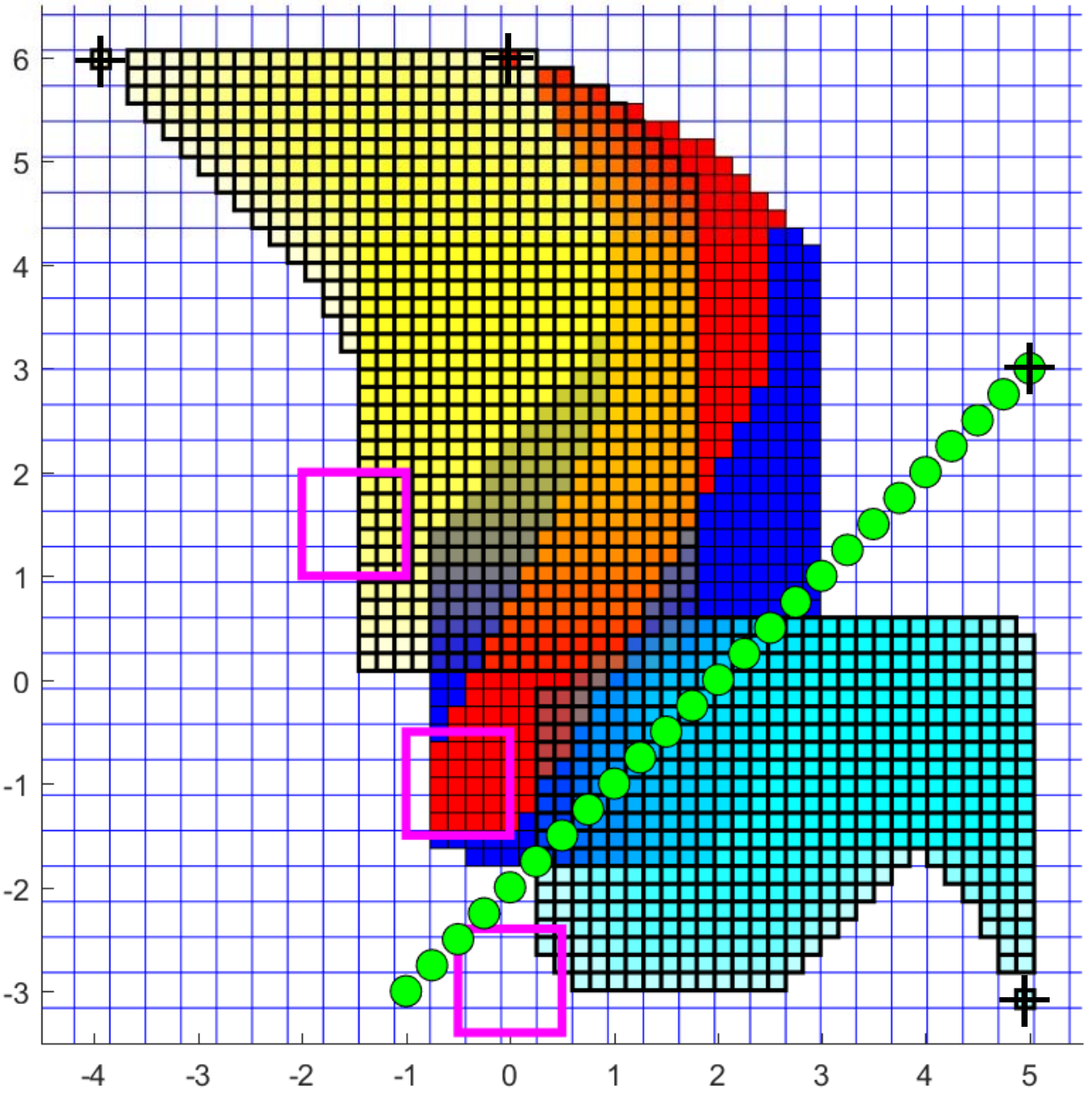} \vspace{-0.3cm}
\center{\small{(i)} \hspace{0.45\columnwidth}\small{(ii)}} \vspace{-0.3cm}
\caption{Reachable cells of the agents for (i) $\lambda=0.2$ and (ii) $\lambda=0.3$. Agents 1, 3 and 4 are initially located at the bottom right, top center and top left of the illustrated workspace, respectively. The circles denote the sampled trajectory of agent 2 as determined by Case I and the boxes the corresponding target sets of agents 1, 3 and 4. The union of all discrete paths of agent 3 which end in its target box are highlighted within the union of its reachable cells.}
\end{figure}

For Case II, we exploit the individual controlled transition system of agents 1, 3 and 4 in order to obtain (an underapproximation of) the cells these agents can reach, irrespectively of the choice of $v_{2d}$ for the free input of agent 2. In particular, we define the finite cell sequence $\{Q_2^{\kappa}\}_{\kappa\in\{0,1,\ldots,NT\}}$ as $Q_2^{\kappa}=\{l\in\I:\exists v_{2,d}\in\U_d\;{\rm with}\;x_2(\kappa\delta t,x_{20};v_{2c}+v_{2d}(\cdot))\in S_l\}$. Also, we inductively define for $\kappa=0,1,\ldots,NT$ the sets $Q_1^{\kappa}=\cup_{l_1\in Q_1^{\kappa-1}}\cap_{l_2\in Q_2^{\kappa-1}}{\rm Post}_1^{c}(l_1;l_2)$, $Q_3^{\kappa}=\cup_{l_3\in Q_3^{\kappa-1}}\cap_{l_2\in Q_2^{\kappa-1}}{\rm Post}_3^3(l_3;l_2)$ and $Q_4^{\kappa}=\cup_{l_4\in Q_4^{\kappa-1}}\cap_{l_3\in Q_3^{\kappa-1}}{\rm Post}_4^c(l_4;l_3)$, with $Q_1^{0}=\{l_1^0\}$, $Q_3^{0}=\{l_3^0\}$ and $Q_4^{0}=\{l_4^0\}$ (we use the same notational convention as above for the operators ${\rm Post}_i^c(\cdot)$, and the notation $l_i^0$ for the initial cells of the agents $i=1,3,4$). Next, consider any selection of sequences $l_1^0l_1^1\cdots l_1^{NT}$, $l_3^0l_3^1\cdots l_3^{NT}$ and $l_4^0l_4^1\cdots l_4^{NT}$, of agents 1, 3 and 4, that satisfy $l_1^{\kappa}\in Q_1^{\kappa}$, $l_3^{\kappa}\in Q_3^{\kappa}$ and $l_4^{\kappa}\in Q_4^{\kappa}$, respectively. Then, by taking into account the definition of the sets $Q_i^{\kappa}$, $i=1,3,4$, the definition of the individual transition systems of agents 1, 3, 4, and the particular coupling between the agents in this example, we arrive at the following conclusion. For each agent 1, 3 and 4, it is possible to assign a sequence of control laws, such that each corresponding agent will reach the cells with indices $l_1^{\kappa}$, $l_3^{\kappa}$ and $l_4^{\kappa}$ at time $\kappa\delta t$, respectively, for any selection of the input $v_{2d}$ of agent 2.
In Fig. 4 we illustrate the union of the reachable cells of agents 1, 3 and 4 for $\lambda=0.3$ and $\lambda=0.4$, respectively. Notice that  the underapproximation of agents' 1 and 3 reachable cells increases with the selection of the larger parameter $\lambda$, namely, with the exploitation of a larger part of the free input for planning. However, the same observation does not necessarily hold for the reachable cells of agent 4. The reason why the area covered by the reachable cells of agent 4 remains approximately the same, is that the corresponding area increases for agent 3 for larger values of $\lambda$. Thus, although the reachability properties of agent 4 are improved, this is compensated by the fact that each illustrated transition of agent 4 to a certain cell needs to be possible for an increasing number of different positions of agent 3.

\begin{figure}[H]
\includegraphics[width =0.49\columnwidth]{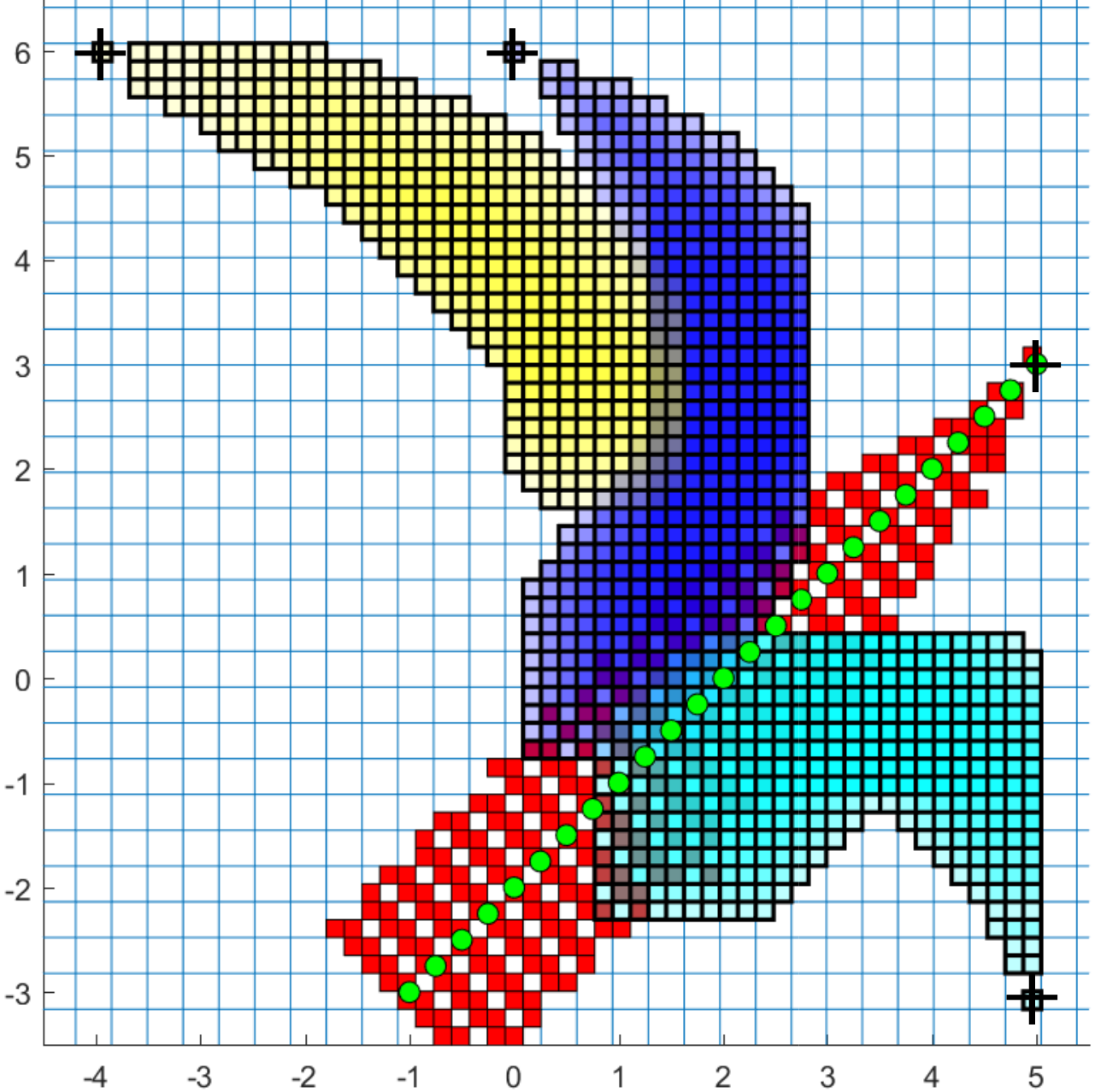} \includegraphics[width =0.49\columnwidth]{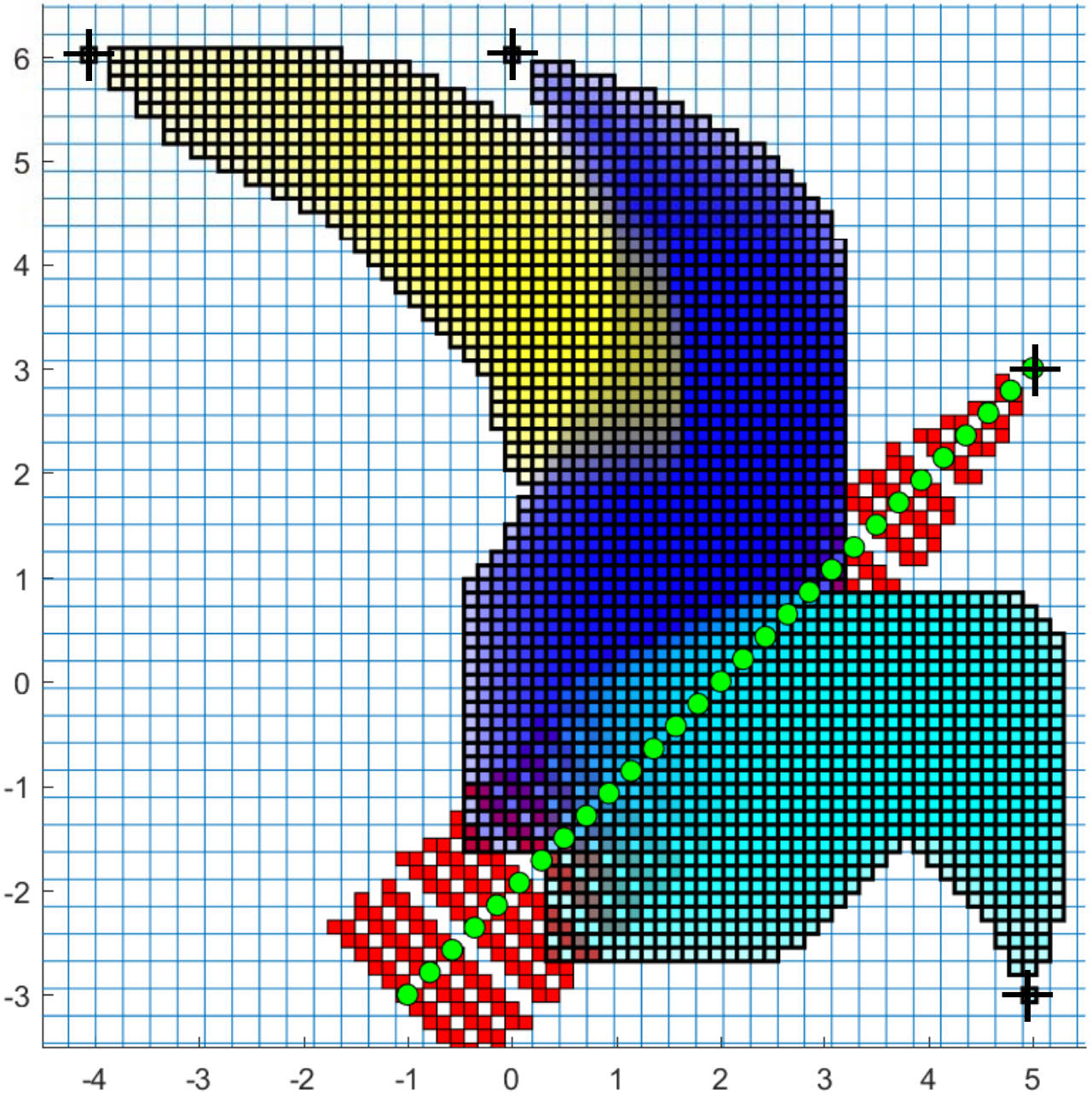} \vspace{-0.3cm}
\center{\small{(i)} \hspace{0.45\columnwidth}\small{(ii)}} \vspace{-0.3cm}
\caption{Reachable cells of the agents for (i) $\lambda=0.3$ and (ii) $\lambda=0.4$. Agents 1, 3 and 4 are initially located at the bottom right, top center and top left of the illustrated workspace, respectively. The circles denote the nominal sampled trajectory of agent 2 and their nearby cells represent the cells where agent 2 can lie at the sampling times, for all possible inputs of Case II.}
\end{figure}

The code for the simulation results has been implemented in MATLAB and the worst case running time for the illustrated results is of the order of 45 minutes, on a PC with an Intel(R) Core(TM) i7-4600U CPU @ 2.10GHz processor. 

\section{Conclusions}

We have provided a decentralized abstraction methodology for multi-agent systems and quantified bounds on the space and  time discretization in order to obtain meaningful transition systems with multiple transition possibilities. The abstraction framework is based on the design of hybrid feedback control laws that take into account the agents' coupled constraints and guarantee the implementation of the discrete transitions by the continuous time controllers.

Ongoing work includes the improvement of the acceptable choices of $d_{\max}$ and $\delta t$ in order to obtain coarser abstractions and reduce the size of each agent's transition system. Another possible direction for complexity reduction is the modification (localization) of the current framework through an event based online abstraction with updated choices of $d_{\max}$ and $\delta t$. Finally, it should be noted that while this report provides informal indicators of how the results can be used for planning, we are currently formalizing a distributed planning methodology from high level specifications that builds on the derived abstractions.

\section{Appendix}

In the Appendix we provide the proofs of Propositions \ref{discrete:transitions:result} and \ref{completeness:result} and certain technical parts from the proof of Proposition \ref{proposition:property:P}.

\begin{proof}[Proof of Proposition \ref{discrete:transitions:result}]
Indeed, consider a final cell configuration $\bf{l}'=(l_1',\ldots,l_N')$ as in the statement of the proposition. By the definitions of the operators ${\rm Post}_i(\cdot)$, $i\in\N$ and the transition relation of each corresponding agent's individual transition system, it follows that each transition $l_i\overset{{\rm pr}_i(\bf{l})}{\longrightarrow}l_i'$ is well posed in the sense of Definition \ref{well:posed:discretization}(i). Thus, given the initial cell configuration ${\rm pr}_i(\bf{l})$ and the well posed transition $l_i\overset{{\rm pr}_i(\bf{l})}{\longrightarrow}l_i'$  for each agent, it follows from Definition \ref{well:posed:discretization}(i) that we can pick for each $i\in\N$ a control law $k_{i,{\rm pr}_i(\bf{l})}(\cdot)$ that satisfies Property (P) and a vector $w_i\in W$, such that Condition (C) holds. Next, notice that for each agent $i$ the projection of the initial cell configuration $\bf{l}=(l_1,\ldots,l_N)$ is ${\rm pr}_i(\bf{l})$, namely, the cell configuration according to which the feedback law $k_{i,{\rm pr}_i(\bf{l})}(\cdot)$ was selected, and that the feedback law satisfies Property (P). Thus, it follows from Condition (C) that the solution of the closed loop system is well defined on $[0,\delta t]$, and that for each $i\in\N$, the $i$-th component of the solution satisfies \eqref{contoler:compatibility}.
\end{proof}

\begin{proof}[Proof of Proposition \ref{completeness:result}]
\textbf{Proof of (i).} Let $w_i\in W$, $i\in\N$  and $x(0)\in D^N$ with $x_{i}(0)\in S_{l_i}$, $i\in\N$ be the initial condition of the closed-loop system. Then, it follows from the local Lipschitz property on the functions $f_{i}(\cdot)$ and the corresponding property on the mappings $k_{i,{\rm pr}_i(\bf{l})}(\cdot;x_{i0},w_i)$ provided by (P1), that the dynamics of the closed loop system are given by a locally Lipschitz function on $D^N$. Hence, there exists a unique solution $x(\cdot)=x(\cdot,x(0))$ to the initial value problem, which is defined and remains in $D^N$ for all times in its right maximal interval of existence $[0,T_{\max})$. We proceed by proving that each component $x_i(\cdot)$, $i\in\N$ of the solution satisfies
\begin{equation} \label{xis:in:initialcell:plus:Rmax}
x_i(t)\in (S_{l_i}+B(R_{\max}))\cap D,\forall t\in[0,\min\{T_{\max},\tilde{T}\}).
\end{equation}



\noindent Indeed, suppose on the contrary that \eqref{xis:in:initialcell:plus:Rmax} is violated, and hence, by taking into account that $x_i(t)\in D$ for all $t\in[0,T_{\max})$, that there exists $\iota\in\N$ and a time $T$ with
\begin{equation} \label{xhati:atT}
T\in (0,\min\{T_{\max},\tilde{T}\})\;{\rm and}\;x_{\iota}(T)\notin S_{l_{\iota}}+B(R_{\max}).
\end{equation}

\noindent By recalling that $x_{i}(0)\in S_{l_i}$, $i\in\N$, 
 we may define
\begin{align}
\tau:=\max\{\bar{t} \in [0,T]:x_{i}(t) & \in{\rm cl}(S_{l_{i}}+B(R_{\max})), \nonumber \\
\forall t & \in [0,\bar{t}],i\in\N\}. \label{time:tau}
\end{align}

\noindent Then, it follows from \eqref{time:tildeT}, \eqref{xhati:atT} and \eqref{time:tau} that there exists $\ell\in\N$ such that
\begin{equation} \label{xtildei:at:tau}
x_{\ell}(\tau)\in\partial(S_{l_{\ell}}+B(R_{\max}))
\end{equation}

\noindent and that
\begin{equation} \label{tau:vs:deltat}
\tau\le T<\tilde{T}\le\delta t.
\end{equation}

\noindent It also follows from \eqref{time:tau}, 
\eqref{time:tildeT}, \eqref{tau:vs:deltat} and Property (P2)
that for all $t\in[0,\tau]$ it holds
\begin{align}
|k_{\ell,{\rm pr}_{\ell}(\bf{l})}(t,x_{\ell}(t),\bf{x}_{j(\ell)}(t);x_{\ell 0},w_{\ell})|\le v_{\max} \label{feedback:for:tildei}
\end{align}

\noindent Hence, we get from \eqref{single:integrator}, \eqref{general:feedback:law}, \eqref{Rmax}, \eqref{feedback:for:all2}, \eqref{feedback:for:tildei} and \eqref{tau:vs:deltat}, which implies that $\tau<\delta t$, that
\begin{align}
|x_{\ell}(\tau)&-x_{\ell 0}|\le \int_{0}^{\tau}|f_{\ell}(x_{\ell}(s),\bf{x}_{j(\ell)}(s))|\nonumber \\
&+|k_{\ell},{\rm pr}_{\ell}(\bf{l})(s,x_{\ell}(s),\bf{x}_{j(\ell)}(s);x_{\ell 0},w_{\ell})|ds \nonumber \\
& \le \int_{0}^{\tau}(M+v_{\max})ds<\delta t(M+v_{\max})=R_{\max}. \label{contradiction:distance}
\end{align}

\noindent In order to finish the proof of \eqref{xis:in:initialcell:plus:Rmax} we exploit the following elementary fact.

\noindent \textbf{Fact I.} Consider a nonempty set $S\subset\Rat{n}$ and a constant $R>0$. Then for every $x\in\partial(S+B(R))$ it holds $|x-y|\ge R,\forall y\in S$.

\noindent \textbf{Proof of Fact I.} Indeed, suppose on the contrary that there exists $\tilde{y}\in S$ with $|x-\tilde{y}|\le R-\varepsilon$ for certain $\varepsilon>0$. Then for all $\tilde{x}\in{\rm int}(B(x;\varepsilon))$ we have
$$
|\tilde{x}-\tilde{y}|\le|\tilde{x}-x|+|x-\tilde{y}|<\varepsilon+R-\varepsilon=R,
$$

\noindent and hence, $\tilde{x}\in S+B(R)$ for all $\tilde{x}\in{\rm int}(B(x;\varepsilon))$, which implies that $x\notin\partial(S+B(R))$ and contradicts our statement. $\triangleleft$

By exploiting Fact I with $S=S_{l_{\ell}}$, $R=R_{\max}$, $y=x_{\ell 0}$ and $x=x_{\ell}(\tau)$ we deduce from \eqref{contradiction:distance} that $x_{\ell}(\tau)\notin \partial(S_{l_{\ell}}+B(R_{\max}))$ which contradicts \eqref{xtildei:at:tau}, and provides validity of \eqref{xis:in:initialcell:plus:Rmax}.

We now prove the following claim:

\noindent \textbf{Claim I.} It holds $T_{\max}\ge\tilde{T}$.

\noindent $\triangleright$\textbf{Proof of Claim I.} Indeed, suppose on the contrary that
\begin{equation} \label{Tmax:lt:tildeT}
T_{\max}<\tilde{T}.
\end{equation}

\noindent For each $i\in\N$ let $v_i:[0,\infty)\to\Rat{n}$ be a piecewise continuous function satisfying
\begin{equation} \label{free:input:eq:feedback}
v_i(t)=k_{i,{\rm pr}_{i}(\bf{l})}(t,x_i(t),\bf{x}_j(t);x_{i0},w_i),\forall t\in[0,T_{\max}).
\end{equation}

\noindent Notice that due to \eqref{time:tildeT} and \eqref{Tmax:lt:tildeT} we have that $T_{\max}<\min\{\delta t,\min\{T(x_{i0},w_i):i\in\N\}\}$, and thus, we get from  \eqref{xis:in:initialcell:plus:Rmax} and (P2) that $|v_i(t)|\le v_{\max}$, $\forall t\in[0,T_{\max})$. Hence, we may select $v_i(\cdot)$ to satisfy $|v_i(t)|\le v_{\max}$, $\forall t\ge 0$ (select for instance $v_i(t)=0$ for $t\ge T_{\max}$). Thus, if we denote by $\xi(\cdot)$ the solution of \eqref{single:integrator}-\eqref{general:feedback:law} with free inputs $v_i(\cdot)$, $i\in\N$ and the same initial condition with $x(\cdot)$, it follows from the Invariance Assumption (IA) that $\xi(t)$ is defined and remains in $D^N$ for all $t\ge 0$. Furthermore, it follows from standard arguments from the theory of ODEs that $\xi(t)=x(t),\forall t\in [0,T_{\max})$.
Hence, since $\xi(t)$ belongs to a compact subset of $D^N$ for all $t\in [0,T_{\max}]$, 
the same holds for $x(t)$ on $[0,T_{\max})$. The latter contradicts maximality of $[0,T_{\max})$
since by \eqref{Tmax:lt:tildeT} and \eqref{time:tildeT} it holds  $T_{\max}<\tilde{T}\le \min\{T(x_{i0},w_i):i\in\N\}$ and the mappings $k_{i,{\rm pr}_{i}(\bf{l})}(\cdot)$ are defined for $t\in [0,\min\{T(x_{i0},w_i):i\in\N\})$. Hence, we have shown Claim I. $\triangleleft$

From Claim I, it follows that $x(t)$ is defined and remains in $D^N$ for all $t\in [0,\tilde{T})$ and that \eqref{xis:in:initialcell:plus:Rmax} holds for all $t\in [0,\tilde{T})$. Thus, by applying the same arguments with those in the proof of Claim I, we can determine a continuous function $\xi(\cdot)$ with $\xi(t)=x(t)$ for all $t\in [0,\tilde{T})$ and $\xi(\tilde{T})\in D^N$, which establishes \eqref{limit:solution}.

\noindent \textbf{Proof of (iia).} In the case where (P3) also holds, and hence by \eqref{time:tildeT} we have that $\tilde{T}=\delta t$, it follows from part (i) of the proposition and standard arguments, that the solution $x(\cdot)$ is defined on $[0,T_{\max})$, with $T_{\max}>\delta t$. 
From the latter,
we conclude that $x(t)\in D^{N}$ for all $t\in[0,\delta t]$. Moreover, since $T_{\max}>\delta t=\tilde{T}$, it follows that \eqref{xis:in:initialcell:plus:Rmax} is satisfied for all $t\in[0,\delta t)$. The latter, by virtue of (P2), (P3) and continuity of $x(\cdot)$ implies \eqref{feedback:consistency}.

\noindent \textbf{Proof of (iib).} By exploiting the result of part (iia) of the proposition and defining $v_i(t)=k_{i,{\rm pr}_{i}(\bf{l})}(t,x_i(t),\bf{x}_j(t);$ $x_{i0},w_i)$,
$\forall t\in[0,\delta t)$ we can extend $v_i(\cdot)$ to a piecewise continuous function on $[0,\infty)$ which satisfies \eqref{input:bound}. 
Hence, by applying the same arguments with those in the proof of Claim~I, we conclude that the solutions $x(\cdot)$ of \eqref{single:integrator}-\eqref{general:feedback:law}, \eqref{feedback:for:all2} and $\xi(\cdot)$ of \eqref{single:integrator}-\eqref{general:feedback:law} (with input $v(\cdot)$)
coincide on $[0,\delta t]$.
\end{proof}

\noindent \textbf{Proof of the fact that the intervals provided by Cases I, II, III  in Proposition \ref{proposition:property:P} are well defined.}

The fact that \eqref{dmax:interval:caseA} and \eqref{deltat:interval:caseA} are well defined is straightforward. We proceed by defining
\begin{equation}\label{function:barg}
\bar{g}(\mu):=\frac{2(\lambda(1-\lambda)\mu-2\lambda^2)v_{\max}^{2}}{\mu^2ML}.
\end{equation}

\noindent Then it follows that
\begin{equation}\label{function:g:property}
\bar{g}(\mu)>0\iff \lambda(1-\lambda)\mu-2\lambda^2>0\iff \mu>\frac{2\lambda}{1-\lambda}.
\end{equation}

\noindent Furthermore,
\begin{align*}
{\rm sgn}(\dot{\bar{g}}(\mu)) & ={\rm sgn}\left(\frac{d}{d\mu}\left(\frac{(1-\lambda)\mu-2\lambda}{\mu^2}\right)\right) \\
& = {\rm sgn}\left(\frac{(1-\lambda)\mu^2-((1-\lambda)\mu-2\lambda)2\mu}{\mu^4}\right) \\
& = {\rm sgn}((1-\lambda)\mu-2((1-\lambda)\mu-2\lambda))= {\rm sgn}(-(1-\lambda)\mu+4\lambda),
\end{align*}

\noindent which implies that
\begin{equation}\label{sign:g:derivative}
{\rm sgn}(\dot{\bar{g}}(\mu))[>\;{\rm or}\; =] 0\iff \mu[<\;{\rm or}\; =]\frac{4\lambda}{1-\lambda}.
\end{equation}

\noindent From \eqref{function:g:property} and \eqref{sign:g:derivative} we get that
\begin{equation} \label{function:g:values}
0<\bar{g}(\mu)<\bar{g}\left(\frac{4\lambda}{1-\lambda}\right)=\frac{4\lambda^2v_{\max}^{2}}{16\frac{\lambda^2}{(1-\lambda)^2}ML}=\frac{(1-\lambda)^2v_{\max}^{2}}{4ML},\forall \mu\in\left(\frac{2\lambda}{1-\lambda},\frac{4\lambda}{1-\lambda}\right)\cup\left(\frac{4\lambda}{1-\lambda},\infty\right),
\end{equation}

\noindent which implies that \eqref{dmax:interval:caseBi} and \eqref{dmax:interval:caseBii} are well defined. Also, by considering the function
\begin{equation}\label{function:h:double:prime}
h''(d_{\max}):=\frac{(1-\lambda)v_{\max}+\sqrt{(1-\lambda)^2v_{\max}^{2}-4MLd_{\max}}}{2ML},
\end{equation}

\noindent we deduce that
\begin{equation}\label{h:double:prime:vs:constraint}
\mu<\frac{4\lambda}{1-\lambda}\Rightarrow \frac{\mu}{2\lambda v_{\max}}d_{\max}<\frac{(1-\lambda)v_{\max}}{2ML}\le h''(d_{\max}),\forall d_{\max}\in\left(0,\frac{(1-\lambda)^2 v_{\max}^{2}}{4ML}\right],
\end{equation}

\noindent and thus, that \eqref{deltat:interval:caseBi} in Case II is well defined. Finally, in order to also show the latter for the interval \eqref{deltat:interval:caseBi} in Case III, we note that due to \eqref{function:barg} and \eqref{function:h:double:prime} it holds that
\begin{align*}
h''(\bar{g}(\mu)) & = \frac{(1-\lambda)v_{\max}+\sqrt{(1-\lambda)^2v_{\max}^{2}-4ML\frac{2(\lambda(1-\lambda)\mu-2\lambda^2)v_{\max}^{2}}{\mu^2ML}}}{2ML} \\
& = \frac{(1-\lambda)v_{\max}+\sqrt{v_{\max}^{2}\left((1-\lambda)^2-8\frac{\lambda}{\mu}(1-\lambda)+16\left(\frac{\lambda}{\mu}\right)^2\right)}}{2ML} \\
& = \frac{(1-\lambda)v_{\max}+v_{\max}\sqrt{\left((1-\lambda)-\frac{4\lambda}{\mu}\right)^2}}{2ML}=\frac{\left[(1-\lambda)+\left|(1-\lambda)-\frac{4\lambda}{\mu}\right|\right]v_{\max}}{2ML}.
\end{align*}

\noindent Hence, by taking into account that
\begin{equation}\label{sign:condition:due:to:g}
(1-\lambda)-\frac{4\lambda}{\mu}\ge 0\iff \mu\ge\frac{4\lambda}{1-\lambda}
\end{equation}

\noindent and that in Case III it holds $\mu\ge\frac{4\lambda}{1-\lambda}$, we get that
\begin{equation}\label{h:double:prime:at:mu}
h''(\bar{g}(\mu))=\frac{2\left((1-\lambda)+\frac{2\lambda}{\mu}\right)v_{\max}}{2ML}=\frac{((1-\lambda)\mu-2\lambda)v_{\max}}{\mu ML}.
\end{equation}

\noindent Furthermore, we have that
\begin{equation}\label{constraint:at:mu}
\frac{\mu}{2\lambda v_{\max}}{\bar{g}(\mu)}=\frac{\mu}{2\lambda v_{\max}}\frac{2(\lambda(1-\lambda)\mu-2\lambda^2)v_{\max}^{2}}{\mu^2ML}=\frac{((1-\lambda)\mu-2\lambda)v_{\max}}{\mu ML}.
\end{equation}

\noindent Thus, from \eqref{function:g:values}, \eqref{h:double:prime:at:mu}, \eqref{constraint:at:mu} and the fact that $h''(\cdot)$ is decreasing, it follows that
\begin{equation}\label{h:double:prime:vs:constraint:caseC}
 h''(d_{\max})\ge\frac{\mu}{2\lambda v_{\max}}d_{\max},\forall d_{\max}\in(0,\bar{g}(\mu)],
\end{equation}

\noindent which in conjunction with \eqref{function:barg}, implies that \eqref{deltat:interval:caseBi} is well defined. \\

\noindent \textbf{Proof of the fact that for all Cases I, II and III \eqref{dmax:bound2:plan}-\eqref{mu:slope:requirement:plan} are satisfied in the proof of Proposition \ref{proposition:property:P}.}

\noindent \textbf{Case I: $0\le\mu\le\frac{2\lambda}{1-\lambda}$.}

\noindent By defining
\begin{equation}\label{function:h:prime}
h'(d_{\max}):=\frac{(1-\lambda)v_{\max}-\sqrt{(1-\lambda)^2v_{\max}^{2}-4MLd_{\max}}}{2ML},
\end{equation}

\noindent we obtain that
\begin{equation*}
\dot{h}'(d_{\max})=\frac{1}{\sqrt{(1-\lambda)^2 v_{\max}^{2}-4MLd_{\max}}}.
\end{equation*}

\noindent Hence,
\begin{equation}\label{function:h:prime:derivatine:monotonicity}
\dot{h}'(\cdot)\; \textup{is positive and strictly increasing for}\; 0\le d_{\max}<\frac{(1-\lambda)^2 v_{\max}^{2}}{4ML}
\end{equation}

\noindent and furthermore
\begin{equation*}
\dot{h}'(0)=\frac{1}{(1-\lambda)v_{\max}};\quad h'(0)=0.
\end{equation*}

\noindent The latter in conjunction with \eqref{function:h:prime:derivatine:monotonicity} implies that
\begin{equation}\label{function:h:prime:property}
h'(d_{\max})>\frac{1}{(1-\lambda)v_{\max}}d_{\max}\ge\frac{1}{M+v_{\max}}d_{\max},\forall d_{\max}\in\left(0,\frac{(1-\lambda)^2 v_{\max}^{2}}{4ML}\right].
\end{equation}

\noindent Also, we have that for $0\le\mu\le\frac{2\lambda}{1-\lambda}$ it holds
\begin{equation} \label{mu:slope:caseA}
\frac{\mu}{2\lambda v_{\max}} d_{\max}\le \frac{\frac{2\lambda}{1-\lambda}}{2\lambda v_{\max}} d_{\max}=\frac{1}{(1-\lambda)v_{\max}} d_{\max}.
\end{equation}

\noindent Thus, it follows from  \eqref{dmax:interval:caseA}, \eqref{deltat:interval:caseA}, \eqref{function:h:prime}, \eqref{mu:slope:caseA} and \eqref{function:h:prime:property} that \eqref{dmax:bound2:plan}-\eqref{mu:slope:requirement:plan} are fulfilled (see also Fig. 5).

\begin{figure}[H]
\begin{center}
\begin{tikzpicture}[
axis/.style={very thick, ->, >=stealth'}]
\begin{scope}[rotate=-90, style={very thick}]
        \draw[domain= -2:2, color=orange, bottom color=green, top color=green] plot (\x,4-\x*\x);
\end{scope}
\draw[axis] (-0.5,-2) -- (6,-2) node [right] {$d_{\max}$};
\draw[axis] (0,-2.5) -- (0,3) node [above] {$\delta t$};
\draw[dashed, very thick, color=blue] (0,-2) -- (6,-0.5) node [right, color=black] {$\frac{1}{(1-\lambda)v_{\max}}d_{\max}$};
\draw[dashed, very thick, color=blue] (0,-2) -- (6,1)  node [right, color=black] {$\frac{2}{(1-\lambda)v_{\max}}d_{\max}$};
\draw[dashed, very thick, color=red] (0,-2) -- (6,-1.5) node [right, color=black] {$\frac{1}{M+v_{\max}}d_{\max}$};
\draw[axis, color=red] (0,-2) -- (6,-1) node [right, color=black] {$\frac{\mu}{2\lambda v_{\max}}d_{\max}$}; 
\draw[dashed] (4,-2) -- (4,0);
\fill[black] (4,0) circle (2pt);
\fill[black] (4,-2) circle (2pt) node[below right] {$\frac{(1-\lambda)^2 v_{\max}^{2}}{4ML}$};
\coordinate [label=left:$\frac{(1-\lambda)v_{max}}{ML}$] (A) at (0,2);
\end{tikzpicture}
\caption{\small \sl \textbf{Case I.} Feasible $d_{\max}-\delta t$ region for $0<\mu\le\frac{2\lambda}{1-\lambda}$}
\end{center}
\end{figure}

\noindent \textbf{Case II: $\frac{2\lambda}{1-\lambda}<\mu<\frac{4\lambda}{1-\lambda}$.}

\noindent From \eqref{function:barg} and \eqref{function:h:prime} we obtain that
\begin{align*}
h'(\bar{g}(\mu)) & = \frac{(1-\lambda)v_{\max}-\sqrt{(1-\lambda)^2v_{\max}^{2}-4ML\frac{2(\lambda(1-\lambda)\mu-2\lambda^2)v_{\max}^{2}}{\mu^2ML}}}{2ML} \\
& = \frac{(1-\lambda)v_{\max}-\sqrt{v_{\max}^{2}\left((1-\lambda)^2-8\frac{\lambda}{\mu}(1-\lambda)+16\left(\frac{\lambda}{\mu}\right)^2\right)}}{2ML} \\
& = \frac{(1-\lambda)v_{\max}-v_{\max}\sqrt{\left((1-\lambda)-\frac{4\lambda}{\mu}\right)^2}}{2ML}=\frac{\left[(1-\lambda)-\left|(1-\lambda)-\frac{4\lambda}{\mu}\right|\right]v_{\max}}{2ML}.
\end{align*}

\noindent Hence, by taking into account \eqref{sign:condition:due:to:g} and that $\frac{2\lambda}{1-\lambda}<\mu<\frac{4\lambda}{1-\lambda}$, we get that
\begin{equation}\label{h:prime:at:mu}
h'(\bar{g}(\mu))=\frac{2\left((1-\lambda)-\frac{2\lambda}{\mu}\right)v_{\max}}{2ML}=\frac{((1-\lambda)\mu-2\lambda\mu)v_{\max}}{\mu ML}.
\end{equation}

\noindent Thus, by the fact that $\frac{2\lambda}{1-\lambda}<\mu<\frac{4\lambda}{1-\lambda}$ it follows from \eqref{function:g:values}, \eqref{constraint:at:mu}, \eqref{function:h:prime:derivatine:monotonicity} and \eqref{h:prime:at:mu} that
\begin{equation}\label{h:prime:vs:constraint}
h'(d_{\max})\le\frac{\mu}{2\lambda v_{\max}}d_{\max},\forall d_{\max}\in (0,\bar{g}(\mu)]
\end{equation}

\noindent and
\begin{equation}\label{h:prime:vs:constraint2}
h'(d_{\max})>\frac{\mu}{2\lambda v_{\max}}d_{\max},\forall d_{\max}\in \left(\bar{g}(\mu),\frac{(1-\lambda)^2 v_{\max}^{2}}{4ML}\right].
\end{equation}

\noindent It also holds that
\begin{equation}\label{mu:slope:casesBC}
\mu>\frac{2\lambda}{1-\lambda}\Rightarrow \frac{\mu}{2\lambda v_{\max}}d_{\max}>\frac{1}{(1-\lambda)v_{\max}}d_{\max}>\frac{1}{M+v_{\max}}d_{\max},\forall d_{\max}>0.
\end{equation}

\noindent Hence, it follows from \eqref{dmax:interval:caseBi}-\eqref{deltat:interval:caseA}, \eqref{function:barg}, \eqref{h:prime:vs:constraint}, \eqref{h:prime:vs:constraint2} and \eqref{mu:slope:casesBC} that \eqref{dmax:bound2:plan}-\eqref{mu:slope:requirement:plan} are fulfilled as well (see also Fig. 6).

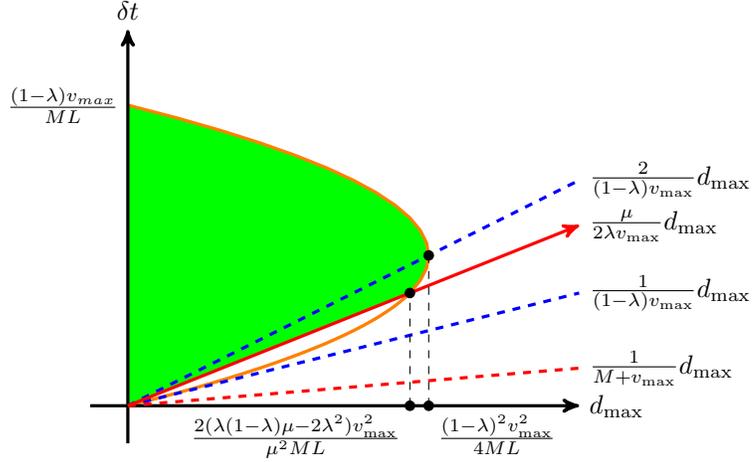
\begin{figure}[H]
\begin{center}
\begin{tikzpicture}[
axis/.style={very thick, ->, >=stealth'}]
\begin{scope}[rotate=-90, style={very thick}]
        \draw[domain= -2:2, color=orange, bottom color=green, top color=green] plot (\x,4-\x*\x);
\end{scope}
\begin{scope}[rotate=-90, style={very thick}]
       \draw[domain= 0.5:2, color=orange, bottom color=white, top color=white] plot (\x,4-\x*\x);
\end{scope}
\draw[axis] (-0.5,-2) -- (6,-2) node [right] {$d_{\max}$};
\draw[axis] (0,-2.5) -- (0,3) node [above] {$\delta t$};
\draw[dashed, very thick, color=blue] (0,-2) -- (6,-0.5) node [right, color=black] {$\frac{1}{(1-\lambda)v_{\max}}d_{\max}$};
\draw[dashed, very thick, color=blue] (0,-2) -- (6,1)  node [right, color=black] {$\frac{2}{(1-\lambda)v_{\max}}d_{\max}$};
\draw[dashed, very thick, color=red] (0,-2) -- (6,-1.5) node [right, color=black] {$\frac{1}{M+v_{\max}}d_{\max}$};
\draw[axis, color=red] (0,-2) -- (6,0.4) node [right, color=black] {$\frac{\mu}{2\lambda v_{\max}}d_{\max}$};
\draw[dashed] (3.75,-2) -- (3.75,-0.5);
\draw[dashed] (4,-2) -- (4,0);
\fill[black] (3.75,-0.5) circle (2pt);
\fill[black] (4,0) circle (2pt);
\fill[black] (3.75,-2) circle (2pt) node[below left] {$\frac{2(\lambda(1-\lambda)\mu-2\lambda^2)v_{\max}^{2}}{\mu^2ML}$};
\fill[black] (4,-2) circle (2pt) node[below right] {$\frac{(1-\lambda)^2 v_{\max}^{2}}{4ML}$};
\coordinate [label=left:$\frac{(1-\lambda)v_{max}}{ML}$] (A) at (0,2);
\end{tikzpicture}
\caption{\small \sl \textbf{Case II.} Feasible $d_{\max}-\delta t$ region for $\frac{2\lambda}{1-\lambda}<\mu<\frac{4\lambda}{1-\lambda}$}
\end{center}
\end{figure}

\noindent \textbf{Case III $\mu\ge\frac{4\lambda}{1-\lambda}$.}

In this case, notice that for $\mu=\frac{4\lambda}{1-\lambda}$ and $d_{\max}=\frac{(1-\lambda)^2 v_{\max}^{2}}{4ML}$ we have $\frac{\mu}{2\lambda v_{\max}}d_{\max}=h'(d_{\max})$, and thus, we deduce from \eqref{function:h:prime:derivatine:monotonicity} that
\begin{equation}\label{h:prime:vs:constraint:caseC}
\mu\ge\frac{4\lambda}{1-\lambda}\Rightarrow \frac{\mu}{2\lambda v_{\max}}d_{\max}\ge h'(d_{\max}),\forall d_{\max}\in\left(0,\frac{(1-\lambda)^2 v_{\max}^{2}}{4ML}\right].
\end{equation}

\noindent Hence, it follows from \eqref{dmax:interval:caseBi}, \eqref{deltat:interval:caseBi}, \eqref{function:barg}, \eqref{function:g:values} and \eqref{h:prime:vs:constraint:caseC} that \eqref{dmax:bound2:plan}-\eqref{mu:slope:requirement:plan} hold (see also Fig. 7).

\begin{figure}[H]
\begin{center}
\begin{tikzpicture}[
axis/.style={very thick, ->, >=stealth'}]
\begin{scope}[rotate=-90, style={very thick}]
        \draw[domain= -2:2, color=orange, bottom color=green, top color=green] plot (\x,4-\x*\x);
\end{scope}
\begin{scope}[rotate=-90, style={very thick}]
       \draw[domain= -0.5:2, color=orange, bottom color=white, top color=white] plot (\x,4-\x*\x);
\end{scope}
\draw[axis] (-0.5,-2) -- (6,-2) node [right] {$d_{\max}$};
\draw[axis] (0,-2.5) -- (0,3) node [above] {$\delta t$};
\draw[dashed, very thick, color=blue] (0,-2) -- (6,-0.5) node [right, color=black] {$\frac{1}{(1-\lambda)v_{\max}}d_{\max}$};
\draw[dashed, very thick, color=blue] (0,-2) -- (6,1)  node [right, color=black] {$\frac{2}{(1-\lambda)v_{\max}}d_{\max}$};
\draw[dashed, very thick, color=red] (0,-2) -- (6,-1.5) node [right, color=black] {$\frac{1}{M+v_{\max}}d_{\max}$};
\draw[axis, color=red] (0,-2) -- (6,2) node [right, color=black] {$\frac{\mu}{2\lambda v_{\max}}d_{\max}$};
\draw[dashed] (3.75,-2) -- (3.75,0.5);
\draw[dashed] (4,-2) -- (4,0);
\fill[black] (3.75,0.5) circle (2pt);
\fill[black] (4,0) circle (2pt);
\fill[black] (3.75,-2) circle (2pt) node[below left] {$\frac{2(\lambda(1-\lambda)\mu-2\lambda^2)v_{\max}^{2}}{\mu^2ML}$};
\fill[black] (4,-2) circle (2pt) node[below right] {$\frac{(1-\lambda)^2 v_{\max}^{2}}{4ML}$};
\coordinate [label=left:$\frac{(1-\lambda)v_{max}}{ML}$] (A) at (0,2);
\end{tikzpicture}
\caption{\small \sl \textbf{Case III.} Feasible $d_{\max}-\delta t$ region for $\mu\ge\frac{4\lambda}{1-\lambda}$}
\end{center}
\end{figure}

\section{Acknowledgements}

This work was supported by the H2020 ERC Starting Grant BUCOPHSYS and the Swedish Research Council (VR).

\end{document}